\documentclass[11pt]{article}
\usepackage{fullpage}
\usepackage{url}
\usepackage[colorlinks=true,linkcolor=Emerald,citecolor=RoyalBlue]{hyperref}
\usepackage{amsmath,nccmath, amsfonts,amsthm,multirow,mdframed,amssymb}
\usepackage{thmtools,thm-restate}

\usepackage{times}
\usepackage{tikz}
\usepackage{paralist}
\usepackage{graphicx}
\usepackage{floatpag}
\usepackage{float}
\DeclareMathOperator{\supp}{supp}
\usepackage{bbold}
\usepackage{enumitem}
\usepackage{subcaption} 
\usepackage{soul}
\usepackage{calc}
\usepackage{upgreek}
\usepackage{mathtools}
\usepackage{color}
\usepackage[dvipsnames]{xcolor}
\usepackage{xspace}
\usepackage{tcolorbox}
\usepackage{stackengine}
\usepackage{tocloft}
\usepackage{physics}
\usepackage{cleveref}
\usepackage{cleveref-forward}
\AddToHook{cmd/appendix/before}{\crefalias{section}{appendix}}

\usepackage{algorithm}
\usepackage{algorithmic}

\usepackage{authblk}

\newcommand{\symbolfootnote}[2]{%
    \begingroup
    \renewcommand{\thefootnote}{\fnsymbol{footnote}}%
    \footnotetext[#1]{#2}%
    \endgroup
}

\stackMath
\allowdisplaybreaks

\usepackage{xcolor}
\newtheorem{theorem}{Theorem}
\newtheorem*{theorem*}{Theorem}
\newtheorem{lemma}[theorem]{Lemma}
\newtheorem{corollary}[theorem]{Corollary}
\newtheorem{definition}{Definition}[section]

\newtheorem{fact}{Fact}[section]

\newtheorem{remark}{Remark}[section]

\newcommand{\set}[1]{\left\{#1\right\}}

\newcommand{\poly}{\operatorname{poly}}
\newcommand{\pr}[1]{\mathrm{Pr}\!\left[#1\right]}
\newcommand{\lr}[1]{\left(#1\right)}
\newcommand{\E}[1]{\mathbb{E}\!\left[#1\right]}
\newcommand{\EE}{\mathbb{E}}

\newcommand{\C}{\mathbb{C}}

\newcommand{\R}{\mathbb{R}}
\newcommand{\cA}{\mathcal{A}}

\newcommand{\cD}{\mathcal{D}}

\newcommand{\cF}{\mathcal{F}}
\newcommand{\cG}{\mathcal{G}}
\newcommand{\cH}{\mathcal{H}}

\newcommand{\cP}{\mathcal{P}}

\newcommand{\cT}{\mathcal{T}}
\newcommand{\cU}{\mathcal{U}}

\newcommand{\sgn}{\operatorname{sgn}}

\newcommand{\arctanh}{\operatorname{arctanh}}
\newcommand{\STAB}{\mathrm{STAB}}
\newcommand{\StabCone}{\mathrm{StabCone}}
\newcommand{\wh}[1]{\widehat{#1}}
\newcommand{\wt}[1]{\widetilde{#1}}
\newcommand{\ad}{\operatorname{ad}}

\newcommand{\Sep}{\operatorname{Sep}}

\title{
When quantum thermal states look classical%
\thanks{We use ``when'' in the sense of \emph{imaginary} time.}}

\author[1,$\dagger$]{Harald Putterman}
\author[2,$\dagger$]{Alexander Zlokapa}
\author[1,3]{Jordan Cotler}

\affil[1]{%
Department of Physics,
Harvard University,
Cambridge, MA 02138, USA}

\affil[2]{%
Center for Theoretical Physics---a Leinweber Institute,
MIT}

\affil[3]{%
Harvard Quantum Initiative,
60 Oxford St.,
Cambridge, MA 02138, USA}

\date{}

\begin{document}

\maketitle

\symbolfootnote{2}{These authors contributed equally.}

\thispagestyle{empty}
\begin{abstract}
At high temperature, quantum Gibbs states retain several classical features of the maximally mixed state: the absence of entanglement, the absence of magic, analyticity of the partition function, and efficient classical estimation of thermal observables. We prove new and sharp bounds showing that these features persist down to finite temperatures independent of system size, but fail at distinct inverse-temperature scales, forming a hierarchy of classical-to-quantum transitions. Our bounds apply to long-range Pauli Hamiltonians of locality $k$ and local strength $s$, meaning that the total absolute strength of all terms acting on any one qubit is at most $s$. This class includes power-law interactions with summable tails. Our main results are as follows, all of which are tight up to constants.
\begin{itemize}[leftmargin=1em]
    \item Despite long-range interactions, the separability transition still occurs at a constant temperature $\beta_{\rm sep} = \Theta(1/(sk))$, which is tight even among commuting Hamiltonians. This resolves an open question of~\cite{rouze2025efficient}; we also make the original result~\cite{bakshi2024high} tight for bounded-degree Hamiltonians.
    \item At the same scale, we give a polynomial-time classical algorithm that prepares the Gibbs state as a mixture of pure product stabilizer states. This holds below temperatures at which prior quasipolynomial-time classical algorithms~\cite{bakshi2024high} and quantum Glauber dynamics~\cite{rouze2026optimal,bergamaschi2026fast} succeed.
    \item For Hamiltonians $\epsilon$-close to commuting, the Gibbs state can be represented as a mixture of stabilizer states up to $\beta_{\rm stab} = \Theta(\log(1/\epsilon)/(sk))$, parametrically below the separability scale.
    \item The infinite-temperature phase transition is colder and occurs at $\beta_{\rm phase}=\Theta(1/\max\{s, \sqrt{sk}\})$. For geometrically local Hamiltonians with bounded interaction range, this implies exponential decay of correlations, closing an open question of \cite{harrow2020classical}.
    \item In this same regime, we give polynomial-time classical algorithms for estimating $\log Z(\beta)$ and local thermal expectations, ruling out the proposed superpolynomial advantage for this task in long-range Pauli systems~\cite{sanchez2025high}.
\end{itemize}
Our proofs leverage new cluster expansions to evaluate the partition function of the Gibbs state, of the post-selected Gibbs state after measuring qubits, and of Gibbs-like quantities obtained by moving to the interaction picture. These produce polynomial-time algorithms via a randomized approach that importance samples polymers in the cluster expansion.
\end{abstract}

\newpage
\tableofcontents
\newpage

\section{Introduction}
\label{sec:intro}

At infinite temperature, a quantum system is described by the maximally mixed state.  In this limit there is little distinction between quantum and classical statistical mechanics.  For instance, the infinite temperature state is unentangled, it is a mixture of stabilizer states, and the usual computational questions about thermal observables are trivial.  One may therefore ask in what ways a quantum system continues to be classical at finite temperature.  A striking feature of high-temperature quantum systems is that many notions of classicality persist uncompromised down to finite temperatures.  As the system is cooled further, different forms of classicality fail at sharply defined, but generally distinct, transition scales.  In this work, we establish a hierarchy of such classical-to-quantum transitions, both sharpening and generalizing bounds on transitions from previous work as well as introducing new transitions.

For a Hamiltonian $H$, the finite-temperature Gibbs state is
$\rho_\beta = e^{-\beta H}/\Tr(e^{-\beta H})$, where $\beta = 1/T$ is the inverse temperature.  For small $\beta$ (equivalently, large $T$), the Gibbs state is evidently close to the maximally mixed state.  There are several different ways of characterizing classicality, leading to different kinds of classical-to-quantum transitions.  A natural characterization is for a state to be unentangled, or \textit{separable}, namely a classical mixture of product states.  The work of~\cite{bakshi2024high} showed that at sufficiently high temperatures, $\rho_\beta$ is completely unentangled, while at a finite temperature independent of the system size the state can sharply transition to being entangled.  However, that work did not tightly characterize the transition temperature in terms of the interaction strength and locality of the Hamiltonian, and its results applied only to local Hamiltonians, not to the physically salient setting of long-range interactions. Both of these shortcomings left open the possibility of quantum advantage at high temperature: quantum Gibbs samplers were shown to mix in polynomial time at temperatures lower than those achieved by known polynomial-time classical algorithms \cite{bakshi2024high,rouze2026optimal,bergamaschi2026fast}. In the case of long-range interactions, it also remained conceptually unclear whether a sudden death of entanglement occurred; indeed, \cite{rouze2025efficient,rouze2026optimal,bergamaschi2026fast} explicitly asked if long-range systems have a death of entanglement transition at constant temperature, motivating the conjecture that such Hamiltonians admitted superpolynomial quantum advantage~\cite{sanchez2025high}.

We resolve this question and identify tight bounds on the separability temperature given long-range interactions. We also give a polynomial-time classical algorithm that prepares the thermal state as a mixture of pure product stabilizer states all the way to the separability transition temperature (up to constants), ruling out superpolynomial quantum advantage. Our results hold at asymptotically lower temperatures than known quantum Gibbs samplers~\cite{rouze2026optimal,bergamaschi2026fast}, at lower temperatures than the quasipolynomial-time classical algorithm of~\cite{bakshi2024high}, and even when interactions are sufficiently long-range for Lieb-Robinson bounds to no longer appear quasilocal~\cite{kuwahara2020strictly}.

The separable Gibbs states above are classical in a resource-theoretic sense: because there is no entanglement, the Gibbs state can be written as a probability distribution over product states, which are classically succinct to represent. A Gibbs state may similarly be considered classical if it can be expressed as a mixture of \emph{stabilizer} states. These states lack magic and therefore cannot supply the resource needed to promote Clifford operations to universal quantum computation.  Nonetheless, stabilizer states may be highly entangled, and indeed commuting Pauli Hamiltonians have stabilizer Gibbs states at every temperature.  We show that for Hamiltonians which are $\epsilon$-close to commuting, the corresponding classical-to-quantum transition occurs at a parametrically colder scale than the separability transition. The proof of this result departs conceptually from~\cite{bakshi2024high}, which produces states that are both product and stabilizer. We rework the methods of~\cite{bakshi2024high} in the interaction picture and adopt a new pinning strategy that preserves stabilizerness but not separability.

Finally, we consider more \emph{physical} grounds for referring to a Gibbs state as classical. The thermodynamic notion of an infinite-temperature phase has both physical implications (decay of correlations) and computational implications (estimating thermal expectations). This phase is defined in terms of the partition function $Z(z) = \Tr(e^{-zH})$, which defines the free energy and is related to thermal expectations. Lee-Yang theory identifies phase transitions with complex zeros of $Z(z)$ approaching the real axis.  Thus a zero-free region around $z=0$ means that the system remains in the analytic phase connected to infinite temperature, allowing thermodynamic data to be accessed by convergent high-temperature expansions.  We determine the temperature of the first phase transition when cooling from infinite temperature. Notably, this temperature is asymptotically \emph{colder} than the separability transition.

Although defined in terms of complex temperatures, the high-temperature phase is physically interpretable. Throughout the entire phase, we show that for geometrically local Hamiltonians with bounded interaction range, correlations between any pair of observables decay exponentially. This resolves an open question of~\cite{harrow2020classical}, which asked if the result holds in dimensions above 1D. Indeed, a computational property morally similar to correlation decay holds even in the absence of geometric locality or bounded interaction range. We show that a related truncation leads to a polynomial-time classical algorithm for estimating the partition function and local thermal expectations. This improves on the quasipolynomial-time classical algorithms of~\cite{sanchez2025high} for Pauli Hamiltonians with long-range interactions and reaches colder temperatures.

To state our results quantitatively, we need a Hamiltonian class in which the temperature scales remain meaningful as the system size grows.  We use the standard thermodynamic normalization, in which the Hamiltonian has bounded energy density, so that $\norm{H}_{\rm op} = \Theta(n)$ for an $n$-qubit system.  Previous high-temperature results for quantum Gibbs states often assume bounded degree, where each qubit participates in only a bounded number of local terms~\cite{bakshi2024high,harrow2020classical}.  This is a useful mathematical model, but it excludes many natural systems.  Power-law and exponentially decaying interactions generally couple every site to every other site, even when the total interaction strength seen by any fixed site remains bounded.  Since the questions we ask are exact questions about separability, stabilizerness, and complex zeros, it is not enough to truncate these long-range tails and appeal to a bounded-degree approximation.  Indeed, in the context of high-temperature Gibbs states, it was previously unclear whether constant-temperature classicality transitions survive in genuinely long-range systems~\cite{rouze2025efficient,sanchez2025high,rouze2026optimal,bergamaschi2026fast}.

We therefore work with Pauli Hamiltonians of bounded local interaction strength.  Let
\begin{align}
H = \sum_{a \in \mathcal{A}} c_a P_a \,,
\end{align}
where $\mathcal{A}$ indexes the interaction terms, $c_a \in \mathbb{R}$, and each $P_a$ is a distinct Pauli string.  We assume that
\begin{align}
|\supp(P_a)|\le k \quad \text{for all } a\in\mathcal A\,, \qquad \max_{x\in[n]}\sum_{a\,:\,x\,\in\,\supp(P_a)} |c_a| \le s\,, \qquad \max_{a\in\cA} |c_a| \leq 1\,.
\end{align}
We call such an $H$ an $(s,k)$-long-range Pauli Hamiltonian. This condition permits a qubit to interact with arbitrarily many other qubits, but requires the sum of the absolute strengths of all terms involving that qubit to be bounded by $s$.  For example, for pair interactions on a $D$-dimensional lattice whose strength decays as $(1 + |i-j|)^{-\alpha}$, the parameter $s$ remains bounded whenever $\alpha > D$.  Thus the model includes the long-range tails present in many physical Hamiltonians, rather than only their bounded-degree truncations.  In particular, this includes interactions longer-ranged than those for which Lieb-Robinson bounds retain an essentially local form~\cite{kuwahara2020strictly}. \Cref{fig:main_figure} depicts the thresholds discussed above in terms of parameters $s$ and $k$, and in terms of $\epsilon$ that measures how close a Hamiltonian is to commuting; we postpone the formal definition of $\epsilon$ to \Cref{def:epscomm}. We proceed to describe our results more formally in the remaining of the introduction, before providing more detailed proof sketches in \Cref{sec:overview} and the full proofs in subsequent sections.

In the rest of \Cref{sec:intro} we provide a high-level overview of our results and discuss open questions. In \Cref{sec:overview} we present our results more formally and outline our proofs. Our main results for the death of entanglement, death of magic, classical estimation of thermal expectations, and classical preparation of separable Gibbs states are proven in \Cref{sec:entanglement,sec:magic,sec:expect,sec:prep}. Finally, we show these bounds are tight in \Cref{sec:upper} and show a more general separability result for nonlocal Hamiltonians in \Cref{sec:nonlocal}.

\begin{figure}
    \centering
    \includegraphics[width=1\linewidth]{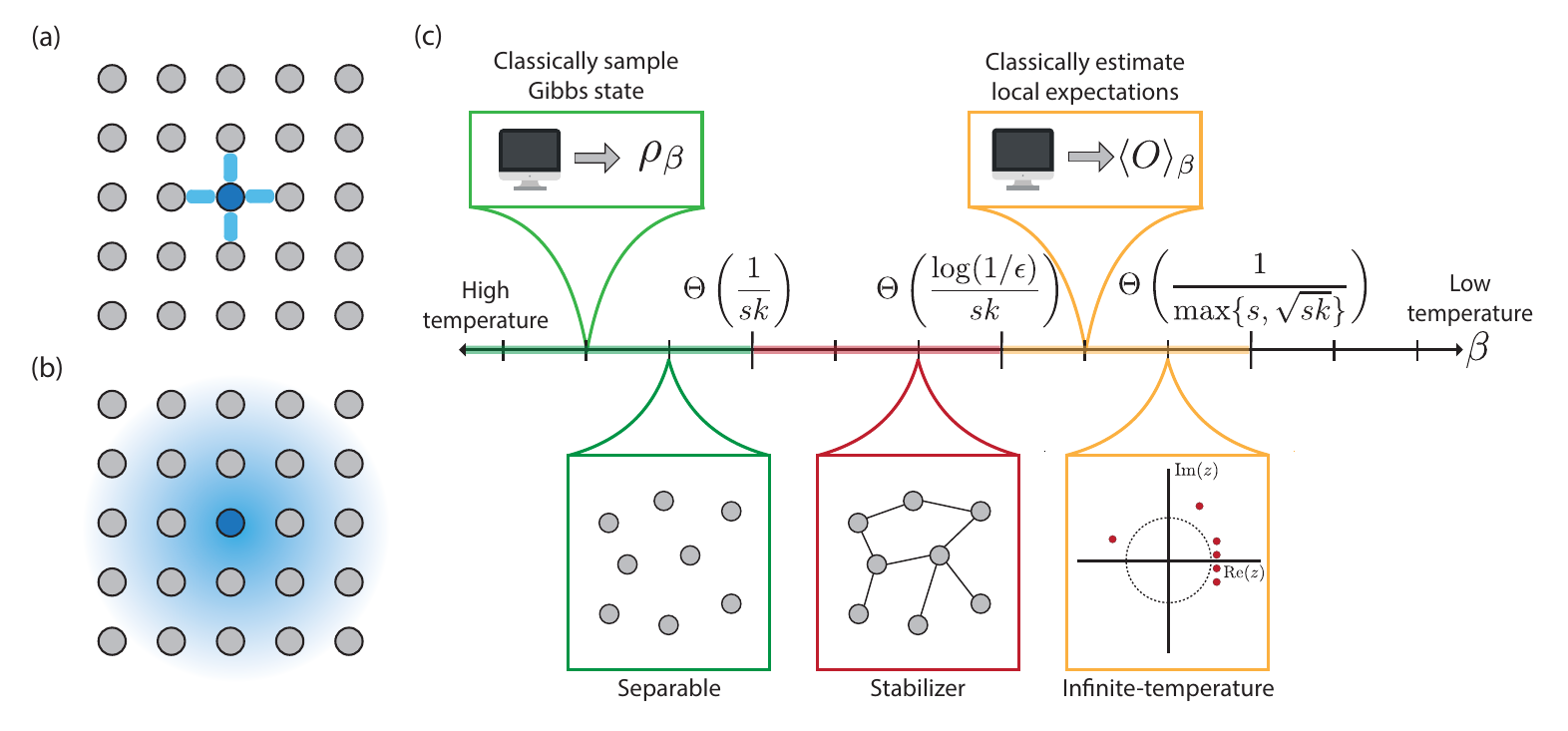}
    \caption{Properties of long-range interacting systems. (a) Depiction of a local Hamiltonian with bounded degree, where each site only interacts with a finite number of sites. (b) Schematic of a long-range interacting system where each site can interact with all of the sites in the lattice. (c) Inverse temperatures $\beta$ at which all long-range Hamiltonians are guaranteed to be physically or computationally classical.}
    \label{fig:main_figure}
\end{figure}

\subsection{Classical representations of quantum Gibbs states}

We describe our results showing separability and stabilizerness. A separable state can be said to be ``classical'' in a representational sense: it is described by a probabilistic mixture of product states, each of which possesses a succinct classical description. Given access to this representation, one can efficiently estimate thermal expectations or even sample measurement outcomes of the Gibbs state. Separable states are only one possible family of states that provide such a classical representation of the Gibbs state. Below, we also study \emph{stabilizer} states and show a ``death of magic'' transition. This naturally occurs at a colder temperature than the death of entanglement: for example, commuting Pauli Hamiltonians always have stabilizer thermal states, although their eigenstates may be entangled (e.g., a stabilizer code state). We will tightly bound the temperature of stabilizerness in terms of how commuting a Hamiltonian is, leading to a distinct transition temperature from separability.

\begin{theorem*}[Classical representations of Gibbs states, informal]
    Let $\cH(s,k;\epsilon)$ denote the family of $(s,k)$-long-range Pauli Hamiltonians that are $\epsilon$-close to commuting (\Cref{def:epscomm}).
    \begin{itemize}
        \item The death of entanglement transition occurs at $\beta_{\rm sep} = \Theta(1/sk)$, independent of $\epsilon$.
        \item The death of magic transition occurs at $\beta_{\rm stab} = \Theta(\log(1/\epsilon)/sk)$.
    \end{itemize}
\end{theorem*}
\noindent
As in~\cite{bakshi2024high}, these transition temperatures are guaranteed for all Hamiltonians in the family $\cH$; e.g., there exists a constant $c > 0$ such that for \emph{all} $(s,k)$-long-range Pauli Hamiltonians, the Gibbs state is separable for all $\beta \leq c/(sk)$. Unlike~\cite{bakshi2024high}, our thresholds are also \emph{tight}: there exists another constant $C > c$ such that for some $H \in \cH$, the Gibbs state is entangled for all $\beta \geq C/(sk)$ for all sufficiently large $s, k$. Our results for $\beta_{\rm stab}$ are similarly tight up to constants for asymptotically large $s, k, 1/\epsilon$. Examples of models that are $\epsilon$-close to commuting include an Ising model with a transverse field of strength $\epsilon$, or a toric code with local Pauli perturbations of strength $\epsilon$.

Compared to the proof of~\cite{bakshi2024high}, for which the bounded-degree assumption is essential, our separability result requires a different strategy for iteratively pinning qubits to produce product states. Our modification also leads to a tight threshold for the death of entanglement for the Hamiltonian family considered in the original separability work (which does not admit long-range interactions); further modifications also lead to a separability threshold for \emph{nonlocal} Hamiltonians as well.

\begin{remark}[Low-intersection Hamiltonians]
    \cite{bakshi2024high} considered a different family of \emph{low-intersection} Hamiltonians defined in terms of a dual degree $\mathfrak d$ and locality $k$.\footnote{The dual degree $\mathfrak d$ is the maximum degree of the graph defined with vertex set $\cA$ and edges between $a, b\in\cA$ such that $\supp(P_a) \cap \supp(P_b) \neq \emptyset$.} In these parameters, we show in \Cref{thm:sep}(b) the tight result $\beta_{\rm sep} = \Theta(1/\mathfrak d)$, which improves upon the result of \cite{bakshi2024high} that $\rho_\beta$ is separable for all $\beta \lesssim 1/\mathfrak d k$.
\end{remark}

\begin{remark}[Nonlocal Hamiltonians]
    We show a tight separability transition for $s$-strength long-range Pauli Hamiltonians that allow nonlocal terms with sufficiently small strength. In \Cref{thm:sep}(c), we show that $\rho_\beta$ is separable at constant temperature if $\sum_{a\ni x}|c_a|\exp[\Theta(|\supp(P_a)|)]$ is bounded. Conversely, for any $\sum_{a\ni x}|c_a|\poly(|\supp(P_a)|)$, the Gibbs state can be entangled at all constant temperatures.
\end{remark}
\noindent
To establish a death of magic temperature that is different from the death of entanglement, we require a conceptually different proof. The technique of~\cite{bakshi2024high} produced separable states that were \emph{also} stabilizer states, whereas we now wish to avoid showing separability. We move to the interaction picture to isolate the part of the Hamiltonian that is far from commuting; we then pin entire noncommuting terms instead of individual qubits.

We refer to the death of entanglement and death of magic as ``sharp'' transitions because the Gibbs state suddenly goes from having an exact classical representation (above some temperature) to not having such a representation (below that temperature). However, these are not \emph{phase} transitions in the rigorous sense provided by thermodynamics, and for instance cannot be detected by local parameters. We turn to thermodynamic phase transitions next.

\subsection{Infinite-temperature phase and decay of correlations}

In the limit of large system size, a phase transition in $\beta$ is defined as a non-analyticity of the free energy density $-\lim_{n\to\infty}\frac{1}{n\beta}\log Z(\beta)$, where $Z(\beta) = \Tr(e^{-\beta H})$ denotes the partition function. At first glance, such a non-analyticity seems impossible: since $Z(\beta)$ is a sum of positive numbers $e^{-\beta E_i}$ for eigenvalues $E_i$ of the Hamiltonian, the quantity seems analytic. A phase transition is defined by the theory of Lee-Yang or Fisher zeros by considering the values of complex $\beta$ for which $Z(\beta) = 0$. If these complex zeros approach the real-$\beta$ axis in the thermodynamic limit $n\to\infty$, a non-analyticity appears (\Cref{fig:zf}).

Conversely, if there are no complex zeros for any $|\beta| \leq \beta_*$, then the thermal state is said to be in the \emph{infinite-temperature phase} for all real $0 \leq \beta < \beta_*$. This defines another notion of ``classical'': for any $\beta < \beta_*$, one can smoothly interpolate from $\beta H$ to $\beta H'$ by passing through $\beta = 0$ without running into any non-analyticities, even if $H$ is classical and $H'$ is quantum.

\begin{theorem*}[Infinite-temperature phase, informal]
    Let $\cH(s,k)$ denote the family of $(s,k)$-long-range Pauli Hamiltonians. Then
    \begin{align}
        \beta_{\rm phase} = \sup_{\beta \geq 0} \left\{\beta \in \R \,:\, \Tr(e^{-z H}) \neq 0 \text{ for all } |z| < \beta \text{ and every }H \in \cH(s,k)\right\} = \Theta\lr{\frac{1}{\max\{s,\sqrt{sk}\}}}
    \end{align}
    and there exist $H \in \cH(s,k)$ with a phase transition at $\beta_{\rm phase}$ (up to a constant).
\end{theorem*}
\noindent
Notably, the infinite-temperature phase extends to asymptotically \emph{colder} temperatures than separability. Prior work (under the additional assumption of bounded degree) showed bounds with locality dependence $\beta_{\rm phase} \gtrsim 1/k$~\cite{harrow2020classical,mann2021efficient,yao2022polynomial,mann2024algorithmic,zlokapa2026syk}. Given that we showed $\beta_{\rm sep} = \Theta(1/sk)$ above, it would a priori be possible that the onset of entanglement coincides with the thermodynamic phase transition. Our result shows this is false. Our lower bound on $\beta_{\rm phase}$, similarly to prior work, uses a cluster expansion and the Kotecky-Preiss criterion~\cite{kotecky1986cluster}. However, since our Hamiltonian family is decomposed into Pauli terms, we can use a more careful strategy to count polymers since traceless terms do not contribute to the partition function in the Taylor expansion of $Z(\beta)$ around $\beta=0$. This lower bound is matched by an upper bound given in \Cref{sec:upper}, where we provide Hamiltonians that exhibit a thermodynamic phase transition at $\beta_{\rm phase}$.

The infinite-temperature phase identified above also carries physical meaning. For geometrically local Hamiltonians equipped with a distance $d$, the zero-free disk is characterized by an exponential decay of correlations between any pair of observables. In particular, these observables may be arbitrarily placed on the lattice; this resolves an open question of~\cite{harrow2020classical}, which required the observables to be separated at least logarithmically in system size for lattices beyond 1D. While prior work studied correlation decay at high temperature~\cite{kliesch2014locality,frohlich2015some,nguyen2024high,capel2025decay}, our result tightly relates correlation decay to the optimal zero-free disk (up to a constant).

\begin{theorem*}[Exponential decay of correlations]
    For any geometrically $k$-local Pauli Hamiltonian with interaction range $R$ and interaction strength $s$, every pair of observables $A, B$ supported on disjoint regions of the lattice satisfies for all $\beta = O(1/\max\{s,\sqrt{sk}\})$ that
    \begin{align}
        \abs{\langle AB \rangle_\beta - \langle A \rangle_\beta \langle B \rangle_\beta} \leq C \norm{A} \norm{B} \exp[-\frac{C' d(\supp(A),\supp(B))}{R}]
    \end{align}
    for constants $C, C' > 0$ that depend on the locality of the observables.
\end{theorem*}

The proof follows almost immediately from the cluster expansion we use to show the zero-free disk. The cluster expansion controls truncation of a polymer representation of the partition function; for thermal expectations, this truncation permits one to only consider a small neighborhood around the observable. The neighborhood is small enough to imply exponential correlation decay. It also gives rise to efficient classical algorithms, which we proceed to describe next---notably, these algorithms do not require any notion of geometric locality or bounded interaction range.

\subsection{Classically easy quantum Gibbs states}

The infinite-temperature phase is closely related to \emph{computational} questions. One expects on physical grounds that the Gibbs state in the infinite-temperature phase can be efficiently prepared by system-bath dynamics. More surprisingly, the work of \cite{harrow2020classical} realized that the absence of a physical phase transition yields efficient \emph{classical} algorithms for quantum thermal states. They give \emph{quasi}polynomial-time algorithms that estimate the log partition function, which were later improved to polynomial-time algorithms, albeit assuming bounded degree and at asymptotically warmer temperatures than $1/\beta_{\rm phase}$~\cite{mann2021efficient,yao2022polynomial}. We give polynomial-time classical algorithms for both estimating the log partition function and thermal expectations of local Pauli observables for all $\beta \lesssim \beta_{\rm phase}$.

\begin{figure}[t!]
    \centering
    \includegraphics[width=0.75\linewidth]{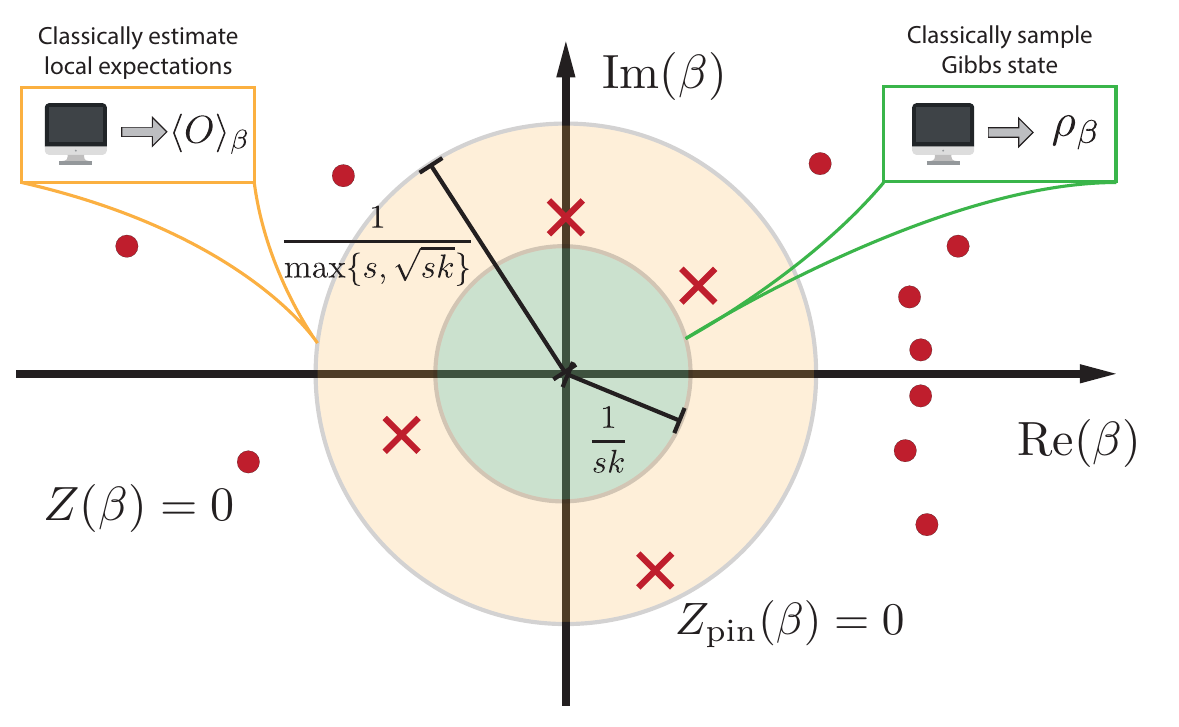}
    \caption{Implications of zero-freeness for phase transitions and classical algorithms. Red dots indicate where the partition function evaluates to zero, $Z(\beta)=\Tr(e^{-\beta H})=0$. A phase transition is defined by Lee-Yang theory to be an extensive density of complex zeros of $Z(\beta)$ approaching the real axis at some $\beta_*$, causing $\lim_{n\to\infty}\frac{1}{n}\log Z(\beta_*)$ to become non-analytic. The partition function of every long-range Pauli Hamiltonian has a zero-free disk of radius $\Theta(1/\max\{s,\sqrt{sk}\})$, ensuring that the system is in the same phase as infinite temperature ($\beta=0$) and that polynomial-time classical algorithms can efficiently estimate thermal expectations. Red crosses indicate where pinned partition functions $Z_{\rm pin}(\beta)$ vanish; these take the form $\Tr_{[n]\setminus Y} \bra{y}e^{-\beta H} \ket{y} = 0$ for $y\in\{0,1\}^Y$ and $Y \subseteq [n]$. The zero-free disk of the pinned partition function occurs at a strictly smaller radius $\Theta(\log(k)/sk)$; this approximately coincides with the separability transition temperature $\beta_{\rm sep}=\Theta(1/sk)$, in which the Gibbs state can be efficiently prepared classically.}
    \label{fig:zf}
\end{figure}

\begin{theorem*}[Classical algorithms for the quantum infinite-temperature phase, informal]
    Let $H$ be an $(s,k)$-long-range Pauli Hamiltonian, let $O$ be a bounded-norm $k'$-local observable, and for some universal constant $c_{\rm q} > 0$, let $0\leq\beta\leq c_{\rm q}/\max\{s, \sqrt{s\max\{k,k'\}}\}$. There are randomized classical algorithms that, given $\epsilon,\delta\in(0,1)$, estimate $\log Z(\beta)$ and $\Tr(O\rho_\beta)$ to additive error $\epsilon$ with probability at least $1-\delta$ in time polynomial in $n,|\cA|,1/\epsilon, k, \log(1/\delta)$, and the number of Pauli strings in the Pauli decomposition of $O$.
\end{theorem*}
\noindent
Our approach for obtaining a polynomial-time algorithm differs from prior approaches in the zero-freeness literature~\cite{mann2021efficient,yao2022polynomial} and instead more closely resembles the recently used randomized approaches related to diagrammatic Monte Carlo that were rigorously analyzed for cumulant expansions~\cite{chen2025convergence}.

Because one cannot in general ``prepare'' a Gibbs state classically, the above result for estimating thermal expectations is a reasonable computational sense in which quantum Gibbs states can be considered ``classical''. However, if the Gibbs state has a classical representation, such as a mixture of product states, one can ask about classical preparability of the quantum state. We show that throughout the separable phase $\beta \lesssim \beta_{\rm sep}$, classical preparation is possible in polynomial time. (One cannot expect an analogous result of efficient classical algorithms for $\beta \lesssim \beta_{\rm stab}$, as classical Hamiltonians are commuting but can be hard to sample from at low temperature.)

\begin{theorem*}[Classical algorithms for separable quantum Gibbs states, informal]
    There exists a universal constant $c_{\rm sep}>0$ such that for every $(s,k)$-long-range Pauli Hamiltonian $H$, every $0\leq\beta\leq c_{\rm sep} \beta_{\rm sep} = O(1/sk)$, and every $\epsilon\in(0,1)$, there is a randomized polynomial-time classical algorithm that outputs a pure product stabilizer state $\Phi$ satisfying $\frac12\norm{\EE\,\Phi-\rho_\beta}_1\leq\epsilon$ with runtime polynomial in $n,|\cA|, k$, and $1/\epsilon$.
\end{theorem*}
\noindent
Unlike \cite{bakshi2024high}, our classical algorithm remains polynomial-time at lower temperatures than quantum algorithms for bounded-degree Hamiltonians in~\cite{rouze2025efficient,rouze2026optimal,bergamaschi2026fast}. Moreover, our result is tight: at lower temperatures, entanglement prevents a classical mixture of product states from being close to the Gibbs state (\Cref{thm:sep-upper}). In comparison, under the additional assumption of bounded degree, \cite{bakshi2024high} showed polynomial-time preparability for $\beta \lesssim 1/s^2k^3$; in their original parameters, for $\beta \lesssim 1/\mathfrak d^2 k$.

Just as our classical algorithms for computing thermal expectations relied heavily on zero-freeness, this result requires $\beta_{\rm phase} > \beta_{\rm sep}$. Indeed, our proofs require a stronger notion of ``pinned'' zero-freeness, which applies to the partition function of the state defined by measuring some qubits of the Gibbs state in the computational basis and post-selecting on an arbitrary measurement outcome (\Cref{fig:zf}). Interestingly, this reaches slightly colder temperatures than the separability transition (\Cref{lem:pin-zf,thm:pin-upper}).

\begin{remark}[Pinned zero-freeness]
    The zero-free radius of the pinned partition function $\Tr(\bra{x}e^{-\beta H}\ket{x})$ for $x \in \{0,1\}^m$ and $m \leq n$ is tightly given by $\Theta(\log(k)/(sk))$.
\end{remark}

\subsection{Related work and open questions}\label{sec:open}

Our results suggest that for many well-defined notions of being ``classical'', either physical or computational, the Gibbs state has a transition temperature marking the change from classical to quantum. While all the characterizations we give are tight up to constants, our work points towards a more universal picture of classical transitions in quantum systems. We highlight several open directions and discuss how our work provides useful techniques for addressing each one.

\paragraph{Quantum advantage for sampling measurement outcomes.} Besides estimating the log partition function and local thermal expectations, another reasonable computational task is to sample from the distribution of measurement outcomes in the computational basis, $p(x) = \bra{x}\rho_\beta\ket{x}$ for $x \in \{0,1\}^n$. Our results already give a polynomial-time classical algorithm for this task for $\beta \lesssim 1/sk$, where the Gibbs state is separable and can be classically prepared. In the absence of separability, our techniques for showing ``pinned'' zero-freeness are sufficient to imply polynomial-time algorithms for sampling from the measurement distribution; however, although pinned zero-freeness holds for $\beta$ up to $\log(k)/sk$, extending the polynomial-time result beyond $\beta \sim 1/sk$ requires an additional variance bound. If a Gibbs state can be prepared efficiently on a quantum computer throughout the infinite-temperature phase (as suggested by physics arguments), then sampling measurement outcomes may be quantumly easy and classically hard for, e.g.~$1/sk \lesssim \beta \lesssim 1/s$. Showing this separation may be related to the IQP results of~\cite{bergamaschi2024quantum,rajakumar2024gibbs}, but these constructions have too poor $k$ dependence to resolve the issue.

\paragraph{Larger class of physical systems.} Our work addressed long-range interactions with bounded strength on local systems. Here, ``strength'' was defined as the $\ell_1$ norm of Pauli Hamiltonian terms acting on a qubit. One can also study Hamiltonian terms that are not Pauli operators~\cite{ramkumar2025high} and different notions of strength \cite{hastings2006spectral, tran2021lieb, kim2025thermal}. Many natural systems act locally on \emph{fermions} and have bounded energy density but diverging $\ell_1$ strength. Two such examples are electronic structure Hamiltonians, which have a $1/r$ Coulomb potential, and the Sachdev-Ye-Kitaev (SYK) model, which has diverging $\ell_1$ strength but bounded $\ell_2$ strength. Recent work shows that (somewhat more unwieldy) cluster expansions yield zero-freeness for disordered all-to-all models with unbounded $\ell_1$ strength~\cite{zlokapa2026rigorous}, but the resulting classical algorithms remain quasipolynomial-time and do not extend to all temperatures predicted by non-rigorous physics computations~\cite{zlokapa2026syk}. Polynomial-time algorithms are known for weakly interacting fermionic systems through a careful cumulant expansion~\cite{chen2025convergence}. We expect that combining such techniques with those of the present work can illuminate if all-to-all Pauli models have a sudden death of entanglement at constant temperature, classically preparable Gibbs states, or polynomial-time classical algorithms for estimating thermal expectations. We also note related questions and techniques in the bosonic setting~\cite{tong2024locally,tong2025long}.

\paragraph{Structural properties of Gibbs states.} The algorithms we study here for quantum Gibbs states do not resemble the ``natural'' mixing dynamics of system-bath dynamics. Such dynamics are expected to mix quickly throughout the infinite-temperature phase, but few techniques exist to rigorously control quantum mixing times. High-temperature fast-mixing results for quantum Gibbs samplers have thus far only been shown at separable temperatures where our ``pinned'' zero-freeness condition holds~\cite{rouze2025efficient,rouze2026optimal,bakshi2025dobrushin,bergamaschi2026fast}. This suggests that pinned zero-freeness may be a useful tool for understanding fast mixing and static properties of the Gibbs state, possibly including quantum generalizations of the Dobrushin-Shlosman conditions~\cite{dobrushin1985completely,dobrushin1987completely}. Because existing techniques typically relate quantum fast mixing to structural properties of the Gibbs state, such as approximate Markov properties or conditional mutual information decay~\cite{kuwahara2020clustering,bluhm2025strong,chen2025locally,bakshi2025dobrushin,kuwahara2025clustering,rosa2026static}, pinned zero-freeness may also provide a new tool for further understanding such properties. If pinning plays such a central role in mixing, it is also possible that, contrary to physics arguments of~\cite{harrow2020classical}, quantum generalizations of Glauber dynamics might \emph{not} mix quickly in the entire infinite-temperature phase but instead only for $\beta \lesssim \log(k)/sk$. Indeed, at colder temperatures, we show in \Cref{sec:upper} that pinning can create non-analyticities at such temperatures.

\paragraph{Tight separable and stabilizer transitions.} Our results give asymptotically tight bounds on the transition temperatures for separability, stabilizerness, largest zero-free disk, etc. However, we do not identify the optimal \emph{constants} in each of these bounds. This is not always important: e.g.~we show that the thermodynamic phase transition occurs at asymptotically colder temperatures than separability. The death of entanglement and death of magic transitions occur at asymptotically different temperatures only when the Hamiltonian is close to commuting. Generically, however, both occur at $\beta = \Theta(1/sk)$. It is open whether the constants are different; we conjecture that $\beta_{\rm sep} < \beta_{\rm stab}$.

\paragraph{Open quantum dynamics beyond Lieb-Robinson bounds.} Lieb-Robinson bounds quantify the speed of information propagation in a many-body quantum system \cite{lieb1972finite, nachtergaele2006lieb}.  A long series of works, starting with \cite{hastings2006spectral}, has generalized Lieb-Robinson bounds to long-range interacting Hamiltonians. For power-law interacting systems with interactions decaying like $1/r^\alpha$ on a $D$-dimensional lattice, these results have culminated in showing a linear light cone for $\alpha > 2D+1$ and faster operator spreading for $\alpha \leq 2D+1$~\cite{foss2015nearly, else2020improved, kuwahara2020strictly, chen2019finite, tran2021lieb}. In contrast, our results still hold for all $\alpha > D$, since our argument is based on cluster expansions rather than Lieb-Robinson bounds. Since proofs of fast mixing of quantum Gibbs samplers typically rely on Lieb-Robinson bounds, our classical results hold for a larger class of systems~\cite{rouze2025efficient,rouze2026optimal,bakshi2025dobrushin,bergamaschi2025quantum,vsmid2025rapid,tong2025fast}. A notable exception is~\cite{bergamaschi2026fast}, which uses cluster expansions to control mixing times but obtains exponentially worse dependence in $k$ than our classical algorithms. We expect that further developing tools to control Gibbs sampling algorithms via cluster expansions will eventually show that natural quantum algorithms mix quickly throughout temperatures for which it is easy to classically prepare the Gibbs state. We also note that our arguments are reminiscent of zero-freeness conditions used to show rapid mixing of classical Glauber dynamics~\cite{shao2021contraction,regts2023absence,chen2024spectral}.

\section{Technical overview}\label{sec:overview}

\subsection{Preliminaries}

We first fix notation for the Hamiltonian families and transition temperatures used throughout the paper.  As mentioned above, Hamiltonians in this work are finite-dimensional Pauli Hamiltonians of the form
\begin{align}
    H=\sum_{a\in\cA}c_aP_a \,,
\end{align}
where each $P_a$ is a distinct non-identity Pauli string and $c_a\in\R$.  We write
$\supp(P_a)\subseteq[n]$ for the set of qubits on which $P_a$ acts nontrivially,
and
\begin{align}
    \rho_\beta(H)=\frac{e^{-\beta H}}{\Tr(e^{-\beta H})},
    \qquad
    Z(z)=\Tr(e^{-zH})
\end{align}
for the Gibbs state and the complex-temperature partition function. Throughout, $\Tr$ denotes ordinary trace, $\Tr_S$ denotes the partial trace over a set $S$ of qubits, $\tr$ is normalized by Hilbert space dimension, and we introduce a ``pinned'' trace $\tr_y$ later in \Cref{sec:prep}.

\begin{definition}[Separable and stabilizer states]
\label{def:sepstab}
For $n$ qubits, let
\begin{align}
    \Sep_n
    =
    \operatorname{conv}\left\{
        \rho_1\otimes\cdots\otimes\rho_n:
        \rho_i\in\mathbb C^{2\times 2},\ \rho_i\succeq0,\ \Tr\rho_i=1
    \right\}
\end{align}
denote the set of fully separable states.  Let $\STAB_n$ denote the convex hull
of pure $n$-qubit stabilizer states. Recall that a pure state $\ket{\psi}\in(\mathbb C^2)^{\otimes n}$ is a stabilizer state if there exists an abelian subgroup $S\leq \mathcal P_n$ for the Pauli group $\mathcal P_n$ such that
\begin{align}
    -I\notin S,\qquad |S|=2^n,\qquad
    g\ket{\psi}=\ket{\psi}\quad\forall g\in S .
\end{align}
\end{definition}

\begin{definition}[Transition temperatures]\label{def:transition-radii}
Let $\cF=\{\cF_n\}_{n\geq1}$ be a family of Hamiltonians, where $\cF_n$ is a subset of the set of all Hermitian matrices in $(\C^2)^{\otimes n}$. Define
\begin{align}
    \beta_{\rm sep}(\cF)
    &=
    \sup\left\{
        r\geq0:
        \rho_\beta(H)\in\Sep_n
        \text{ for all } n,\ H\in\cF_n,\ 0\leq\beta\leq r
    \right\},\\
    \beta_{\rm stab}(\cF)
    &=
    \sup\left\{
        r\geq0:
        \rho_\beta(H)\in\STAB_n
        \text{ for all } n,\ H\in\cF_n,\ 0\leq\beta\leq r
    \right\},\\
    \beta_{\rm phase}(\cF)
    &=
    \sup\left\{
        r\geq0:
        Z(z)\neq0
        \text{ for all } n,\ H\in\cF_n,\ |z|<r
    \right\}.
\end{align}
\end{definition}
\noindent
The main Hamiltonian family we consider consists of local Pauli terms with bounded $\ell_1$ interaction strength on each qubit.

\begin{definition}[Long-range Pauli Hamiltonians]\label{def:lr-pauli}
For $s,k>0$, $H=\sum_{a\in\cA}c_aP_a$ is an $(s,k)$-long-range Pauli Hamiltonian if it satisfies
\begin{align}
    |\supp(P_a)|\leq k
    \qquad\text{for all }a\in\cA,
    \qquad
    \max_{x\in[n]}\sum_{a:x\in\supp(P_a)}|c_a|\leq s .
\end{align}
\end{definition}
\noindent
A bounded-degree $k$-local Hamiltonian with uniformly bounded coefficients is
an $(O(d),k)$-long-range Hamiltonian, where $d$ is the maximum number of terms
incident to a qubit.  However, \Cref{def:lr-pauli} also allows all-to-all
interactions, provided the total incident interaction strength remains bounded. We also compare to the low-intersection parameterization used in
\cite{bakshi2024high}.

\begin{definition}[Low-intersection Pauli Hamiltonians]\label{def:low-intersection}
For $\mathfrak d,k>0$, $H=\sum_{a\in\cA}c_aP_a$ is a $(\mathfrak d,k)$-low-intersection Pauli Hamiltonian if it satisfies
\begin{align}
    |c_a|\leq1,
    \qquad
    |\supp(P_a)|\leq k,
\end{align}
and if the dual interaction graph has maximum degree at most $\mathfrak d$. The dual interaction graph has vertex set $\cA$ and an edge between $a,b\in\cA$ whenever $\supp(P_a)\cap\supp(P_b)\neq\emptyset$.
\end{definition}
\noindent
Every $(\mathfrak d,k)$-low-intersection Hamiltonian is an
$(\mathfrak d+1,k)$-long-range Hamiltonian.  If a Hamiltonian has ordinary degree $d$, then its dual degree satisfies $d-1\leq\mathfrak d\leq k(d-1)$. We will also use a nonlocal family in which terms of arbitrary weight are allowed \cite{hastings2006spectral} but their coefficients decay exponentially in the support size \cite{frohlich2015some}.

\begin{definition}[Exponentially decaying nonlocal Pauli Hamiltonians]
\label{def:exp-nonlocal}
For $s,\gamma>0$, $H=\sum_{a\in\cA}c_aP_a$ is a $(s,\gamma)$-nonlocal Pauli
Hamiltonian if it satisfies
\begin{align}
    \max_{x\in[n]}\sum_{a:x\in\supp(P_a)} |c_a|e^{\gamma|\supp(P_a)|} \leq s .
\end{align}
\end{definition}
\noindent

\subsection{Death of entanglement}

In \cite{bakshi2024high}, it is shown that the Gibbs states of $(\mathfrak{d}, k)$-low-intersection Pauli Hamiltonians are separable at $\beta \lesssim 1/\mathfrak{d} k$. Crucially, these Hamiltonians have bounded interaction degree, which is used to establish a convergent series for the quantity
\begin{align}\label{eq:main-prop}
    e^{-\beta H/2} e^{\beta(H- H_{(a^*)})/2}, \qquad H_{(a^*)} =\sum_{a\in\mathcal{A}: \operatorname{supp}(P_a) \cap \operatorname{supp}(P_{a^*})\neq \emptyset} c_a P_a.
\end{align}
This quantity is referred to as a ``propagator'', since it adds support onto Hamiltonian terms that touch the term $P_{a^*}$ \cite{alhambra2023quantum}. In the presence of long-range interactions, it is unclear on physical grounds whether one should expect separability at a constant temperature; operationally, the expansion given in \cite{bakshi2024high} for the propagator fails. Despite these failures, we show that \cref{eq:main-prop} can be expanded in a convergent series. We begin by expanding the propagator in the form $\sum_{t=0}^\infty \beta^t f_t /t!$, where each $f_t$ is a sum of products of $t$ Hamiltonian terms, i.e.,
\begin{align}
    f_t = \sum_{(b_1,\dots,b_t) \in Q_t} \mu_{(b_1,\dots,b_t)}P_{b_1}\cdots P_{b_t}
\end{align}
for some set $Q_t$. Much like a cluster expansion, $Q_t$ captures a sequence of $t$ terms in the Hamiltonian whose supports overlap. Moreover, the support of $P_{b_1}$ must intersect the support of $P_{a^*}$. To construct these clusters, \cite{bakshi2024high} identified the recurrence relation $f_{t+1} = -[H, f_t] - f_t H_{(a^*)}$. Writing out the commutator $-[H, f_t]$ as
\begin{align}
    -\sum_a c_a [ P_a, f_t] 
    &= - \sum_{(b_1,\dots,b_t)\in Q_t} \sum_a c_a \mu_{(b_1,\dots,b_t)} [P_a , P_{b_1}\cdots P_{b_t}],
\end{align}
we see that convergence can be controlled by bounding the norms of the above terms. In particular, one can extract factors of $\sup_x \sum_{a \ni x} |c_a|$ to obtain dependence on the interaction strength $s$ rather than the degree of the interaction.
Once the propagator is appropriately expanded, we sample tuples $(b,E)$ of coefficients and Pauli operators such that 
\begin{align}
    e^{-\frac{\beta}{2} H} e^{\frac{\beta}{2} (H- H_{(a^*)})} = \mathbb{E}[I + b E]
\end{align}
where $|b| \lesssim (\beta k s)^t$ and $t$ is the number of terms of the Hamiltonian which make up $E$. To show separability, one hopes to eventually find a distribution over coefficients $c_j$ and Pauli operators $X_j$ such that the Gibbs state can be represented as
\begin{align}\label{eq:main-sep-decomp}
    e^{-\beta H} = \mathbb{E}_{c,X}\left[\bigotimes_j (I + c_j X_j)\right].
\end{align}
If the coefficients are sufficiently small ($|c_j| \leq 1$), this gives an explicit representation of the Gibbs state as a mixture of product states. To reach this decomposition, the proof of~\cite{bakshi2024high} iteratively constructs decompositions of the form
\begin{align}\label{eq:main-invar}
    e^{-\beta H} = e^{-\frac{\beta}{2} H^{(S)}}\mathbb{E}_{c,X}\left[\bigotimes_j (I + c_j X_j)\right]e^{-\frac{\beta}{2} H^{(S)}},
\end{align}
where $H^{(S)}$ is the Hamiltonian restricted to terms supported only on qubits in a set $S \subseteq [n]$. Each iteration consists of choosing some Hamiltonian term $P_{a^*}$ in $H^{(S)}$, and then removing the Hamiltonian terms touching $a^*$. This introduces the propagator $e^{-\beta H^{(S)}/2} e^{\beta (H^{(S)} - H^{(S)}_{(a^*)})/2}$. Since the propagator is representable in the form $\mathbb{E}[I + b E]$, we can expand it as such and remove $\supp(P_{a^*})$ from $S$. By updating the distribution over $c$ and $A$, we can also maintain the decomposition \cref{eq:main-invar} with $S$ becoming smaller each iteration, until $S=\emptyset$ and the Gibbs state takes the form \cref{eq:main-sep-decomp}.

To obtain the optimal separability temperature, we improve upon the strategy of~\cite{bakshi2024high} for choosing $a^*$ in each iteration. Consider applying the iteration
\begin{align}
    e^{-\beta H} &= e^{-\frac{\beta}{2} H^{(S\setminus a^*)}}e^{\frac{\beta}{2} H^{(S\setminus a^*)}}e^{-\frac{\beta}{2} H^{(S)}}\mathbb{E}_{c,X}\left[\bigotimes_j (I + c_j X_j)\right]e^{-\frac{\beta}{2} H^{(S)}}e^{\frac{\beta}{2} H^{(S\setminus a^*)}}e^{-\frac{\beta}{2} H^{(S\setminus a^*)}}\nonumber \\
    &= \frac{1}{2} e^{-\frac{\beta}{2} H^{(S\setminus a^*)}}\mathbb{E}_{c,X,b,E}\left[(I + b_1 E_1)^\dagger \bigotimes_j(I+c_jX_j) (I+b_2 E_2) + h.c.\right]e^{-\frac{\beta}{2} H^{(S\setminus a^*)}}
\end{align}
where we wrote $S\setminus a^*$ as shorthand for $S \setminus \supp(P_{a^*})$. In~\cite{bakshi2024high}, $a^*$ is chosen such that $P_{a^*}$ intersects the support of the $X_j$. We use a more refined strategy that exploits an explicit decomposition of $X_j$ into Hamiltonian terms. By adaptively choosing $a^*$ according to that decomposition (as shown later in \Cref{fig:new_iterative_pinning_example}), we can reduce the number of iterations required to reach $S = \emptyset$ by a factor of $k$. This ultimately improves our separability temperature for low-intersection Hamiltonians to $\beta_{\rm sep}=\Theta(1/\mathfrak d)$ instead of the bound $\beta_{\rm sep} = \Omega(1/\mathfrak d k)$ of~\cite{bakshi2024high}. For the Hamiltonian families defined above, we state this result with explicit constants below, and we show matching entangled Hamiltonians to demonstrate that these temperatures are tight. Note that in the nonlocal setting of \Cref{thm:sep}(c), a single Hamiltonian term may act on an extensive number of sites. The proof sketch provided above no longer suffices; we use a more sophisticated weighting scheme described in \Cref{sec:nonlocal}.

\begin{theorem}[Death of entanglement]
\label{thm:sep}
The death of entanglement temperature is given for each Hamiltonian family as follows.
\begin{enumerate}[label=(\alph*),leftmargin=*]
    \item For $(s,k)$-long-range Pauli Hamiltonians (\Cref{def:lr-pauli}), $\beta_{\rm sep} = \Theta(1/sk)$. Specifically, for every Hamiltonian $H$ the Gibbs state is separable whenever
    \begin{align}
        0\leq \beta\leq \frac{1}{72sk}.
    \end{align}
    Conversely, for every even $k\geq2$ and every $s>0$, there exists an
    $(s,k)$-long-range Pauli Hamiltonian whose Gibbs state is entangled whenever
    \begin{align}
        \beta>
        \frac{2}{s}\arctanh\frac{3}{k+2} = \Omega\lr{\frac{1}{sk}},
    \end{align}
    even if the Hamiltonian is commuting.
    \item For $(\mathfrak d,k)$-low-intersection Pauli Hamiltonians (\Cref{def:low-intersection}), $\beta_{\rm sep} = \Theta(1/\mathfrak d)$. Specifically, for every Hamiltonian $H$ the Gibbs state is separable whenever
    \begin{align}
        0\leq\beta\leq \frac{1}{96\mathfrak d}.
    \end{align}
    Conversely, there is a sequence of commuting $k$-local Pauli Hamiltonians with
    dual degree $\mathfrak d$ whose Gibbs states are entangled for all
    \begin{align}
        \beta>
        \frac{96 \log 2}{\mathfrak d+2} = \Omega\lr{\frac{1}{\mathfrak d}}.
    \end{align}
    \item For every $(s,\gamma)$-nonlocal Pauli Hamiltonian $H$, the Gibbs state is separable whenever
    \begin{align}
        0\leq\beta\leq
        \min\left\{
            \frac{\gamma}{16s},
            \frac{1}{16s}
        \right\}.
    \end{align}
    For every constant $\beta>0$ and non-decreasing $w(n) = e^{o(n)}$, there is a sequence of Pauli Hamiltonians $H_n=\sum_ac_aP_a$ satisfying
    \begin{align}
        \max_{x\in[n]}\sum_{a:x\in\supp(P_a)} |c_a| w(|\supp(P_a)|) = O(1)
    \end{align}
    but whose Gibbs states $\rho_\beta(H_n)$ are entangled at all constant temperatures (i.e., $\beta$ independent of $n$).
\end{enumerate}
\end{theorem}
\noindent
As a physically motivated example, we can apply the above result to the specific case of power-law interacting Hamiltonians on a $D$-dimensional lattice, i.e.~with interaction strengths at most $1/\|x-y\|_2^\alpha$ between lattice sites $x, y$. Our results guarantee separability at some constant temperature when $\alpha > D$, while we show that for every $\alpha \leq D$ there exist Hamiltonians that are entangled at \emph{any} temperature independent of system size.

\subsection{Classical preparability}

Due to the constructive nature of the proof, the separability results allow us to prepare an \emph{unnormalized} separable state proportional to the Gibbs state. Prior work \cite{bakshi2024high} (which applied at higher temperatures $\beta \lesssim 1/\mathfrak d k$) found a polynomial-time algorithm to compute the normalization factor $Z(\beta)$ to convert this into efficient classical state preparation (for $\beta \lesssim 1/\mathfrak d^2 k$). Working in our more general long-range Hamiltonian family, we show that a polynomial-time classical algorithm can in fact cover the entire separable regime (up to constants) given by \Cref{thm:sep}(a).

\begin{theorem}[Classical preparation of separable Gibbs states]
\label{thm:samp}
For every $(s,k)$-long-range Pauli Hamiltonian $H$, every
\begin{align}
    0\leq\beta\leq \frac{1}{4096e\,sk},
\end{align}
and every $\epsilon\in(0,1)$, there is a randomized classical algorithm running
in time
\begin{align}
    \wt O\left(
        n|\cA|k+
        \frac{kn^{11}}{\epsilon^5}\log|\cA|
    \right)
\end{align}
which outputs a pure product stabilizer state
$\Phi=\phi_1\otimes\cdots\otimes\phi_n$ such that $\frac12\left\|\EE\,\Phi - \rho_\beta(H)\right\|_1\leq\epsilon$.
Moreover, there are $(s,k)$-long-range Pauli Hamiltonians such that for any
\begin{align}
    \beta \geq \frac{2}{s}\arctanh\frac{3}{k+2} = \Omega\left(\frac1{sk}\right)
\end{align}
the Gibbs state has constant trace distance from every separable state:
\begin{align}
    \inf_{\sigma\in\Sep_n} \frac12\norm{\rho_{\beta}(H)-\sigma}_1 \geq \frac{1}{2}.
\end{align}
\end{theorem}
\noindent
In terms of parameters $s,k$, the classical algorithm of \cite{bakshi2024high} to prepare Gibbs states only reaches inverse temperatures $\beta\lesssim1/(s^2k^3)$ under the additional assumption of bounded degree. Beyond classical Gibbs samplers, our result also succeeds at colder temperatures than known \emph{quantum} Gibbs samplers. The results of \cite{rouze2026optimal} for local Hamiltonians are given on $D$-dimensional lattices, obtaining fast mixing for $\beta \lesssim 1/2^{\Theta(D)} sk$ assuming bounded degree. In the long-range setting, \cite{rouze2026optimal} also only achieves fast mixing at constant temperature for power laws with $\alpha > 4D+2$, whereas our result remains efficient at constant temperature for any $\alpha > D$. Finally, the quantum algorithm of \cite{bergamaschi2026fast} mixes quickly for $\beta \lesssim 1/(s 2^{\Theta(k)})$.

The proof of our preparation algorithm relies on the same decomposition $\E{I+cA}$ for $|c| \leq 1$ and Pauli $A$ as the separability result above. To compute the partition function and normalize the Gibbs state, we use cluster expansion. In particular, we estimate the partition function at each iteration of the pinning procedure in the separability proof; this is equivalent to the partition function defined after post-selecting on particular measurement outcomes in the computational basis. While a na\"{i}ve analysis gives a quasipolynomial-time algorithm for estimating these quantities via Barvinok's method~\cite{barvinok2014computing,barvinok2016approximating,barvinok2016combinatorics,barvinok2016computing,barvinok2018approximating}, we show how to obtain a polynomial-time algorithm. Due to the similarity of the proof, we postpone this discussion to \Cref{sec:main-expect}. After the pinned partition function is computed, one can efficiently sample from the distribution of measurement outcomes in the computational basis using methods similar to those of~\cite{yin2023polynomial,bakshi2024high}. We compute iterative normalization factors via rejection sampling at each iteration of the separability proof, and then telescope them to obtain the final partition function of the full Gibbs state.

\subsection{Death of magic}

To formalize the death of magic transition, we need to introduce a measure of being close to commuting.  Write a Hamiltonian as
\begin{align}\label{eq:h0v}
    H=H_0+V,
    \qquad
    H_0=\sum_{\sigma\in\cF}u_\sigma P_\sigma,
    \qquad
    V=\sum_{\mu\in\cG}v_\mu Q_\mu ,
\end{align}
where the Paulis $P_\sigma$ commute pairwise, while the Paulis $Q_\mu$ are
arbitrary.  For Pauli strings $W_1,W_2$, define the compatibility relation
\begin{align}\label{eq:compat}
    W_1\sim W_2
    \quad\Longleftrightarrow\quad
    \{W_1,W_2\}=0
    \text{ or there exists }\sigma\in\cF
    \text{ such that }
    \{W_1,P_\sigma\}=\{W_2,P_\sigma\}=0 .
\end{align}
For the decomposition \cref{eq:h0v}, define
\begin{align}\label{eq:wfp}
    w_{\rm free}
    =
    \max_{\mu\in\cG}
    \sum_{\sigma\in\cF:\{P_\sigma,Q_\mu\}=0}|u_\sigma|, \qquad 
    w_{\rm pert}
    =
    \max_{\mu\in\cG}
    \left(
        |v_\mu|
        +
        \sum_{\nu\in\cG\setminus\{\mu\}:\nu\sim\mu}|v_\nu|
    \right).
\end{align}

\begin{definition}[Close to commuting]\label{def:epscomm}
We say that $H$ is $\epsilon$-close to commuting if
\begin{align}
    \inf_{H=H_0+V}
    \frac{w_{\rm pert}}{w_{\rm free}+w_{\rm pert}} \leq \epsilon \in [0,1],
\end{align}
where the infimum is over all decompositions of the form \cref{eq:h0v}, and we set $\epsilon(H)=0$ when $V=0$.
\end{definition}
\noindent
For example, a transverse-field Ising model $H = -\sum_{\langle i,j\rangle}Z_iZ_j-h\sum_iX_i$ and a toric-code Hamiltonian perturbed by local Pauli fields of strength $h$ are both $O(h)$-close to commuting.

\begin{theorem}[Death of magic for nearly commuting Hamiltonians]
\label{thm:stab}
For $(s,k)$-long-range Pauli Hamiltonians that are
$\epsilon$-close to commuting, $\beta_{\rm stab} = \Theta\left(\frac{\log(1/\epsilon)}{sk}\right)$ holds for asymptotically small $\epsilon$ and large $s,k$.
\end{theorem}
\noindent
A key conceptual difficulty in proving this result is that the proof technique of \cite{bakshi2024high} creates separable states that are already stabilizer states, implying death of magic. Obtaining a distinct stabilizer threshold from the separability threshold thus requires a new approach.

To show \Cref{thm:stab}, we move to the interaction picture with respect to the commuting part $H_0$ and no longer pin individual qubits. Instead of $\E{I+cA}$, we decompose the Gibbs state in the form
\begin{align}
    e^{-\beta H_0/2}
    \prod_{j=1}^m(I+\lambda_jX_j)
    e^{-\beta H_0/2},
\end{align}
which results in a stabilizer state if $X_1,\dots,X_m$ are pairwise incompatible in the sense of \cref{eq:compat}, and if
\begin{align}
    |\lambda_j|
    \leq
    \exp[-\beta\sum_{\sigma:\{P_\sigma,X_j\}=0}|u_\sigma|].
\end{align}
Rather than pinning qubits as in the separability proof, we remove perturbing terms $Q_\mu$ in $V$ one at a time. The propagator of \cref{eq:main-prop} gets rewritten in terms of $\cU_S$ that satisfies
\begin{align}
    e^{-\beta(H_0 + V_S)} = e^{-\beta H_0/2} \cU_S(\beta/2, -\beta/2) e^{-\beta H_0/2},
\end{align}
where $V_S$ denotes the terms in $V$ contained in a set $S$ of Hamiltonian terms (rather than qubits). We expand the propagator with a Dyson series to control quantities in terms of nested commutators. The rest of the proof then follows similarly to the separability result. To show the result is tight, we use a simple Hamiltonian of the form $\frac{1}{k}\sum_{i=1}^k Z_i + \epsilon X_1 \cdots X_k$.

\subsection{Infinite-temperature phase, correlation decay, and thermal expectations}
\label{sec:main-expect}
As shown in \Cref{fig:zf}, the first thermodynamic phase transition from the infinite-temperature phase is bounded by the radius of the zero-free disk the partition function around the origin. For long-range Pauli Hamiltonians, we show that this radius is larger than the separability (and preparability) radius by a
factor of at least $\sqrt{k}$.

\begin{theorem}[Infinite-temperature phase]
\label{thm:zf}
For $H=\sum_{a\in\cA}c_aP_a$ be an $(s,k)$-long-range Pauli Hamiltonian, $Z(z) \neq 0$ for all $|z| \leq \beta_{\rm phase}$, where
\begin{align}
    \beta_{\rm phase} = \Theta\lr{\frac{1}{\max\{s,\sqrt{sk}\}}}.
\end{align}
This is tight, i.e., there exist $(s,k)$-long-range Pauli Hamiltonians with a non-analytic free energy density at $\beta_{\rm phase}$ (up to a constant).
\end{theorem}

This result leads to concrete physical and computational properties of the Gibbs state. Throughout the zero-free disk established above, we show that geometrically local Hamiltonians, for which a notion of distance between sites is well-defined, exhibit correlation decay between any pair of observables. Notably, this resolves open question 1(a) of~\cite{harrow2020classical}, which only showed a similar result for 1D or commuting Hamiltonians.

\begin{theorem}[Exponential decay of correlations]\label{thm:decay}
    Let $H=\sum_{a\in\cA} c_a P_a$ be a geometrically local $(s,k)$-long-range Pauli Hamiltonian on a lattice $\Lambda$ with distance $d$. Assume $H$ has interaction range $R$, i.e., for every $a\in \cA$,
    \begin{align}
        \max_{x,y \in \supp(P_a)} d(x,y) \leq R.
    \end{align}
    Let $A, B$ be arbitrary Hermitian observables supported on disjoint sets $X, Y \subset \Lambda$. Then for
    \begin{align}
        0 \leq \beta \leq \frac{1}{128e\,\max\{s, \sqrt{s\max_{a\in\cA} |c_a| \max\{k, |X|+|Y|\}}\}}
    \end{align}
    it holds that
    \begin{align}
        \abs{\langle AB \rangle_\beta - \langle A \rangle_\beta \langle B \rangle_\beta} \leq 48 \cdot 2^{|X|+|Y|} \norm{A} \norm{B} \exp[-\frac{\log 8}{6R}d(X,Y)].
    \end{align}
\end{theorem}

This correlation decay suggests that estimating thermal expectations may be computationally easy by truncating the Hamiltonian to a small neighborhood around the observable. Although that intuition only holds in geometrically local Hamiltonians with bounded interaction range, we show that it in fact extends to all-to-all Hamiltonians with unbounded interaction range. Na\"{i}vely, the zero-free region of \Cref{thm:zf} leads to a quasipolynomial-time algorithm for estimating thermal expectations (via Barvinok's method). We improve this to a polynomial-time algorithm throughout the entire zero-free disk (up to constants).

\begin{theorem}[Polynomial-time estimation of thermal expectations]\label{thm:expect}
    Let
    \begin{align}
        O=\sum_{b\in\cA_O}d_bQ_b,
        \qquad
        B_O=\sum_{b\in\cA_O}|d_b|,
        \qquad
        |\supp(Q_b)|\leq k_O.
    \end{align}
    For every real
    \begin{align}
        0\leq\beta\leq\frac{1}{512e\,\max\{s, \sqrt{s\max_{a\in\cA}|c_a|\max\{k,k_O\}}\}},
    \end{align}
    there is a randomized classical algorithm which estimates
    $\log Z(\beta)$ and $\Tr(O\rho_\beta(H))$ to additive error $\epsilon$ with failure probability at most $\delta$ in time
    \begin{align}
        \wt O\left(
            (|\cA|+|\cA_O|)\max\{k,k_O\}
            +
            \frac{B_O^3+n^3}{\epsilon^3}
            \log(|\cA|+|\cA_O|)
            \log\frac1\delta
        \right),
    \end{align}
    where $\wt O$ suppresses additional logarithmic factors.
\end{theorem}

\noindent
In comparison, \cite{sanchez2025high} analyzed long-range interactions and obtained a quasipolynomial-time algorithm for partition function estimation for $\beta \lesssim 1/k$, causing them to raise the possibility of a superpolynomial quantum advantage for partition function estimation. Our result shows that, at least for Pauli Hamiltonians, a polynomial-time classical algorithm exists for the same long-range interactions and for even colder temperatures than \cite{sanchez2025high}.

Prior polynomial-time algorithms for estimating the partition function \cite{mann2021efficient,yao2022polynomial,mann2024algorithmic} apply only to systems with bounded interaction degree. These works, as well as quasipolynomial-time algorithms~\cite{harrow2020classical}, generally apply to temperatures with locality dependence $\beta \lesssim 1/k$ (or worse). The source of our improved $1/\sqrt k$ scaling is partly unique to the choice of the Pauli basis, as we now explain.
Directly expanding the partition function
\begin{align}
    \Tr(e^{-zH})
    =
    \sum_{m\geq0}
    \frac{(-z)^m}{m!}
    \sum_{a_1,\ldots,a_m}
    \left(\prod_{j=1}^m c_{a_j}\right)
    \Tr(P_{a_1}\cdots P_{a_m}),
\end{align}
one sees that the trace in each term can be factorized over the connected components of the supports of the terms $P_{a_1}, \dots, P_{a_m}$. These form the polymers in the cluster expansion, which we control via the Kotecky-Preiss criterion, similarly to~\cite{mann2021efficient}. Note that if we instead constructed polymers from Trotter steps (which is another standard option, see e.g.~\cite{zlokapa2026syk}), the following argument would not immediately apply.

To construct the polymers, we grow the product from $P_{a_1}\cdots P_{a_j}$ to $P_{a_1}\cdots P_{a_j} P_{a_{j+1}}$ by choosing a new Hamiltonian term with overlapping support on $P_{a_j}$. Na\"{i}vely, there are $k$ choices for a qubit in this overlapping support, since each term is $k$-local. However, the final product
$P_{a_1}\cdots P_{a_m}$ only contributes to the trace if every qubit is acted on trivially overall. Thus, whenever a qubit has a nontrivial Pauli applied to it, a later term must act nontrivially on that same qubit. We can thus confine half of the choices of $P_{a_{j+1}}$ to ``repairing'' a non-identity qubit, reducing the effective combinatorial growth from $(sk)^m$ to
$(\sqrt{sk})^m$.

Once we have a convergent cluster expansion, zero-freeness follows immediately from Kotecky-Preiss. Correlation decay (\Cref{thm:decay}) also follows almost immediately: for geometrically local Hamiltonians, truncating the cluster expansion results in support only on a local neighborhood of polymers around an observable, ensuring that its expectation is independent of Hamiltonian terms further away. The zero-free disk also almost immediately gives a deterministic algorithm for estimating the partition function by truncating the cluster expansion at order $O(\log n/\epsilon)$. However, enumerating all these terms requires quasipolynomial time $n^{O(\log n/\epsilon)}$, as with most generic applications of Barvinok's method. Our polynomial-time algorithm avoids this by constructing a random variable that equals the truncated cluster expansion in expectation; we control its variance to show it succeeds with high probability. We use importance sampling---i.e., if an object with contribution $a$ is selected with probability $p$, then outputting $a/p$ gives an
unbiased contribution---on the polymer lengths, the polymers themselves, and the ways in which the polymers are connected.

We use a similar cluster expansion and sampling procedure to prove the polynomial-time algorithm for Gibbs state preparation given in \Cref{thm:samp}. However, rather than estimating $\Tr(e^{-zH})$, we consider the partition function that normalizes the Gibbs state post-selected on measurement outcomes for individual qubits. This pinned partition function arises naturally in each step of the pinning procedure in the separability proof of \Cref{thm:sep} (similar to~\cite{bakshi2024high}).

\section{Death of entanglement}
\label{sec:entanglement}
We prove in this section parts (a) and (b) of \Cref{thm:sep}; part (c) is shown in \Cref{sec:nonlocal}. We will also only give the lower bounds on $\beta_{\rm sep}$ here; the upper bounds are shown in \Cref{thm:sep-upper} (\Cref{sec:upper}).

We briefly establish some notation. For Hamiltonian term labels $a, b\in \cA$, we use $a \cap b$ as shorthand for relations between $\operatorname{supp}(P_a)$ and $\operatorname{supp}(P_b)$.  We will use $|a|$ to denote the number of elements in $\operatorname{supp}(P_a)$.  We will use the notation that $H^{(Q)}$ is the Hamiltonian restricted to the terms $Q \subseteq \mathcal{A}$.  Additionally $H^{(S)}$ with $S\subseteq \Lambda$ means all terms in the Hamiltonian such that the term is contained in $S$ 
\begin{align}
    H^{(S)} = \sum_{a\in\mathcal{A}: \supp(P_a) \subseteq S} \lambda_a P_a
\end{align}
and $H_{(S)}$ means all terms in the Hamiltonian such that the term touches $S$
\begin{align}
    H_{(S)} = \sum_{a\in\mathcal{A}: \supp(P_a) \cap S \neq \emptyset} \lambda_a P_a
\end{align}
For an interaction label $a\in\mathcal{A}$, write $H_{(a)}$ for the Hamiltonian restricted to terms touching $a$.  We denote the interaction labels which correspond to unpinned terms of the Hamiltonian as $\mathcal{A}^{(S)} = \{a\in\mathcal{A}: a\subseteq S\}$.

\subsection{Separable decomposition}
To show separability (\Cref{def:sepstab}), we will show that the Gibbs state $\rho_\beta$ can be expressed as a positive mixture of $\bigotimes_l (I + c_{i,l} P_{i,l})$ where the $ P_{i,l}$ are disjoint.  We will work with the unnormalized Gibbs state $e^{-\beta H}$ since the normalization only rescales the $p_i$. Making use of the fact that for $|c| \leq 1$, $I+ cP$ is separable, one can conclude that the Gibbs state can be written as a mixture of products of stabilizer states.  This is captured by the following condition for separability.
\begin{lemma}[Separable Pauli factors]
\label{lem:sep-pauli-factor}
    If a positive operator $\rho$ can be expressed as
    \begin{align}
        \rho = \sum_i p_i \bigotimes_l (I + c_{i,l} P_{i,l})
    \end{align}
    where $p_i\geq 0$, $P_{i, l}$ is a Pauli product on block $l$, $|c_{i,l}| \leq 1$, $c_{i,l}\in \mathbb{R}$, and the $l$ are disjoint, then the operator is separable into a mixture of stabilizer states.  
\end{lemma}
\begin{proof}
Clearly, it is a positive combination of terms each of which are PSD.  Now we just need to show that each of the factors can be decomposed into a product of single qubit terms.  We'll focus on showing one $I+cP$ term is separable when $c\neq 0$ ($c=0$ is trivial) which then generalizes to the product that appears in the lemma since the product of disjoint separable operators is separable. To start note that
\begin{align}
    I + cP &= (I - |c|)I + |c| (I + \frac{c}{|c|} P). 
\end{align}
Let $\sigma = \operatorname{sign}(c)$ and expand $I + \frac{c}{|c|} P = I + \sigma P$ with $P = P_1 \otimes \cdots \otimes P_m$ as
\begin{align}
    I + \sigma P 
    &= \frac{1}{2^m} \bigotimes_{r=1}^m \left[(I + P_r) + (I-P_r)\right] + \frac{\sigma}{2^m} \bigotimes_{r=1}^m \left[(I + P_r) - (I-P_r)\right] \nonumber \\
    &= \frac{1}{2^m} \sum_{j \in \{0,1\}^m}\bigotimes_{r=1}^m \left[(I + (-1)^{j_r} P_r) \right] + \frac{\sigma}{2^m} \sum_{j \in \{0,1\}^m}(-1)^{|j|}\bigotimes_{r=1}^m \left[(I + (-1)^{j_r} P_r) \right] \nonumber \\
    &=  \sum_{j \in \{0,1\}^m}(1 +\sigma(-1)^{|j|})\bigotimes_{r=1}^m \frac{1}{2}\left[(I + (-1)^{j_r} P_r) \right]
\end{align}
Since the $P_r$ are individual Pauli operators $\frac{1}{2} \left[I+(-1)^{j_r}P_r\right]$ is a single-qubit operator with trace one.  Furthermore all of the coefficients $1 +\sigma(-1)^{|j|}$ are either 0 or 2 so all the coefficients are nonnegative.  Hence $I + \sigma P$ is separable.  Every term in $\rho$ is a tensor product of disjoint separable operators and is thus separable.  Taking a positive combination of such operators preserves separability, and hence $\rho$ is separable. Explicitly,
\begin{align}
    \rho
    &= \sum_i p_i \bigotimes_l (I_l + c_{i,l} P_{i,l}) \nonumber \\
    &= \sum_i p_i \bigotimes_l \left[(1 - |c_{i,l}|)I_l + |c_{i,l}| \left(I_l + \frac{c_{i,l}}{|c_{i,l}|} P_{i,l}\right)\right] \nonumber \\
    &= \sum_i p_i \bigotimes_l \left[(1 - |c_{i,l}|)I_l + |c_{i,l}| \left(\frac{1}{2^{m_{i,l}}} \sum_{j_l \in \{0,1\}^{m_{i,l}}}(1 + \sigma_{i,l}(-1)^{|j_l|})\bigotimes_{r=1}^{m_{i,l}} \bigl(I + (-1)^{(j_l)_r} P_{i,l,r}\bigr) \right)\right].\label{eq:icp}
\end{align}
Finally, every state appearing in the decomposition is a product of stabilizer states as $\frac{1}{2}\left(I \pm \{X,Y,Z\}\right)$ is a stabilizer state.
\end{proof}

\subsection{Propagator}

Our first step will be to show that the so-called propagator $e^{-\beta H} e^{\beta (H- H_{(a^*)})}$ can be expanded in a convergent power series below a critical temperature.  Then, using the propagator, we will adaptively pin all the sites in the lattice to end up with a separable expression for the Gibbs state.  The key step which needs to be adapted to the case of long-range interactions is the expansion of the propagator.  If one were to na\"{i}vely proceed with the style of argument in \cite{bakshi2024high} it would fail because in the long-range coupled case there is no locality that can be used to control the recursion.  Specifically in \cite{bakshi2024high} they could rely on simple combinatorial properties to show convergence of the series.  This approach would quickly fail because of the all-to-all interactions which would give factors scaling polynomially in $n$.  Despite this apparent failure if the recursion is amortized properly one can still get a convergent series.  By instead using collective bounds of the form
\begin{align}
    \sup_{x\in \Lambda} \sum_{a\ni x} |\lambda_a| \leq s  < \infty
\end{align}
where $s$ is a constant independent of $n$ we can extend the arguments to cover the case of long-range interacting systems.  We will refer to this quantity as the ``on-site'' energy.

\begin{lemma}[Propagator expansion]
\label{lem:sep-propagator-expansion}
    Given an $(s,k)$-long-range Pauli Hamiltonian $H = \sum_a h_a = \sum_a \lambda_a P_a$ and an interaction label of the Hamiltonian $a^*\in\mathcal{A}$, we can expand the propagator $e^{-\beta H} e^{\beta (H- H_{(a^*)})}$ as 
    \begin{align}
        e^{-\beta H} e^{\beta (H- H_{(a^*)})} = \sum_{t=0}^\infty \frac{\beta^t}{t!} f_t(H, H_{(a^*)}),
    \end{align}
    where $f_t(H, H_{(a^*)})$ satisfies a recurrence  $f_{t+1}(H, H_{(a^*)}) = -[H, f_{t}(H, H_{(a^*)})] - f_{t}(H, H_{(a^*)})H_{(a^*)}$ and $f_0(H, H_{(a^*)}) = I$.  Furthermore
    \begin{align}
        f_t(H, H_{(a^*)}) = \sum_{\vec{b}\in Q_{a^*}^{(t)}} \mu_{\vec{b}}P_{\vec{b}}, \label{eq:sep-propagator-series}
    \end{align}
    where $|\operatorname{supp}(P_{\vec{b}})| \leq k t$ and
    \begin{align}
        \sum_{\vec{b} \in Q_{a^*}^{(t)}} |\mu_{\vec{b}}| \leq \prod_{i=1}^t \left(ks + 2 ksi\right) \leq \prod_{i=1}^t \left(3 k s i\right) = (3 k s)^t t!.
    \end{align}
    Here $Q_{a^*}^{(t)}$ for $t\geq 1$ represents all connected products of terms in the Hamiltonian.  Formally $Q_{a^*}^{(t)}$ is a sequence of terms $h_{b_1} , \cdots, h_{b_t}$ with $b_1,\ldots,b_t\in\mathcal{A}$ such that $b_1 \cap a^* \neq \emptyset$ and $b_j \cap (a^* \cup b_1 \cup \cdots \cup b_{j-1}) \neq \emptyset$ for all $j \geq 2$.  
\end{lemma} 

This result is notable since with the most na\"{i}ve expansion one would pick up $\operatorname{poly}(n)^t$ terms and the proof would quickly fail.  One subtlety here to ensure the proof couples properly throughout is that the input to this lemma is the set of sites corresponding to a term $a^*$ and not directly the neighborhood of $h_{a^*}$.  If we instead used the neighborhood the proof would fail as the neighborhood possibly includes polynomially many terms.  
\begin{proof}
First we will show that the $f_t(H, H_{(a^*)})$ abide by a nice recurrence (as is done in \cite{bakshi2024high}).  Recall that here $H_{(a^*)}$ represents the Hamiltonian restricted to those terms that touch the sites in $a^*$.  To begin, we note that the propagator can be re-expressed as 
\begin{align}
    e^{-\beta H} e^{\beta (H - H_{(a^*)})} 
    &= \sum_{t=0}^\infty \sum_{k=0}^t \left( \frac{\beta^k (-H)^k}{k!}\right) \left(\frac{\beta^{t-k} (H - H_{(a^*)})^{t-k}}{(t-k)!}\right) \nonumber \\
    &= \sum_{t=0}^\infty \frac{\beta^t}{t!}\sum_{k=0}^t \left(\frac{(-H)^k (H - H_{(a^*)})^{t-k} t!}{k!(t-k)!} \right) \nonumber \\
    &= \sum_{t=0}^\infty \frac{\beta^t}{t!}\sum_{k=0}^t \binom{t}{k}(-H)^k (H - H_{(a^*)})^{t-k}.
\end{align}
We check that the recurrence relation holds:
\begin{align}
    &-[H, f_{t}(H, H_{(a^*)})] - f_{t}(H, H_{(a^*)})H_{(a^*)} 
    = -H f_{t}(H, H_{(a^*)}) +  f_{t}(H, H_{(a^*)})  (H- H_{(a^*)}) \nonumber \\
    &= \sum_{k=0}^t \binom{t}{k}(-H)^{k+1} (H - H_{(a^*)})^{t-k} +  \sum_{k=0}^t \binom{t}{k}(-H)^k (H - H_{(a^*)})^{t-k+1}  \nonumber \\
    &= (-H)^{t+1} + (H-H_{(a^*)})^{t+1} + \sum_{k=1}^{t} \binom{t+1}{k}(-H)^{k} (H - H_{(a^*)})^{t-k+1}   \nonumber \\
    &= \sum_{k=0}^{t+1} \binom{t+1}{k}(-H)^{k} (H - H_{(a^*)})^{t+1-k} = f_{t+1}(H, H_{(a^*)}).
\end{align}
Using this recurrence we now show the desired bound on the coefficient mass at each order of the propagator expansion.  For the base case we can see directly from the expansion that
\begin{align}
    f_0(H, H_{(a^*)}) = I.
\end{align}
The first step of the recursion is trivial:
\begin{align}
    f_{1}(H, H_{(a^*)}) = -\left[\sum_a \lambda_a P_a, I\right] - I H_{(a^*)} = -H_{(a^*)}.
\end{align}
We can compute the norm at this step which is given by $\sum_{b:b\cap a^* \neq \emptyset} |\lambda_b|\leq \sum_{x\in a^*} \sup_x \sum_{a\ni x} |\lambda_a| \leq k s$.  Now we proceed to the second step of the iteration:
\begin{align}
    f_{2}(H, H_{(a^*)}) = \left[\sum_a \lambda_a P_a, H_{(a^*)}\right] + H_{(a^*)} H_{(a^*)}  .
\end{align}
First, let us look at the second term where we have 
\begin{align}
    H_{(a^*)} H_{(a^*)} = \sum_{b:b\cap a^* \neq \emptyset} \sum_{a:a\cap a^* \neq \emptyset} \lambda_a P_a \lambda_b P_b,
\end{align}
which can have its coefficient mass bounded as $\sum_{b:b\cap a^* \neq \emptyset} \sum_{a:a\cap a^* \neq \emptyset} |\lambda_a|  |\lambda_b| \leq (ks)^2$.  We also expand the first term as
\begin{align}
     \left[\sum_a \lambda_a P_a, H_{(a^*)}\right] =  \sum_a \lambda_a [P_a, H_{(a^*)}] = \sum_{a:a\cap \operatorname{supp}(H_{(a^*)})\neq\emptyset}\lambda_a [P_a, H_{(a^*)}].
\end{align}
Now we can bound the coefficient sum as 
\begin{align}
    \sum_{\vec{b}\in Q_{a^*}^{(2)}} |\mu_{\vec{b}}| 
    &\leq 2\sum_{b:b\cap a^* \neq \emptyset} \sum_{a:a\cap b\neq\emptyset} |\lambda_a \lambda_b|
    \leq  2\sum_{b:b\cap a^* \neq \emptyset} \sum_{a:a\cap b\neq\emptyset} |\lambda_b\|\lambda_a|
    \leq  2\sum_{b:b\cap a^* \neq \emptyset} |\lambda_b| \sum_{x\in b}\sum_{a\ni x} |\lambda_a|
    \nonumber \\
    &\quad\quad\leq  2\sum_{b:b\cap a^* \neq \emptyset} |\lambda_b| \sum_{x\in b} \sup_{x\in \Lambda} \sum_{a\ni x} |\lambda_a|
    \leq  2\sum_{b:b\cap a^* \neq \emptyset} |\lambda_b| \sum_{x\in b} s \leq 2k^2 s^2.
\end{align}
Combined with the second term, we thus have the bound
\begin{align}
    \sum_{\vec{b}\in Q_{a^*}^{(2)}} |\mu_{\vec{b}}| \leq 2 k^2 s^2 + k^2 s^2.
\end{align}
For the inductive step, we start with a series of connected terms $\mu_{\vec{b}}P_{\vec{b}} = h_{b_1} \cdots h_{b_t}$.  Since each term must overlap with one term from earlier in the sum there are at most $t k$ sites in the lattice represented in $\mu_{\vec{b}}P_{\vec{b}}$.  We compute the next order in the commutator using
\begin{align}
    f_{t+1}(H, H_{(a^*)}) = -[H, f_t(H, H_{(a^*)})]  - f_t(H, H_{(a^*)}) H_{(a^*)}.
\end{align}
Expanding out we get
\begin{align}
    f_{t+1}(H, H_{(a^*)}) &= -\sum_a \lambda_a [ P_a, f_t(H, H_{(a^*)})]  - f_t(H, H_{(a^*)}) H_{(a^*)}.
\end{align}
Focusing first on the second term we get 
\begin{align}
    f_t(H, H_{(a^*)}) H_{(a^*)}  = \sum_{\vec{b}\in Q_{a^*}^{(t)}}\sum_{a:a\cap a^* \neq \emptyset} \lambda_a \mu_{\vec{b}} P_{\vec{b}} P_a.
\end{align}
Now we compute a bound on the coefficient mass 
\begin{align}
    \sum_{\vec{b}\in Q_{a^*}^{(t)}}\sum_{a:a\cap a^* \neq \emptyset} | \lambda_a \mu_{\vec{b}}| \leq \sum_{\vec{b}\in Q_{a^*}^{(t)}}\sum_{a:a\cap a^* \neq \emptyset} | \lambda_a \| \mu_{\vec{b}}| &\leq \left(\sum_{a:a\cap a^* \neq \emptyset} | \lambda_a |\right)\left(\sum_{\vec{b}\in Q_{a^*}^{(t)}}| \mu_{\vec{b}}|\right) \leq ks \sum_{\vec{b}\in Q_{a^*}^{(t)}}| \mu_{\vec{b}}|.
\end{align}
Hence for the second part a factor of at most $ks$ is picked up.  Now looking at the first term we have that 
\begin{align}
    -\sum_a \lambda_a [ P_a, f_t(H, H_{(a^*)})] 
    &= - \sum_{\vec{b}\in Q_{a^*}^{(t)}} \sum_a \lambda_a \mu_{\vec{b}} [P_a , P_{\vec{b}}] .
\end{align}
To bound the size of the coefficients we then have 
\begin{align}
    \sum_{\vec{b}\in Q_{a^*}^{(t)}} \sum_{a:a\cap \operatorname{supp}(P_{\vec{b}})\neq\emptyset}|\lambda_a \mu_{\vec{b}}|
    &\leq \sum_{\vec{b}\in Q_{a^*}^{(t)}} \sum_{x \in b_1 \cup \cdots \cup b_t}\sum_{a\ni x}|\lambda_a\| \mu_{\vec{b}}| \nonumber \\
    &\leq \sum_{\vec{b}\in Q_{a^*}^{(t)}} \sum_{x \in b_1 \cup \cdots \cup b_t} \sup_{x \in \Lambda} \sum_{a\ni x}|\lambda_a\| \mu_{\vec{b}}| \nonumber \\
    &\leq k t s \sum_{\vec{b}\in Q_{a^*}^{(t)}} |\mu_{\vec{b}}|,
\end{align}
where we have used that the new term must touch one of the previous terms in the cluster.  Since each $a$ must contain a site in the support of $b_1 \cdots b_t$, and since there are $k t$ possible sites in the support of $\vec{b}$, we have that
\begin{align}
     \sum_{\vec{b}\in Q_{a^*}^{(t+1)}} |\mu_{\vec{b}}| &\leq (2 k t s + ks) \sum_{\vec{b}\in Q_{a^*}^{(t)}} |\mu_{\vec{b}}|\leq \prod_{i=0}^{t}(ks + 2 k s i) \leq \prod_{i=1}^{t+1} \left(ks + 2 k s i\right),
\end{align}
completing the proof of lemma.
\end{proof}
\noindent The preceding lemma shows that for sufficiently small $\beta$ the expansion of the propagator is convergent.  This in turn allows us to show the following lemma allowing the propagator to be expressed as a convex combination of terms of the form $I+c E$ where $c$ is exponentially small in the order $t = |C(E)|$. 
\begin{lemma}[Propagator sampling]
\label{lem:sep-propagator-sampling}
    Under the assumptions of \Cref{lem:sep-propagator-expansion}, there exists a distribution over tuples $(b,E)$ such that 
    \begin{align}
        e^{-\beta H} e^{\beta (H - H_{(a^*)})} = \sum_{t=0}^\infty \frac{\beta^t}{t!}f_t (H, H_{(a^*)}) = \sum_{i} p_i (I + b_i E_i) = \mathbb{E}_i [I + b_i E_i],
    \end{align}
    where $|b_i| \leq (\beta 6 k s)^{t_i}$.  Here $E_i$ decomposes as a product of $a \in C(E_i)$, an ordered list of elements from $\mathcal{A}$, and $t_i = |C(E_i)|\geq 1$ is the order of $E_i$.  If $\beta \leq \frac{1}{6 k s}$ then the coefficients $b_i$ are exponentially suppressed in $t_i$.
\end{lemma}
\begin{proof}
    Write the non-identity part of \cref{eq:sep-propagator-series} as
    \begin{align}
        \sum_{t=1}^\infty \sum_{\vec{b}\in Q^{(t)}_{a^*}} Y_{\vec{b}}, \qquad Y_{\vec{b}}=\frac{\beta^t}{t!}\mu_{b_1\cdots b_t} P_{b_1}\cdots P_{b_t}.
    \end{align}
    We expand out as follows (using an abbreviated notation where $Y_{\vec{b}}$ represents some $\frac{\beta^t}{t!}\mu_{b_1\cdots b_t} P_{b_1}\cdots P_{b_t}$) 
    \begin{align}
        I + \sum_{t=1}^\infty \sum_{\vec{b}\in Q^{(t)}_{a^*}} Y_{\vec{b}} = 
        \sum_{t=1}^\infty \sum_{\vec{b}\in Q^{(t)}_{a^*}} \frac{1}{2^t} \frac{\|Y_{\vec{b}}\|}{\sum_{\vec{b}\in Q^{(t)}_{a^*}} \|Y_{\vec{b}}\|}\left[I + 2^t (\sum_{\vec{b}\in Q^{(t)}_{a^*}} \|Y_{\vec{b}}\|)\frac{Y_{\vec{b}}}{\|Y_{\vec{b}}\|}\right]
    \end{align}
    Thus we can see that if we define $p_i = \frac{1}{2^t} \frac{\|Y_{\vec{b}}\|}{\sum_{\vec{b}\in Q^{(t)}_{a^*}} \|Y_{\vec{b}}\|}$ and $b_i = 2^t (\sum_{\vec{b}\in Q^{(t)}_{a^*}} \|Y_{\vec{b}}\|)$ then the lemma is proved.  Using \Cref{lem:sep-propagator-expansion} we can see that $b_i$ scales as $(\beta 2 (3 k s))^t$.  If we choose $\beta< \frac{1}{6 k s}$ then the series is convergent and hence terms can be sampled based on $p_i$ to realize the propagator.  Since $\frac{Y_{\vec{b}}}{\|Y_{\vec{b}}\|}$ is a Pauli product up to a phase the lemma is proved.  
\end{proof}
Now we can sample from the propagator in a well-defined way where coefficients are suppressed exponentially with $t$.  Next we present the recursive pinning to repeatedly invoke this sampling primitive and show that the Gibbs state is separable. 

\subsection{Pinning procedure}

The argument that follows is an improved version of that in \cite{bakshi2024high}. By pinning terms based on their presence in the monomial we can get optimal dependence on both $\mathfrak{d}$ and $k$.  

The natural quantities to keep track of during the recursion are products of the following form.
\begin{definition}[Hermitian monomial]
\label{def:sep-hermitian-monomial}
    A Hermitian monomial $X$ of order $\hat{t}$ is an operator recursively constructed by an operation which at each step of the recursion performs
    \begin{align}
        \hat{X} = \frac{1}{2} \left(E_1 X E_2^\dagger + E_2 X E_1^\dagger\right),
    \end{align}
    where $E_1$ and $E_2$ are Pauli products (up to a phase) returned by the sampling primitive (either of which may be identity).  In the first iteration $X$ starts as the identity.  The order and support of the monomial are computed recursively as 
    \begin{align}
        \hat{t} = t+t_{E_1}+t_{E_2}
    \end{align}
    and
    \begin{align}
        \operatorname{supp}(\hat X)=\operatorname{supp}(X)\cup \operatorname{supp}(E_1)\cup \operatorname{supp}(E_2)
    \end{align}
    respectively where $t_{E_1}$, $t_{E_2}$ are the order of $E_1$,$E_2$ from the sampling primitive \Cref{lem:sep-propagator-sampling}.  Here, support denotes the union of the supports of the terms which make up the monomial (which is not the support after multiplying out the monomial).  The terms which make up the monomial can be recursively computed as 
    \begin{align}
        C(\hat{X}) = C(X) \sqcup C(E_1) \sqcup C(E_2),
    \end{align}
    where $\sqcup$ indicates concatenation of the terms.
\end{definition}
\noindent Importantly, after many iterations the Hermitian monomial is still a Pauli product as captured by the following lemma.
\begin{lemma}
    \label{lem:sep-hermitian-monomial-pauli}
    A Hermitian monomial is either 0 or a signed Pauli product, i.e., has the form
    \begin{align}
        X \in \{0\}\cup\{\pm P : P \in \mathcal{P}^{\otimes n}\}
    \end{align}
\end{lemma}
\begin{proof}
    For the base case $X=I$ which is a Pauli product.  Assume that $X$ is a signed Pauli product $\pm P$ now let us confirm that at the next order of the recursion it is also a Pauli product.  $E_1$ and $E_2$ are both Pauli products so we know that $E_1 X E_2^\dagger = \omega P$ where $\omega \in \{\pm 1, \pm i\}$.  Hence we have that 
    \begin{align}
        \frac{1}{2} \left(E_1 X E_2^\dagger + E_2 X E_1^\dagger\right)=\frac{1}{2} \left(E_1 X E_2^\dagger + (E_1 X E_2^\dagger)^\dagger\right) = \frac{\omega + \omega^*}{2} P \in \{0, \pm P\}.
    \end{align} 
\end{proof}
\noindent In the algorithm we end up with a configuration of multiple Hermitian monomials.  The configuration is formally defined as 
\begin{definition}[Configuration of Hermitian monomials]
    \label{def:sep-configuration}
     A configuration $\chi$ of Hermitian monomials of length $l$ is described by an ordered set $\{(c_1, X_1), (c_2, X_2), \cdots, (c_l, X_l)\}$ where $X_i$ are disjoint Hermitian monomials and $c_i \in \mathbb{R}$.  The configuration corresponds to an operator given by 
    \begin{align}
        \sigma(\chi)  = \bigotimes_{i=1}^l (I + c_i X_i),
    \end{align}
    where the support is trivial on any element not in the configuration.
\end{definition}
\noindent Our goal then is to show that the Gibbs state can be expressed as a configuration of Hermitian monomials with each $|c| \leq 1$.  If each of the $X$ is then a Pauli then the Gibbs state is separable.  As a reminder before presenting the algorithm the Hamiltonian restrictions are given by 
\begin{align}
    H^{(S)} = \sum_{a\in \mathcal{A}^{(S)}} \lambda_a P_a, \quad H^{(S)}_{(a^*)} = \sum_{a\in\mathcal{A}^{(S)} :\ a\cap a^* \neq \emptyset} \lambda_a P_a.
\end{align}
We keep track of the unpinned labels represented in the current monomial with
\begin{align}
    u_S(X) = \{a\in \mathrm{unique}(C(X)) \,:\, a \cap S \neq \emptyset\}.
\end{align}
See \Cref{alg:sep-iterative-pinning} for a description of the algorithm. Note that in the algorithm it is possible that the sampled $a^*$ is partially pinned.  More carefully, this is when $a^* \in u_S(X_l)$ but $a^* \notin \mathcal{A}^{(S)}$.  In this case we use the convention that $H^{(S)}_{(a^*)}=\sum_{a \subseteq S: a \cap a^* \neq \emptyset} h_a$.  This captures that terms which overlap the already pinned sites, which are already pinned, are ignored. 

\begin{algorithm}
\caption{\textsc{Iterative Pinning}}
\label{alg:sep-iterative-pinning}
\begin{algorithmic}[1]
\REQUIRE Hamiltonian $H = \sum_a h_a = \sum_a \lambda_a P_a$ and sampling primitive satisfying $\mathbb{E}\left[I + b E\right] =  e^{-\eta H^{(S)}} e^{\eta \left(H^{(S)} - H^{(S)}_{(a^*)}\right)}$ with $|b| \leq L(\eta) q^t$.
\ENSURE A configuration $\chi = \{(c_1, X_1), \cdots, (c_j, X_j)\}$ such that $\mathbb{E}[\sigma(\chi)]= e^{-\beta H}$.
\STATE $S = \Lambda$, $\chi = \emptyset$, $l=0$.
\WHILE{$\mathcal{A}^{(S)} \neq \emptyset$}
    \IF{$l \geq 1$ and $\exists \hat{a} \in u_S(X_l)$}
        \STATE $a^* \gets \hat{a}$ and set $\hat{l} \gets l$
    \ELSE
        \STATE $c_{l+1} \gets 0$, $X_{l+1} \gets I$, add $(c_{l+1}, X_{l+1})$ to $\chi$, $\hat{l} \gets l+1$, set $a^*$ to any element in $\mathcal{A}^{(S)}$.
    \ENDIF
    \STATE Sample $b_1$, $E_1$ with $\eta \rightarrow \beta/2$, $H \rightarrow H^{(S)}$, and selected label $a^*$ so that $\mathbb{E}[I + b_1 E_1] = e^{-\frac{\beta}{2}H^{(S)}} e^{\frac{\beta}{2}(H^{(S)} - H^{(S)}_{(a^*)})}$.
    \STATE Sample $b_2$, $E_2$ with $\eta \rightarrow \beta/2$, $H \rightarrow H^{(S)}$, and selected label $a^*$ so that $\mathbb{E}[I + b_2 E_2] = e^{-\frac{\beta}{2}H^{(S)}} e^{\frac{\beta}{2}(H^{(S)} - H^{(S)}_{(a^*)})}$.
    \STATE Sample $J\in \{1,2,3,4,5,6,7\}$ with probabilities $p_1 = q$, $p_{\neq 1} = (1-q)/6$.
    \IF{$J=1$}
        \STATE $\hat{c} \gets (1/p_1) c_{\hat{l}}$,  $\hat{X} \gets X_{\hat{l}}$.
    \ELSIF{$J=2$}
        \STATE $\hat{c} \gets (1/p_2)b_1$, $\hat{X} \gets (E_1+E_1^\dagger)/2$.
    \ELSIF{$J=3$}
        \STATE $\hat{c} \gets (1/p_3)b_2$, $\hat{X} \gets (E_2+E_2^\dagger)/2$.
    \ELSIF{$J=4$}
        \STATE $\hat{c} \gets (1/p_4)b_1 c_{\hat{l}}$, $\hat{X} \gets (E_1^\dagger X_{\hat{l}} +X_{\hat{l}} E_1)/2$.
    \ELSIF{$J=5$}
        \STATE $\hat{c} \gets (1/p_5)b_2 c_{\hat{l}}$, $\hat{X} \gets (E_2^\dagger X_{\hat{l}} +X_{\hat{l}} E_2)/2$.
    \ELSIF{$J=6$}
        \STATE $\hat{c} \gets (1/p_6)b_1 b_2$, $\hat{X} \gets (E_1^\dagger E_2 +E_2^\dagger E_1)/2$.
    \ELSE
        \STATE $\hat{c} \gets (1/p_7)b_1 b_2 c_{\hat{l}}$, $\hat{X} \gets (E_2^\dagger X_{\hat{l}} E_1 +E_1^\dagger X_{\hat{l}} E_2)/2$.
    \ENDIF
    \STATE Set $c_{\hat{l}} \gets \hat{c}$, $X_{\hat{l}} \gets \hat{X}$; if $\hat{X}=0$, set $(c_{\hat{l}}, X_{\hat{l}})$ to $(0,I)$.
    \STATE $l \gets \hat{l}$, $S \gets S \setminus a^*$
\ENDWHILE
\RETURN $\chi$.
\end{algorithmic}
\end{algorithm}

Before formally proving the algorithm works we give some intuition for how our pinning procedure works.  First note that the coefficient $|c|$ grows by a factor of $1/p_1$ with every iteration that we select the case corresponding to $p_1$.  One would correctly worry about settings where we initially sample some long monomial say $I + cX$ where $c \propto (3\beta k s)^t$ but then for every subsequent iteration we sample the case corresponding to $p_1$.  In so doing for every iteration we would pick up a factor of $1/p_1$ and the monomial would update to $I + \left(\frac{1}{p_1}\right)^t cX$.  Then clearly if $\left(\frac{1}{p_1}\right)^t$ were allowed to overcome $c \propto (3\beta k s)^t$ the state would no longer be PSD.

\begin{figure}[H]
    \centering
    \includegraphics[width=.62\linewidth]{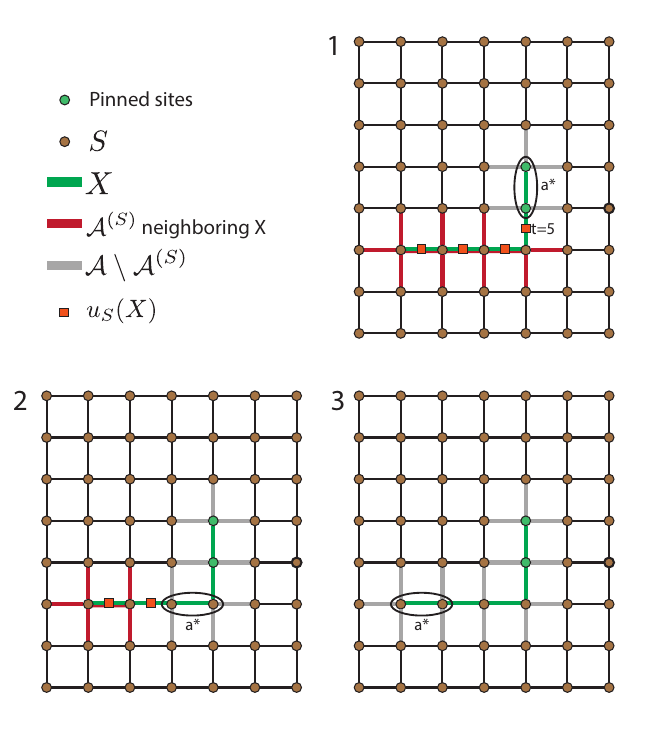}
    \caption{An example of how the pinning algorithm works.  Vertices correspond to qubits and the edges to the terms of the Hamiltonian which are all weight 2.  $S$ corresponds to the unpinned sites, $X$ to the current monomial, $\mathcal{A}^{(S)}$ to the terms of the Hamiltonian which are fully supported on the unpinned sites, and $u_S(X)$ are the terms which make up $X$ that touch an unpinned site.  The situation depicted is the particularly adversarial setting where the algorithm repeatedly samples the branch corresponding to $p_1$, where the monomial stays the same but the coefficient $c$ in $I+cP$ grows.  The initial sampling depicted in (1) selects a monomial of length $t=5$ indicated with the green edges.  Subsequent pinning steps in (2) and (3) choose a term in $u_S(X)$ which can only happen at most $t$ times.  Rather than selecting a term from $u_S(X)$ the previous algorithms select any term in $\mathcal{A}^{(S)}$ neighboring $X$.}
    \label{fig:new_iterative_pinning_example}
\end{figure}

In \Cref{fig:new_iterative_pinning_example} we depict this setting and how the pinning procedure we use would proceed.  We show a 2-local Hamiltonian where each edge of the graph corresponds to a term of the Hamiltonian and each vertex corresponds to a qubit.  We show 5 iterations of the sampling procedure.  In the first panel we can see that we took $a^*$ to be a vertical edge in the middle.  The sampling resulted in the initial monomial having order $t=5$, meaning that it consists of $5$ terms and has
\begin{align}
    |c| \leq L(\beta)\left(\beta q\right)^5.
\end{align}
Now the importance of the adaptivity of the procedure is clear.  For every subsequent iteration of the sampling we always select an $a^*$ which is one of the terms making up the monomial $X$.  Since the number of such terms in $X$ is bounded by $t$, the number of steps where the magnitude of $c$ increases is $< t $ and hence the resulting  $c'$ would be at most
\begin{align}
    |c'| \leq L(\beta)\left(\frac{1}{q}\right)^{t} q^t.
\end{align}
Thus in this example as long as we can make sure that $q \leq 1$ and $L(\beta) \leq 1$ (up to branch probabilities) we can be sure that the $|c'|$ will stay less than 1 and hence the state will be separable.  Our choice of $\beta$ is made exactly so that this will be true.  Our algorithm contrasts with that of \cite{bakshi2024high} where they only select requiring a single site to overlap and hence at most $k t$ iterations may happen.

We now prove that this algorithm works as expected.
\begin{lemma}[Pinning validity]
\label{lem:sep-pinning-validity}
    Every iteration of the algorithm completes and produces a valid configuration if the input is a valid configuration.  Here a valid configuration is one where the $X_i$ making up the configuration are disjoint and each $X_i$ is indeed a Hermitian monomial.  Furthermore if for $i\in [l-1]$ we have $u_S(X_i) = \emptyset$ then after the loop iteration for $i\in [\hat{l}-1]$ we have $u_{\hat{S}}(\hat{X}_i) = \emptyset$.  Also note that if an update gives $\hat{X}=0$ we discard the corresponding factor from the configuration.
\end{lemma}
\begin{proof}
    At every iteration of the algorithm an $a^*$ is chosen that either intersects $S$ or is fully supported in $S$.  In either case $a^* \cap S \neq \emptyset$ so the size of $S$ is strictly decreased and since $S$ is initialized with a finite size the algorithm will eventually terminate.  

    Assume that at the beginning of an iteration of the loop the monomials $X_i$ satisfy that $u_S(X_i) = \emptyset$ for all $i < l$.  During the loop iteration the first case is that there $\exists \hat{a} \in u_S(X_l)$ in which case the final monomial is updated but only using terms which have support fully in $S$, while $u_S(X_i)=\emptyset$ for $i < l$ means the earlier monomials have no terms intersecting $S$.  Note that $a^*$ may be partially pinned already ($a^*\notin \mathcal{A}^{(S)}$) but nonetheless will reduce the support of $S$ because $a^* \cap S \neq \emptyset$.  Alternatively $\nexists \hat{a} \in u_S(X_l)$ and a new monomial is made with $\hat{l} = l+1$ but the previous last monomial $u_{\hat{S}}(X_{l})=\emptyset$ since the new monomial was only made because there were no $a$ intersecting $S$.  
    
    Since the initial configuration in one loop iteration is valid and the loop only modifies the last monomial of the configuration then the new configuration is also valid because the invariant guarantees the final monomial is disjoint from the earlier monomials.  Each $X_i$ is by definition a Hermitian monomial since it is computed by starting with a Hermitian monomial and multiplying by the $E_i$ which are products of terms in the Hamiltonian.
\end{proof}
Next we show that throughout the procedure the expectation of the current monomials is consistent with the Gibbs state.
\begin{lemma}[Gibbs state validity]
\label{lem:sep-gibbs-validity}
    The procedure produces a $\hat{S}$ and new configuration $\hat{\chi}$ such that 
    \begin{align}
        \mathbb{E}\left[e^{-\frac{\beta}{2} H^{(\hat S)}} \sigma(\hat\chi) e^{-\frac{\beta}{2} H^{(\hat S)}}\right] = e^{-\frac{\beta}{2} H^{(S)}} \sigma(\chi) e^{-\frac{\beta}{2} H^{(S)}}.
    \end{align}
\end{lemma}
\begin{proof}
    Let us first show that at every step of the iteration the expectation is consistent with the Gibbs state $e^{-\beta H}$.  In other words we would like to show that at every step of the loop 
    \begin{align}\label{eq:sep-pinning-operator-invariant}
        e^{-\beta H} = \mathbb{E} \left[e^{-\frac{\beta}{2} H^{(S)}} \sigma(\chi) e^{-\frac{\beta}{2} H^{(S)}}\right].
    \end{align}
    We can see this is true at the beginning because $S = \Lambda$ and hence $H^{(S)} = H$ and we have that $e^{-\beta H} = e^{-\frac{\beta}{2} H}\left(I\right) e^{-\frac{\beta}{2} H}$ since here $\sigma(\chi)=I$. 
    
    Now let us consider a subsequent iteration where to start the final element of the configuration is $I + c_l X_l$ and at the end of the iteration the final monomial in the configuration is $I + \hat{c}_{\hat{l}} \hat{X}_{\hat{l}}$.  Denote the configuration before a loop iteration as $\chi$ and the configuration after a loop iteration as $\hat{\chi}$.  Likewise denote the set of unpinned sites before the iteration as $S$ and the set after the loop as $\hat{S}$.  Before the iteration we have the guarantee that 
    \begin{align}
        e^{-\beta H} = \mathbb{E} \left[e^{-\frac{\beta}{2} H^{(S)}} \sigma(\chi) e^{-\frac{\beta}{2} H^{(S)}}\right].
    \end{align} 
    Now let us show that if we fix a configuration before the iteration $\sigma(\chi)$ that after the iteration we have 
    \begin{align}
        \mathbb{E}\left[e^{-\frac{\beta}{2} H^{(\hat S)}} \sigma(\hat\chi) e^{-\frac{\beta}{2} H^{(\hat S)}}\right] = e^{-\frac{\beta}{2} H^{(S)}} \sigma(\chi) e^{-\frac{\beta}{2} H^{(S)}},
    \end{align}
    since all of the earlier elements of the monomial are supported disjointly from $H^{(S)}$.  
    Let us first focus on how the last monomial is updated based on the iteration.  Based on the sampling procedure in each iteration of the loop we find that 
    \begin{align}
        \mathbb{E}[I + \hat{c} \hat{X}] &= \frac{1}{2}\left((I + b_1 E_1^\dagger) (I + c_{\hat{l}} X_{\hat{l}}) (I + b_2 E_2) + (I + b_2 E_2^\dagger) (I + c_{\hat{l}} X_{\hat{l}})(I + b_1 E_1) \right).
    \end{align}
    Since $b_1$,$E_1$ and $b_2$,$E_2$ are sampled independently, and $\mathbb{E}\left[I + b_1 E_1\right] =  e^{-\frac{\beta}{2} H^{(S)}} e^{\frac{\beta}{2} \left(H^{(S)} - H^{(S)}_{(a^*)}\right)}$ and similarly for $\mathbb{E}\left[I + b_2 E_2\right]$, we get that 
    \begin{align}
        \mathbb{E}[I + \hat{c}\hat{X}] &= \frac{1}{2}\Bigg(e^{\frac{\beta}{2} \left(H^{(S)} - H^{(S)}_{(a^*)}\right)} e^{-\frac{\beta}{2} H^{(S)}} (I + c_{\hat{l}} X_{\hat{l}}) e^{-\frac{\beta}{2} H^{(S)}} e^{\frac{\beta}{2} \left(H^{(S)} - H^{(S)}_{(a^*)}\right)}  \nonumber \\
        &\quad \quad \quad + e^{\frac{\beta}{2} \left(H^{(S)} - H^{(S)}_{(a^*)}\right)} e^{-\frac{\beta}{2} H^{(S)}} (I + c_{\hat{l}} X_{\hat{l}})e^{-\frac{\beta}{2} H^{(S)}} e^{\frac{\beta}{2} \left(H^{(S)} - H^{(S)}_{(a^*)}\right)}  \Bigg) \nonumber \\
        &= e^{\frac{\beta}{2} \left(H^{(S)} - H^{(S)}_{(a^*)}\right)} e^{-\frac{\beta}{2} H^{(S)}} (I + c_{\hat{l}} X_{\hat{l}}) e^{-\frac{\beta}{2} H^{(S)}} e^{\frac{\beta}{2} \left(H^{(S)} - H^{(S)}_{(a^*)}\right)} \nonumber \\
        &= e^{\frac{\beta}{2}H^{(\hat S)}} e^{-\frac{\beta}{2} H^{(S)}} (I + c_{\hat{l}} X_{\hat{l}}) e^{-\frac{\beta}{2} H^{(S)}} e^{\frac{\beta}{2} H^{(\hat S)}}.
    \end{align}
    By \Cref{lem:sep-pinning-validity} all monomials $(c_j, X_j)$ in $\sigma(\chi)$ with $j < \hat{l}$ commute with $e^{-\frac{\beta}{2} H^{(S)}}$ since $u_S(X_j) = \emptyset$ and hence none of the labels in $C(X_j)$ intersect $S$.  Thus this result applies to the full configuration as well 
    \begin{align}
        \mathbb{E}\left[e^{-\frac{\beta}{2} H^{(\hat S)}} \sigma(\hat\chi) e^{-\frac{\beta}{2} H^{(\hat S)}}\right] = e^{-\frac{\beta}{2} H^{(S)}} \sigma(\chi) e^{-\frac{\beta}{2} H^{(S)}}.
    \end{align}
    The algorithm will continue to completion by \Cref{lem:sep-pinning-validity} and at this point $\mathcal{A}^{(S)}$ is empty and because the invariant is preserved every round $e^{-\beta H} = \mathbb{E}\left[\sigma(\chi) \right]$.
\end{proof}
Lastly we show that the coefficient $|c|$ for each monomial stays less than 1.
\begin{lemma}[Coefficient control]
\label{lem:sep-constant-controlled}
    If the sampling primitive satisfies $|b_i| \leq L(\beta/2) q^{t_i}$ with
    \begin{align}
        0 < q < 1, \quad \frac{6 L(\beta/2)}{1-q}\leq 1,
    \end{align}
    then at each step of the algorithm each monomial satisfies 
    \begin{align}
        |c_i| \leq q^{|u_S(X_i)|}.
    \end{align}
    Furthermore at the end of the algorithm $|c_i| \leq 1$ for all monomials.
\end{lemma}
\begin{proof}
    Now we need to show that every monomial in the configuration has $|c_i| \leq 1$ at the end of the iteration.  In so doing we will be sure that each term is separable and hence the Gibbs state as a whole is separable.  The potential function used in \cite{bakshi2024high} serves as a means to bound the possibilities of these different events.  We consider the following invariant of a given monomial
    \begin{align}
        |c| \leq \left(q\right)^{|u_S(X)|},
    \end{align}
    where $X = h_{a_1} \cdots h_{a_t}$ are the terms of the Hamiltonian in the current monomial and $\operatorname{supp}(X)$ gives the union of the support of each term in the monomial.  $q$ is a parameter we will determine through the calculation.  Let $c$, $X$ correspond to the monomial before one round of the loop and $\hat{c}, \hat{X}$ to after the loop.  We would like to show that if the invariant is satisfied before the loop then it will also be satisfied after.  Using \Cref{lem:sep-propagator-sampling} we have the guarantee that (note the factor of two because we take $\eta \rightarrow\beta/2$)
    \begin{align}
        |b_1| \leq L(\beta/2) q^{t_1}, \quad |b_2| \leq L(\beta/2) q^{t_2}.
    \end{align}
    Now for each of the possible branching cases we show that the invariant holds afterwards.
    \begin{enumerate}
        \item $p_1$:  If the current factor was recently created, trivially $\hat{c} = c=0$.  Otherwise since $a^* \in u_S(X)$ and $\hat{S} = S \setminus a^*$, $a^* \notin u_{\hat{S}}(X)$ and hence
        \begin{align}
            |u_{\hat{S}}(X)| \leq |u_{S}(X)| - 1 
        \end{align}
        and thus
        \begin{align}
            |\hat{c}| = \frac{1}{q} |c| \leq  \frac{1}{q} q^{|u_S(X)|} \leq q^{|u_S(X)|-1} \leq q^{|u_{\hat{S}}(\hat{X})|}.
        \end{align}
        \item $p_2$:  Here we sample a new term and we make use of the guarantees from \Cref{lem:sep-propagator-sampling} that the coefficients are exponentially small in the length of the monomial.  The bound should match the length of the final monomial being $\hat{t}=t_1$ in this case.  We show the bound is met with the following operations 
        \begin{align}
            |\hat{c}| 
            & = \frac{6}{1-q} |b_1| \leq  \frac{6L}{1-q} q^{t_1} \leq q^{t_1} \leq q^{|u_{\hat{S}}(\hat{X})|}.
        \end{align}
        Here we have used that $|u_{\hat{S}}(\hat{X})|\leq t_1$ and that $\frac{6L}{1-q} \leq 1$.  
        \item $p_3$: Same as $p_2$ except with $t_1\rightarrow t_2$. 
        \item $p_4$: Here the length of the final monomial is $\hat{t} = t+t_1$.  This follows with essentially the same manipulations as the $p_2$ case
        \begin{align}
            |\hat{c}|  
            = \frac{6}{1-q} |b_1| |c| 
            \leq  \frac{6L}{1-q} q^{t_1} q^{|u_S(X)|} 
            \leq q^{|u_S(X)|+t_1}  
            \leq q^{|u_{\hat{S}}(\hat{X})|},
        \end{align}
        where in the last step we have used that $|u_{\hat{S}}(\hat{X})| \leq |u_S(X)|+t_1$ since $C(\hat{X})$ can have grown by at most $t_1$ and $\hat{S}$ can only have decreased.  
        \item $p_5$: Same as $p_4$ except with $t_1\rightarrow t_2$.
        \item $p_6$: Same as $p_2$ except with $t_1\rightarrow t_1+ t_2$. 
        \item $p_7$: Same as $p_4$ except with $t_1\rightarrow t_1+t_2$.
    \end{enumerate}
    Throughout this proof the requirements have been the conditions $0 < q < 1$ and $6 L(\beta/2)/(1-q)\leq 1$ assumed in the lemma statement.  Since we have that for all iterations and all monomials $|c_i| \leq q^{|u_S(X_i)|}$, we then have that $|c_i| \leq 1$.  
\end{proof}
Now we prove a sufficient condition for the algorithm to yield separability.
\begin{lemma}[Pinning procedure]
\label{lem:sep-pinning}
    Consider a sampling procedure that takes a set of sites $S$ and a term $a^* \in \mathcal{A}$ such that $a^* \cap S \neq \emptyset$, and returns coefficients $b$ and Paulis $E$ such that
    \begin{align}
        \mathbb{E}\left[I + b E\right] =  e^{-\eta H^{(S)}} e^{\eta \left(H^{(S)} - H^{(S)}_{(a^*)}\right)}.
    \end{align}
    Here $E$ is constructed from $C(E)$ which is an ordered list of terms from $\mathcal{A}^{(S)}$ and the order of $E$ is $t = |C(E)|$.  If $|b| \leq (\beta v / 2)^t$, then the Gibbs state $e^{-\beta H}$ is separable for $\beta \leq 1/(12v)$.
\end{lemma}
\begin{proof}
    By \Cref{lem:sep-pinning-validity} and \Cref{lem:sep-gibbs-validity} the algorithm eventually terminates with $H^{(S)}=0$ and a valid configuration $\chi$ satisfying 
    \begin{align}
        e^{-\beta H} = \mathbb{E}[\sigma(\chi)].
    \end{align}
    This means the expectation of the final distribution equals the unnormalized Gibbs operator.  Since the configuration is valid each term consists of a Hermitian monomial.  By \Cref{lem:sep-hermitian-monomial-pauli} each of the Hermitian monomials is either 0 or a signed Pauli.  By \Cref{lem:sep-constant-controlled} the coefficient on each term has magnitude $\leq 1$ since the lemma provides the required constraints on $q$ and $L(\beta/2)$.  Hence by \Cref{lem:sep-pauli-factor} the state is separable.
    
    In the specific setting when $|b| \leq (\eta v)^t$ we rewrite the bound on $b$ as
    \begin{align}
        |b| \leq (\eta v)^{t} = \left(\frac{\eta v}{q}\right)^{t} q^t \leq \left(\frac{\eta v}{q}\right) q^t ,
    \end{align}
    where we have assumed that $\eta v \leq q$ and  $0<q<1$.  Hence in this setting $L(\eta) = \frac{\eta v}{q}$.  We can now confirm that our choice of $\beta$ and $q$ satisfy the required inequality.  There are three things we need to confirm for this to all work $0 < q < 1$, $\frac{6 L(\beta/2)}{1-q}\leq 1$, and $L(\beta/2) \leq 1$.  The last constraint is implied by the second since $0 < q < 1$ so we focus on $\frac{6 L(\beta/2)}{1-q}\leq 1$ and find $\frac{6 L(\beta/2)}{1-q} = \frac{3 \beta v}{q(1-q)} \leq 1$ which implies
    \begin{align}
    3 \beta v \leq q(1-q) .
    \end{align}
    The right side is maximized at $q=\frac{1}{2}$ which also satisfies $0<q<1$. Hence the final bound is $\beta v \leq 1/12$ as desired.
\end{proof}

Lastly we tie all these results together to prove separability of $(s,k)$-long-range Hamiltonians.
\begin{proof}[Proof of \Cref{thm:sep}(a), (b)]
By \Cref{lem:sep-propagator-sampling} for a $(s,k)$-long-range Hamiltonian one can sample coefficients $b$ and Paulis $E$ such that
\begin{align}
    \mathbb{E}\left[I + b E\right] =  e^{-\eta H^{(S)}} e^{\eta \left(H^{(S)} - H^{(S)}_{(a^*)}\right)} 
\end{align}
with the guarantee that 
\begin{align}
    |b| \leq (\eta 6 k s)^t.
\end{align}
Hence by \Cref{lem:sep-pinning} the Gibbs state is separable at temperature $\beta \leq \frac{1}{72 s k}$.

By \cite[Algorithm 4.6 and Lemma 4.7]{bakshi2024high} for a $(\mathfrak{d},k)$-low-intersection Hamiltonian one can sample coefficients $b$ and Paulis $E$ such that
\begin{align}
    \mathbb{E}\left[I + b E\right] =  e^{-\eta H^{(S)}} e^{\eta \left(H^{(S)} - H^{(S)}_{(a^*)}\right)} 
\end{align}
with the guarantee that 
\begin{align}
    |b| \leq (\eta 8 \mathfrak{d})^t.
\end{align}
Hence by \Cref{lem:sep-pinning} the Gibbs state is separable at temperature $\beta \leq \frac{1}{96 \mathfrak{d}}$.
\end{proof}

\section{Death of magic}
\label{sec:magic}
We prove in this section the lower bound on $\beta_{\rm stab}$ reported in \Cref{thm:stab}; the upper bound is shown in \Cref{thm:stab-upper} (\Cref{sec:upper}). Although we will ultimately report our bounds in terms of the parameters $w_{\rm free}, w_{\rm pert}$, we will prove our results in terms of a somewhat stricter quantity $\Delta(\beta)$, which satisfies
\begin{align}\label{eq:deltabound}
    \Delta(\beta)\leq w_{\rm pert} \,\eta_\beta(w_{\rm free}) \qquad \text{for} \qquad \eta_\beta(\kappa) = \begin{cases} \displaystyle \frac{e^{\beta\kappa}\lr{e^{\beta\kappa}-1}}{\kappa}, &\kappa>0,\\[1.2em]
    \beta,&\kappa=0, \end{cases}.
\end{align}
To track anticommutation with terms $P_\sigma$ in $H_0$, we introduce
\begin{align}
    \cF_{\rm anti}(W) = \set{\sigma\in\cF:P_\sigma W=-WP_\sigma}.
\end{align}
We will also use notation
\begin{align}\label{eq:kappa}
    \kappa(W)=\sum_{\sigma\in \cF_{\rm anti}(W)}|u_\sigma|, \qquad \kappa(0) = 0
\end{align}
so, e.g., $w_{\rm free} = \max_{\mu\in\cG} \kappa(Q_\mu)$ when $\cG \neq \emptyset$. The quantity $\Delta$ is defined as 0 when $\cG = \emptyset$ and otherwise
\begin{align}\label{eq:delta}
    \Delta(\beta) = \max_{\mu\in\cG} \lr{ |v_\mu|\,\eta_\beta(\kappa(Q_\mu)) + \sum_{\nu\in\cG\setminus\set{\mu}:\,\nu\sim\mu} |v_\nu|\,\eta_\beta(\kappa(Q_\nu)) }
\end{align}
To show \cref{eq:deltabound}, it suffices to observe that $\eta_\beta(\kappa)$ is increasing in $\kappa \geq 0$.

For fixed $\beta$, define the activity of a perturbing term $\nu\in\cG$ by
\begin{align}\label{eq:anu}
a_\nu := |v_\nu|\,\eta_\beta\!\bigl(\kappa(Q_\nu)\bigr).
\end{align}
Thus we have
\begin{align}
\Delta(\beta) = \max_{\mu\in\cG} \left(a_\mu + \sum_{\substack{\nu\in\cG\setminus\{\mu\}\\ \nu\sim\mu}} a_\nu \right).
\end{align}

\subsection{Stabilizer decomposition}\label{subsec:stabdecomp}

Recall that in \cite{bakshi2024high}, the Gibbs state is decomposed as a distribution over stabilizer product states, i.e.~$\bigotimes_i A_i$ for $A_i \in \{\frac{1}{2}\lr{I \pm X}, \frac{1}{2}\lr{I \pm Y}, \frac{1}{2}\lr{I \pm Z}\}$. In particular, each $A_i$ is a projector onto a pure stabilizer state (stabilized by the respective Pauli). Here, we will show that the Gibbs state is a mixture of stabilizer states that may not be separable.

We denote the convex hull of pure $n$-qubit stabilizer states by $\STAB_n$ (\Cref{def:sepstab}), and let
\begin{align}
    \StabCone_n=\set{t\rho:t\geq0,\ \rho\in\STAB_n} = \set{\sum_a c_a \ketbra{\phi_a} \,:\, c_a \geq 0, \, \ket{\phi_a}\text{ is an } n\text{-qubit pure stabilizer state}}
\end{align}
be the stabilizer cone. It is enough to prove that $e^{-\beta H}\in\StabCone_n$, since normalization by its trace then gives $\rho_\beta(H)\in\STAB_n$.

In the separability proof, \cite{bakshi2024high} uses the fact that for Pauli $P$, the matrix $I+c P$ is separable for any $|c| \leq 1$. This is not preserved under products: observe that
\begin{align}
    \lr{I + c_X X\otimes X}\lr{I + c_Z Z\otimes Z}
\end{align}
at $c_X = c_Z = 1$ is proportional to the Bell state $\ket{00}+\ket{11}$ and is thus not separable. It is, however, stabilizer. We generalize this to obtain a criterion for being contained in the stabilizer cone.

\begin{lemma}[Stabilizer criterion for commuting Paulis]
\label{lem:stabcrit}
Let $Q_1,\ldots,Q_m$ be pairwise commuting Hermitian Pauli strings, and
let $|\lambda_j|\leq1$ for every $j$. Then
\begin{align}
    \prod_{j=1}^m(I+\lambda_jQ_j)\in\StabCone_n.
\end{align}
\end{lemma}
\begin{proof}
Because the $Q_j$ commute and are Hermitian, they are simultaneously diagonalizable. For
$\varepsilon=(\varepsilon_1,\ldots,\varepsilon_m)\in\set{\pm1}^m$, let
$\Pi_\varepsilon$ be the projector onto the joint eigenspace $Q_j\psi=\varepsilon_j\psi$ for $j \in [m]$, so
\begin{align}
    \prod_{j=1}^m(I+\lambda_jQ_j) = \sum_{\varepsilon\in\set{\pm1}^m} \prod_{j=1}^m(1+\lambda_j\varepsilon_j)\Pi_\varepsilon.
\end{align}
Since $|\lambda_j|\leq1$, every coefficient
$1+\lambda_j\varepsilon_j$ is nonnegative. It remains to show that
$\Pi_\varepsilon\in\StabCone_n$.

Let $G_\varepsilon$ be the abelian Pauli subgroup generated by the
operators $\varepsilon_jQ_j$. If $-I\in G_\varepsilon$, then the
constraints $Q_j\psi=\varepsilon_j\psi$ are inconsistent and $\Pi_\varepsilon=0$,
which belongs to $\StabCone_n$. Otherwise choose an independent
commuting generating set $S_1,\ldots,S_k$ for $G_\varepsilon$. Then
\begin{align}
    \Pi_\varepsilon=\prod_{a=1}^k\frac{I+S_a}{2}.
\end{align}
Extend $S_1,\ldots,S_k$ to an independent commuting family of Paulis $S_1,\ldots,S_k,T_{k+1},\ldots,T_n$. For $\eta=(\eta_{k+1},\ldots,\eta_n)\in\set{\pm1}^{n-k}$, let $P_\eta$
be the projector onto the common eigenspace
\begin{align}
    S_a=+1\quad(a=1,\ldots,k), \qquad T_b=\eta_b\quad(b=k+1,\ldots,n).
\end{align}
This is a full independent set of $n$ commuting Pauli constraints, so
$P_\eta$ is a rank-one stabilizer projector. Moreover,
\begin{align}
    \sum_{\eta\in\set{\pm1}^{n-k}}P_\eta = \prod_{a=1}^k\frac{I+S_a}{2} = \Pi_\varepsilon.
\end{align}
Thus $\Pi_\varepsilon$ is a positive sum of pure stabilizer projectors,
and hence $\Pi_\varepsilon\in\StabCone_n$.
\end{proof}

In the absence of such commuting structure, we need a sharper criterion. We give one here based on the compatibility relation.
\begin{lemma}[Stabilizer criterion for incompatible Paulis]
\label{lem:free-slack}
Let $X_1,\ldots,X_m$ be signed Hermitian Pauli strings such that
$X_i\not\sim X_j$ for all $i\ne j$. If $|\lambda_j|\leq e^{-\beta\kappa(X_j)}$ for all $1\leq j\leq m$, then
\begin{align}
    e^{-\beta H_0/2} \prod_{j=1}^m(I+\lambda_jX_j) e^{-\beta H_0/2} \in\StabCone_n.
\end{align}
\end{lemma}
\begin{proof}
Write
\begin{align}
    H_j=\sum_{\sigma\in \cF_{\rm anti}(X_j)}u_\sigma P_\sigma, \qquad H_{\rm rest}=H_0-\sum_{j=1}^mH_j.
\end{align}
Since $X_j$ anticommutes with $H_j$ by definition, we have the identity
\begin{align}
    e^{-\beta H_j/2} (I+\lambda_j X_j) e^{-\beta H_j/2} = e^{-\beta H_j} + e^{-\beta H_j/2} \lambda_j e^{\beta H_j/2} X_j = e^{-\beta H_j} + \lambda_j X_j.
\end{align}
Since all the free terms that anticommute with $X_j$ were removed into the corresponding $H_j$, $H_{\rm rest}$ commutes with every $X_j$. Since $X_j \not\sim X_k$ ensures that $\cF_{\rm anti}(X_j)$ and $\cF_{\rm anti}(X_k)$ are disjoint for $k \neq j$, every Pauli term in $H_j$ commutes with $X_k$. This gives
\begin{align}\label{eq:restdecomp}
    &e^{-\beta H_0/2} \prod_{j=1}^m(I+\lambda_jX_j) e^{-\beta H_0/2} = e^{-\beta H_{\rm rest}} \prod_{j=1}^m\lr{e^{-\beta H_j}+\lambda_jX_j}.
\end{align}
Since $H_{\rm rest}$ is composed of commuting Pauli strings, $e^{-\beta H_{\rm rest}} \in \StabCone_n$. We will show momentarily that each term $e^{-\beta H_j} + \lambda_j X_j$ is expanded into either a projector from the commuting family
\begin{align}\label{eq:proj1}
    \{P_\sigma\,:\,\sigma\in\cF_{\rm anti}(X_j)\}
\end{align}
or is the Pauli projector
\begin{align}\label{eq:proj2}
    I + \sgn(\lambda_j) X_j.
\end{align}
Since the anticommutation sets $\cF_{\rm anti}(X_j)$ and $\cF_{\rm anti}(X_k)$ are disjoint, and $[X_j, X_k] = 0$, these stabilizer projectors all commute. Finally, since all the terms commute with $H_{\rm rest}$, it only contributes projectors that commute with all of the above. Hence, \cref{eq:restdecomp} is in $\StabCone_n$, completing the proof.

We now show that each factor $e^{-\beta H_j} + \lambda_j X_j$ is given by projectors as claimed in \cref{eq:proj1} and \cref{eq:proj2}. Decompose
\begin{align}
    e^{-\beta H_j} + \lambda_j X_j = \lr{e^{-\beta H_j} - |\lambda_j| I} + |\lambda_j|\lr{I + \sgn(\lambda_j)X_j}.
\end{align}
The second term is a positive multiple of $I \pm X_j$, which is twice a Pauli stabilizer projector. To show that the first term is also in the stabilizer cone, we diagonalize in the basis of $H_j$. Since
\begin{align}
    \kappa(X_j) = \sum_{\sigma\in\cF_{\rm anti}(X_j)} |u_\sigma|,
\end{align}
and since the terms in $H_j$ commute and have eigenvalues $\pm u_\sigma$, every eigenvalue of $H_j$ is in $[-\kappa(X_j), \kappa(X_j)]$. Since $|\lambda_j| \leq e^{-\beta \kappa(X_j)}$ by assumption, we have that $e^{-\beta H_j} - |\lambda_j| I$ is PSD in the same basis of commuting Paulis as $H_j$, and is thus a positive linear combination of stabilizer projectors.
\end{proof}

\subsection{Propagator}
The key to the separability proof is the repeated application of a propagation step that pins qubits individually. Roughly, this takes the form of
\begin{align}
    e^{-\beta H} e^{\beta (H - H_A)} = \sum_{t=0}^\infty \frac{\beta^t}{t!}f_t (H, H_A) = \sum_{i} p_i (I + c_i P_i) = \mathbb{E}_i [I + c_i P_i],
\end{align}
and then the pinning procedure updates a set $S$ of unpinned sites into $\wh S = S \setminus A$. In this subsection, we will show a similar procedure, but we will introduce a propagator $\cU_S(t_1, t_2)$ that will ultimately satisfy
\begin{align}
    \cU_{\wh S}(0, \beta/2) \cU_S(\beta/2, 0) = \E{I+bE}
\end{align}
for some distribution over $b \geq 0$ and Pauli strings $E$.

Let us begin by defining $\cU_S$ and noting some of its properties. For $S\subseteq\cG$, set
\begin{align}
    V_S=\sum_{\mu\in S}v_\mu Q_\mu , \qquad V_S(t)=e^{tH_0}V_Se^{-tH_0}.
\end{align}
Similarly, we will write $Q_\mu(t) = e^{tH_0}Q_\mu e^{-tH_0}$ so $V_S(t) = \sum_{\mu \in S} v_\mu Q_\mu(t)$.
For $t_2,t_1\in[-\beta/2,\beta/2]$, let $\cU_S(t_2,t_1)$ be the solution of
\begin{align}\label{eq:usdef}
    \partial_{t_2}\cU_S(t_2,t_1) = -V_S(t_2)\cU_S(t_2,t_1), \qquad \cU_S(t_1,t_1)=I,
\end{align}
so that
\begin{align}
    e^{-\beta(H_0+V_S)} = e^{-\beta H_0/2} \cU_S(\beta/2,-\beta/2) e^{-\beta H_0/2}.
\end{align}

\begin{lemma}[Properties of $\cU_S$]
\label{lem:propagator-identities}
For every $S\subseteq\cG$, the propagator satisfies
\begin{align}
    \cU_S(t_3,t_2)\cU_S(t_2,t_1) = \cU_S(t_3,t_1), \qquad \cU_S(t_2,t_1)^{-1}=\cU_S(t_1,t_2), \qquad \cU_S(t_2,t_1)^\dagger = \cU_S(-t_1,-t_2).
\end{align}
Moreover, if a phased Pauli string $W$ satisfies
$W\not\sim Q_\mu$ for every $\mu\in S$, then $W$ commutes with
$V_S(t)$ for every $t$ and hence with $\cU_S(t_2,t_1)$ for all
$t_1,t_2$.
\end{lemma}
\begin{proof}
The composition and inverse identities follow from uniqueness of
solutions to the defining differential equation. Since
$V_S(t)^\dagger=V_S(-t)$, the adjoint
$\cU_S(t_2,t_1)^\dagger$ and the propagator
$\cU_S(-t_1,-t_2)$ solve the same differential equation with the
same initial condition.

For the final claim, expand $Q_\mu(t)=e^{tH_0}Q_\mu e^{-tH_0}$. Every Pauli string in
this expansion is a product of $Q_\mu$ and free Paulis from
$\cF_{\rm anti}(Q_\mu)$. If $W\not\sim Q_\mu$, then $W$ commutes with
$Q_\mu$ and with every Pauli in $\cF_{\rm anti}(Q_\mu)$. Hence $W$
commutes with $Q_\mu(t)$, and therefore with $V_S(t)$. Applying this to the ODE of the commutator gives
\begin{align}
    \partial_t[W, \cU_S(t, t_1)] &= -WV_S(t) \,\cU_S(t, t_1) + V_S(t) \,\cU_S(t, t_1)W = -V_S(t)[W, \cU_S(t, t_1)].
\end{align}
Since $[W, \cU_S(t_1, t_1)] = [W, I] = 0$, we find that $[W, \cU_S(t, t_1)] = 0$ for all $t$.
\end{proof}

In the separability proof, an invariant such as $|c| \leq q^{|u_S(X)|}$ is maintained throughout the procedure. (In the coarser version, $|u_S(X)|$ is replaced by $|S \cap \supp(X)|$.) In our setting, we replace this with $\Phi_S$, which will count unpinned terms represented in the current monomial $X$. Instead of counting the number of unpinned sites intersecting $X$, it counts the remaining perturbative terms $Q_\mu$ satisfying $Q_\mu \sim X$. We define it as $\Phi_S(W) = 0$ when $\Delta(\beta) = 0$ or $W = 0$, and
\begin{align}\label{eq:phi}
    \Phi_S(W) = \Delta(\beta)^{-1} \displaystyle\sum_{\mu\in S:\,\mu\sim W} |v_\mu|\,\eta_\beta(\kappa(Q_\mu))
\end{align}
otherwise.

\begin{lemma}[Properties of $\Phi$]
\label{lem:potential}
Let $S\subseteq\cG$.

\begin{enumerate}
\item If $W,Z$ are nonzero phased Pauli strings, then
\begin{align}
    \Phi_S(WZ)\leq \Phi_S(W)+\Phi_S(Z).
\end{align}

\item If $D=(\mu_1,\ldots,\mu_t)$ is a list of labels in $S$, and $E$
is, up to a phase, a product of $Q_{\mu_1},\ldots,Q_{\mu_t}$ and
free Paulis from $\bigcup_{j=1}^t\cF_{\rm anti}(Q_{\mu_j})$, then
\begin{align}
    \Phi_S(E)\leq t.
\end{align}
\end{enumerate}
\end{lemma}
\begin{proof}
If $\Delta(\beta)=0$, both claims are immediate from the
definition of $\Phi_S$. Assume $\Delta(\beta)>0$. For the first claim, if $Q_\mu$ anticommutes with $WZ$, then it
anticommutes with $W$ or with $Z$. Also
\begin{align}
    \cF_{\rm anti}(WZ) = \cF_{\rm anti}(W)\triangle \cF_{\rm anti}(Z) \subseteq \cF_{\rm anti}(W)\cup \cF_{\rm anti}(Z).
\end{align}
Thus $\mu\sim WZ$ implies $\mu\sim W$ or $\mu\sim Z$, and the
claim follows from the definition of $\Phi$ in \cref{eq:phi}.

For the second claim, if $\mu\sim E$, then either $Q_\mu$
anticommutes with one of the $Q_{\mu_j}$, or $Q_\mu$ shares a free
Pauli with one of the $Q_{\mu_j}$. Hence $\mu\sim\mu_j$
for some $j$. Therefore
\begin{align}
    \begin{aligned} \Delta(\beta) \Phi_S(E) = \sum_{\mu\in S:\,\mu\sim E}|v_\mu|\,\eta_\beta(\kappa(Q_\mu)) &\leq \sum_{j=1}^t \sum_{\mu\in S:\,\mu\sim\mu_j}|v_\mu|\,\eta_\beta(\kappa(Q_\mu)) \leq t\Delta(\beta). \end{aligned}
\end{align}
\end{proof}

We record a few reused lemmas for convenience.
\begin{lemma}\label{lem:graph-sample-source-activity}
    For $\mu \in \cG$, the Pauli expansion
    \begin{align}
        Q_\mu(s) = \sum_{A\subseteq \cF_{\rm anti}(Q_\mu)} q_{\mu,A}(s)R_{\mu,A}
    \end{align}
    holds for Pauli strings $R_{\mu,A}$ that are (up to a phase) products of $Q_\mu$ with Paulis in $\cF_{\rm anti}(Q_\mu)$, and for coefficients $q_{\mu,A}(s)$ that satisfy
    \begin{align}
        \sum_{A\subseteq \cF_{\rm anti}(Q_\mu)} \int_0^{\beta/2} |v_\mu|\,|q_{\mu,A}(s)|\, e^{\beta\kappa(R_{\mu,A})}\,ds = \frac12\,|v_\mu|\,\eta_\beta(\kappa(Q_\mu)).
    \end{align}
\end{lemma}
\begin{proof}
    Since the free Paulis all commute, we can evaluate the conjugation one free term at a time. If $P_\sigma$ commutes with $Q_\mu$, then
    \begin{align}
        e^{su_\sigma P_\sigma}Q_\mu e^{-su_\sigma P_\sigma} = Q_\mu. \label{eq:graph-sample-commuting-free-conjugation}
    \end{align}
    If $P_\sigma$ anticommutes with $Q_\mu$, then
    \begin{align}
        e^{su_\sigma P_\sigma}Q_\mu e^{-su_\sigma P_\sigma} = \cosh(2su_\sigma)Q_\mu + \sinh(2su_\sigma)P_\sigma Q_\mu. \label{eq:graph-sample-anticommuting-free-conjugation}
    \end{align}
    Hence, in the Pauli expansion, each $R_{\mu,A}$ is, up to phase, the Pauli string
    \begin{align}
        \lr{\prod_{\sigma\in A}P_\sigma}Q_\mu , \label{eq:graph-sample-dressed-pauli-term}
    \end{align}
    and
    \begin{align}
        \sum_{A\subseteq \cF_{\rm anti}(Q_\mu)} |q_{\mu,A}(s)| = \prod_{\sigma\in\cF_{\rm anti}(Q_\mu)} \lr{|\cosh(2su_\sigma)|+|\sinh(2su_\sigma)|} &= e^{2s\kappa(Q_\mu)}. \label{eq:graph-sample-dressed-coefficient-l1}
    \end{align}
    since for every $A\subseteq \cF_{\rm anti}(Q_\mu)$, multiplying $Q_\mu$ by
    free Paulis does not change which free Paulis anticommute with it, so $\kappa(R_{\mu,A})=\kappa(Q_\mu)$.
    This gives
    \begin{align}
        &\sum_{A\subseteq \cF_{\rm anti}(Q_\mu)} \int_0^{\beta/2} |v_\mu|\,|q_{\mu,A}(s)|\, e^{\beta\kappa(R_{\mu,A})}\,ds  = |v_\mu|e^{\beta\kappa(Q_\mu)} \int_0^{\beta/2} e^{2s\kappa(Q_\mu)}\,ds  = \frac12\,|v_\mu|\,\eta_\beta(\kappa(Q_\mu)). \label{eq:graph-sample-source-activity}
    \end{align}
\end{proof}

We will use the following standard facts.
\begin{fact}[Dyson expansion]\label{fac:dyson}
Let $A:[0,T]\to \mathbb C^{d\times d}$ be continuous. The solution of
\begin{align}
    G'(t)=A(t)G(t), \qquad G(0)=I,
\end{align}
is given, for $0\leq t\leq T$, by
\begin{align}
    G(t) = I+ \sum_{m\geq1} \int_{0\leq s_m\leq\cdots\leq s_1\leq t} A(s_1)\cdots A(s_m)\, ds_1\cdots ds_m.
\end{align}
\end{fact}

\begin{fact}[Tree function]
\label{fac:tfunc}
    The function $T(x) = \sum_{n \geq 1} \frac{n^{n-1}}{n!} x^n$ satisfies $T(1/3) \in (0,1)$.
\end{fact}
\begin{proof}
    Lagrange inversion gives $x=T(x)e^{-T(x)}$ for $0 \leq x < 1/e$. On $[0,1]$, the function $f(t) = t e^{-t}$ satisfies $f'(t) \geq 0$, where the inequality is strict for $t \in [0,1)$; moreover, it has endpoints $f(0) = 0$ and $f(1) = 1/e > 1/3$. By the intermediate value theorem, $T(1/3)$ is thus given by a unique $t \in (0,1)$ such that $te^{-t}=1/3$.
\end{proof}

We can now show the main propagation lemma.

\begin{lemma}[Propagator sampling]
\label{lem:graph-sample}
Let $S\subseteq\cG$, let $\mu_*\in S$, and assume
$\Delta(\beta)>0$. Put
\begin{align}
    \wh S=S\setminus\set{\mu_*}, \qquad \theta=\frac{|v_{\mu_*}|\,\eta_\beta(\kappa(Q_{\mu_*}))}{\Delta(\beta)}.
\end{align}
Then $0<\theta\leq1$. There is a distribution over $b\geq0$, phased
Pauli $E$, and list of labels $D$ in $S$ such that
\begin{align}
    \E{I+bE} = \cU_{\wh S}(0,\beta/2)\cU_S(\beta/2,0).
\end{align}
Whenever $b\ne0$, the list
$D=(\mu_1,\ldots,\mu_t)$ has length $t\geq1$ and satisfies:

\begin{enumerate}
\item $\mu_1=\mu_*$;
\item for every $j>1$, either $\mu_j=\mu_*$ or
$\mu_j\sim\mu_i$ for some $i<j$;
\item $E$ is, up to a phase, a product of
$Q_{\mu_1},\ldots,Q_{\mu_t}$ and free Paulis from
$\bigcup_{j=1}^t\cF_{\rm anti}(Q_{\mu_j})$;
\item $\Phi_S(E)\leq t$;
\item $b e^{\beta\kappa(E)}
        \leq \theta\lr{3\Delta(\beta)}^t.$ 
\end{enumerate}
For the exceptional $b=0$ outcome one may take $E=I$ and $D=\emptyset$.
\end{lemma}
\begin{proof}
The object that we want to express as an expectation is $G(\beta/2)$ for
\begin{align}
    G(t)=\cU_{\wh S}(0,t)\cU_S(t,0), \qquad 0\leq t\leq \beta/2.
\end{align}
Using the defining equation \cref{eq:usdef}
\begin{align}\label{eq:deriv1}
    \partial_t\,\cU_S(t,0) = -V_S(t)\,\cU_S(t,0),
\end{align}
and \Cref{lem:propagator-identities} to get
\begin{align}\label{eq:deriv2}
    \partial_t\cU_{\wh S}(0,t) = \cU_{\wh S}(0,t)V_{\wh S}(t),
\end{align}
we find
\begin{align}
    G'(t) &= \cU_{\wh S}(0,t)\bigl(V_{\wh S}(t)-V_S(t)\bigr)\cU_S(t,0)  = -v_{\mu_*}\, \cU_{\wh S}(0,t)Q_{\mu_*}(t)\cU_S(t,0) = -v_{\mu_*}\, \cU_{\wh S}(0,t)Q_{\mu_*}(t)\cU_{\wh S}(t,0)G(t). \label{eq:graph-sample-G-ode}
\end{align}
The Dyson expansion (\Cref{fac:dyson}) of \cref{eq:graph-sample-G-ode} gives
\begin{align}
    G(\beta/2) &= I+ \sum_{m\geq1} \int_{0\leq s_m\leq\cdots\leq s_1\leq \beta/2} (-v_{\mu_*})^m \prod_{i=1}^m \lr{ \cU_{\wh S}(0,s_i)Q_{\mu_*}(s_i)\cU_{\wh S}(s_i,0) } \,ds_1\cdots ds_m. \label{eq:graph-sample-G-dyson}
\end{align}
Next we expand each factor in the product in terms of $\ad_A(B)=[A,B]$. For fixed $B$ and $0 \leq s \leq \beta/2$, observe that \cref{eq:deriv1} and \cref{eq:deriv2} give
\begin{align}
    \partial_s \cU_{\wh S}(0, s) B \cU_{\wh S}(s,0) = \cU_{\wh S}(0,s) [V_{\wh S}(s), B] \,\cU_{\wh S}(s,0)
\end{align}
and thus
\begin{align}
    \cU_{\wh S}(0,s)B\cU_{\wh S}(s,0) = B + \int_0^s \cU_{\wh S}(0,\tau)\ad_{V_{\wh S}(\tau)}(B)\,\cU_{\wh S}(\tau,0)\,d\tau
\end{align}
Iterating this argument gives
\begin{align}
    &\cU_{\wh S}(0,s)Q_{\mu_*}(s)\,\cU_{\wh S}(s,0) = \sum_{r\geq0} \int_{0\leq \tau_r\leq\cdots\leq\tau_1\leq s} \ad_{V_{\wh S}(\tau_r)} \cdots \ad_{V_{\wh S}(\tau_1)} \bigl(Q_{\mu_*}(s)\bigr) \,d\tau_1\cdots d\tau_r. \label{eq:graph-sample-conjugation-expansion}
\end{align}
Combining
\cref{eq:graph-sample-G-dyson},
\cref{eq:graph-sample-conjugation-expansion} and the Pauli expansion
\begin{align}
    Q_\mu(s) = \sum_{A\subseteq \cF_{\rm anti}(Q_\mu)} q_{\mu,A}(s)R_{\mu,A}
\end{align}
given in \Cref{lem:graph-sample-source-activity}, we obtain (after rearranging sums)
\begin{align}\label{eq:g1}
    G(\beta/2) &= I+ \sum_{m\geq1} \sum_{r_1,\dots,r_m\geq0} \sum_{A_1,\dots,A_m \subseteq \cF_{\rm anti}(Q_{\mu_*})} \int_{0\leq s_m\leq\cdots\leq s_1\leq \beta/2}\int_{0\leq \tau_{i,r_i}\leq\cdots\leq\tau_{i,1}\leq s_i} \notag \\
    &\quad\times (-v_{\mu_*})^m\lr{\prod_{i=1}^m \ad_{V_{\wh S}(\tau_{i,r_i})} \cdots \ad_{V_{\wh S}(\tau_{i,1})} \lr{q_{\mu_*,A_i}(s_i)R_{\mu_*,A_i}} \,d\tau_{i,1}\cdots d\tau_{i,r_i}}ds_1\cdots ds_m.
\end{align}
We continue expanding this. To make notation simpler, relabel the sum over $A_i \subseteq \cF_{\rm anti}(Q_{\mu_*})$ as a sum over $A_{i,0} = A_i$, and relabel $s_i$ as $\tau_{i,0}$. Since $V_{\wh S} = \sum_{\nu\in\wh S}v_\nu Q_\nu$, we can expand the nested commutators
\begin{align}
    &\ad_{V_{\wh S}(\tau_{i,r_i})} \cdots \ad_{V_{\wh S}(\tau_{i,1})}\lr{q_{\mu_*,A_{i,0}}(\tau_{i,0})R_{\mu_*,A_{i,0}}} \notag\\
    &= \sum_{\nu_{i,1},\cdots,\nu_{i,r_i} \in \wh S} \lr{\prod_{j=1}^{r_i} v_{\nu_{i,j}}} \ad_{Q_{\nu_{i,r_i}}(\tau_{i,r_i})} \cdots \ad_{Q_{\nu_{i,1}}(\tau_{i,1})}\lr{q_{\mu_*,A_{i,0}}(\tau_{i,0})R_{\mu_*,A_{i,0}}}.
\end{align}
We further decompose
\begin{align}
    Q_{\nu_{i,j}}(\tau_{i,j}) = \sum_{A_{i,j} \subseteq \cF_{\rm anti}(Q_{\nu_{i,j}})} q_{\nu_{i,j},A_{i,j}}(\tau_{i,j})R_{\nu_{i,j},A_{i,j}}
\end{align}
to obtain
\begin{align}
    &\ad_{V_{\wh S}(\tau_{i,r_i})} \cdots \ad_{V_{\wh S}(\tau_{i,1})}\lr{q_{\mu_*,A_{i,0}}(\tau_{i,0})R_{\mu_*,A_{i,0}}}\notag\\
    &= \sum_{\nu_{i,1},\cdots,\nu_{i,r_i} \in \wh S} \sum_{A_{i,1}\subseteq \cF_{\rm anti}(Q_{\nu_{i,1}})} \hspace{-0.5em}\cdots\hspace{-0.5em} \sum_{A_{i,r_i}\subseteq \cF_{\rm anti}(Q_{\nu_{i,r_i}})} \lr{\prod_{j=1}^{r_i} v_{\nu_{i,j}}}\lr{\prod_{j=0}^{r_i} q_{\nu_{i,j},A_{i,j}}(\tau_{i,j})}\notag\\
    &\qquad \times \ad_{R_{\nu_{i,r_i},A_{i,r_i}}}\hspace{-1em}\cdots\ad_{R_{\nu_{i,1},A_{i,1}}}\lr{R_{\mu_*,A_{i,0}}},
\end{align}
so \cref{eq:g1} becomes, also setting notation $\nu_{i,0} = \mu_*$ for all $i$,
\begin{align}\label{eq:g2}
    G(\beta/2) &= I+ \sum_{m\geq1} \sum_{r_1,\dots,r_m\geq0} \int_{0\leq \tau_{m,0}\leq\cdots\leq \tau_{1,0}\leq \beta/2}\int_{0\leq \tau_{i,r_i}\leq\cdots\leq\tau_{i,1}\leq \tau_{i,0}} \lr{\prod_{i=1}^m\prod_{j=0}^{r_i} d\tau_{i,j}}(-1)^m\notag \\
    &\qquad\times\prod_{i=1}^m\sum_{\nu_{i,1},\cdots,\nu_{i,r_i} \in \wh S} \sum_{A_{i,0}\subseteq \cF_{\rm anti}(Q_{\nu_{i,0}})} \cdots \sum_{A_{i,r_i}\subseteq \cF_{\rm anti}(Q_{\nu_{i,r_i}})} \lr{\prod_{j=0}^{r_i} v_{\nu_{i,j}} q_{\nu_{i,j},A_{i,j}}(\tau_{i,j})}\notag\\
    &\qquad\times \ad_{R_{\nu_{i,r_i},A_{i,r_i}}}\cdots\ad_{R_{\nu_{i,1},A_{i,1}}}\lr{R_{\nu_{i,0},A_{i,0}}}.
\end{align}
Observe that either
\begin{align}
    \ad_{R_{\nu_{i,r_i},A_{i,r_i}}}\cdots\ad_{R_{\nu_{i,1},A_{i,1}}}\lr{R_{\nu_{i,0},A_{i,0}}} = 0
\end{align}
or else
\begin{align}\label{eq:rprod}
    \ad_{R_{\nu_{i,r_i},A_{i,r_i}}}\cdots\ad_{R_{\nu_{i,1},A_{i,1}}}\lr{R_{\nu_{i,0},A_{i,0}}} = 2^{r_i} \prod_{j=r_i}^0 R_{\nu_{i,j},A_{i,j}}.
\end{align}
As shorthand, denote all the objects being summed over (e.g., $\tau_{i,j}, \nu_{i,j}, A_{i,j}$, etc.) by a history $h$. The non-vanishing contributions can be written as a phased Pauli string $E(h)$ given by
\begin{align}\label{eq:ec}
    E(h) = \frac{c(h)}{|c(h)|}\prod_{i=1}^m \prod_{j=r_i}^0 R_{\nu_{i,j}, A_{i,j}}, \quad c(h) = (-1)^m \prod_{i=1}^m 2^{r_i} \prod_{j=0}^{r_i} v_{\nu_{i,j}} q_{\nu_{i,j}, A_{i,j}}(\tau_{i,j}).
\end{align}
For a generic history, note that $c(h)$ vanishes unless $\ad_{R_{\nu_{i,r_i},A_{i,r_i}}} \cdots \ad_{R_{\nu_{i,1},A_{i,1}}} \lr{R_{\nu_{i,0},A_{i,0}}} \neq 0$ for all $i$.
We will later specify the coefficient $b$ and the probability distribution under which $\E{I+bE} = G(\beta/2)$.
For each non-vanishing contribution with history $h$, we record the labels by
\begin{align}\label{eq:dh}
    D(h) = (\nu_{1,0},\dots,\nu_{1,r_1},\nu_{2,0},\dots,\nu_{2,r_2},\dots,\nu_{m,0},\dots,\nu_{m,r_m}).
\end{align}
Since the total number of sets $A_{i,j}$ that appear is $t = \sum_{i=1}^m (r_i+1)$, the length of $D(h)$ is $t$, where $t \geq 1$ because the first entry is $\nu_{1,0} = \mu_*$. This shows the first property.

We now prove the second property. If $j=0$ then $\nu_{i,j} = \mu_*$ by definition. If $\nu_{i,j}$ with $j > 0$ appears in a nonzero nested commutator, then $R_{\nu_{i,j}, A_{i,j}}$ anticommutes with the Pauli string produced by the earlier factors by \cref{eq:rprod}. Hence, it anticommutes with at least one other $R_{\nu_{i,\ell},A_{i,\ell}}$ for $\ell < j$. That is, by \Cref{lem:graph-sample-source-activity},
\begin{align}
    \left\{\lr{\prod_{\sigma \in A_{i,j}} P_\sigma} Q_{\nu_{i,j}}, \lr{\prod_{\sigma \in A_{i,\ell}} P_\sigma} Q_{\nu_{i,\ell}}\right\} = 0.
\end{align}
Since the free Paulis $P_\sigma$ commute pairwise, this implies that either $Q_{\nu_{i,j}}$ and $Q_{\nu_{i,\ell}}$ anticommute (giving $Q_{\nu_{i,j}} \sim Q_{\nu_{i,\ell}}$), or that $Q_{\nu_{i,j}}$ anticommutes with some of the free Paulis in $A_{i,\ell}$, which implies that
\begin{align}
    \cF_{\rm anti}(Q_{\nu_{i,j}}) \cap \cF_{\rm anti}(Q_{\nu_{i,\ell}}) \neq \emptyset
\end{align}
since $A_{i,\ell}\subseteq\cF_{\rm anti}(Q_{\nu_{i,\ell}})$, and thus $Q_{\nu_{i,j}} \sim Q_{\nu_{i,\ell}}$. This shows the second property.

To show the third property, we observe by \Cref{lem:graph-sample-source-activity} that each term is, up to a phase, a product of perturbative Paulis $Q_{\nu_{i,j}}$ and free Paulis from the sets $\cF_{\rm anti}(Q_{\nu_{i,j}})$. The fourth property then follows directly from \Cref{lem:potential}.

It remains to show the fifth property and that $\E{I+bE} = G(\beta/2)$. We will show that the appropriate choice is
\begin{align}
    b(h) = \theta(3\Delta)^{|D(h)|} e^{-\beta \kappa(E(h))}
\end{align}
and
\begin{align}\label{eq:dp}
    dp(h) = \frac{|c(h)| e^{\beta \kappa(E(h))}}{\theta(3\Delta)^{|D(h)|}} d\pi(h),
\end{align}
where $\pi$ is the measure of the integrals and sums in \cref{eq:g2}, i.e., for test function $F(h)$ we define for $\nu_{i,0}=\mu_*$
\begin{align}
    \int_\cH F(h) d\pi(h) &= \sum_{m\geq1} \sum_{r_1,\ldots,r_m\geq0} \sum_{\substack{\nu_{i,j}\in\widehat S\\ 1\leq i\leq m,\ 1\leq j\leq r_i}} \sum_{\substack{A_{i,j}\subseteq\cF_{\rm anti}(Q_{\nu_{i,j}})\\ 1 \leq i \leq m, 0 \leq j \leq r_i}}\notag \\
    &\quad\times \int_{0\leq \tau_{m,0}\leq\cdots\leq \tau_{1,0}\leq \beta/2} \prod_{i=1}^m \left[ \int_{0\leq \tau_{i,r_i}\leq\cdots\leq\tau_{i,1}\leq\tau_{i,0}} \prod_{j=1}^{r_i}d\tau_{i,j} \right] \prod_{i=1}^m d\tau_{i,0}\, F(h).
\end{align}
Observe that this trivially satisfies the fifth property,
\begin{align}
    b(h) e^{\beta \kappa(E(h))} \leq \theta(3\Delta)^{|D(h)|}.
\end{align}
Moreover, by construction,
\begin{align}
    \E{bE} = \int_\cH b(h) E(h) dp(h) = \int_\cH |c(h)| E(h) d\pi(h) = G(\beta/2)-I,
\end{align}
i.e., $\E{I+bE} = G(\beta/2)$ as desired.

To complete the proof, it thus suffices to check that $\int dp \leq 1$; we can assign any remaining probability to $b=0, E=I, D=\emptyset$, which leaves $\E{I+bE} = G(\beta/2) = \cU_{\wh S}(0, \beta/2)\cU_S(\beta/2,0)$ unchanged. Accordingly, we bound
\begin{align}\label{eq:dpint}
    \int dp &= \sum_{m\geq1} \sum_{r_1,\ldots,r_m\geq0} \sum_{\substack{\nu_{i,j}\in\widehat S\\ 1\leq i\leq m,\ 1\leq j\leq r_i}} \sum_{\substack{A_{i,j}\subseteq\cF_{\rm anti}(Q_{\nu_{i,j}})\\ 1 \leq i \leq m, 0 \leq j \leq r_i}}\notag \\
    &\quad\times \int_{0\leq \tau_{m,0}\leq\cdots\leq \tau_{1,0}\leq \beta/2} \prod_{i=1}^m \left[ \int_{0\leq \tau_{i,r_i}\leq\cdots\leq\tau_{i,1}\leq\tau_{i,0}} \prod_{j=1}^{r_i}d\tau_{i,j} \right] \prod_{i=1}^m d\tau_{i,0} \frac{|c(h)| e^{\beta \kappa(E(h))}}{\theta(3\Delta)^{|D(h)|}}.
\end{align}
Using \cref{eq:ec}, applying the definition of \cref{eq:kappa} to $\kappa(E(h)) \leq \sum_{i=1}^m\sum_{j=0}^{r_i}\kappa(R_{\nu_{i,j},A_{i,j}})$, and $|D(h)| \leq \sum_{i=1}^m (r_i+1)$ by \cref{eq:dh}, we can bound the last term in \cref{eq:dpint} by
\begin{align}
    \frac{|c(h)| e^{\beta \kappa(E(h))}}{\theta(3\Delta)^{|D(h)|}} \leq \frac{1}{\theta} \prod_{i=1}^m \frac{2^{r_i}\prod_{j=0}^{r_i} |v_{\nu_{i,j}}|\,|q_{\nu_{i,j},A_{i,j}}(\tau_{i,j})|
    e^{\beta\kappa(R_{\nu_{i,j},A_{i,j}})}}{(3\Delta)^{r_i+1}},
\end{align}
for non-vanishing contributions $c(h) \neq 0$, where as before we let $\nu_{i,0} = \mu_*$ for all $i$. Since the RHS is nonnegative, we can bound $\int dp$ by dropping the ordering $0\leq \tau_{m,0}\leq\cdots\leq \tau_{1,0}\leq \beta/2$ to obtain a product of separately ordered $0 \leq \tau_r \leq \cdots \leq \tau_0 \leq \beta/2$, i.e.,
\begin{align}
    \int dp \leq  \frac{1}{\theta} \sum_{m \geq 1} \sum_{r_1,\dots,r_m \geq 0} \prod_{i=1}^m M_{r_i}
\end{align}
for
\begin{align}
    M_r &= \frac{1}{(3\Delta)^{r+1}} \sum_{\nu_1,\dots,\nu_r\in\wh S} \sum_{A_j \subseteq \cF_{\rm anti}(Q_{\nu_j}), 0 \leq j \leq r} \int_{0 \leq \tau_r \leq \cdots \leq \tau_0 \leq \beta/2} 2^r\prod_{j=0}^r |v_{\nu_j}|\,|q_{\nu_j,A_j}(\tau_j)|
    e^{\beta\kappa(R_{\nu_j,A_j})} d\tau_j\notag\\
    &\quad \times 1\{\ad_{R_{\nu_r,A_r}}\cdots \ad_{R_{\nu_1,A_1}}(R_{\nu_0,A_0}) \neq 0\},
\end{align}
where the indicator function checks if the history leads to a non-vanishing contribution, as discussed below \cref{eq:ec}. For the set $\cT_{r+1}$ of labeled trees on $\{0,\dots,r\}$ and edge set $E(T)$ for $T \in \cT_{r+1}$, we claim that $M_r$ can be bounded by
\begin{align}
    M_r &\leq \frac{1}{r!(3\Delta)^{r+1}} \sum_{T\in\cT_{r+1}} \sum_{\nu_1,\dots,\nu_r\in\wh S} \sum_{A_j \subseteq \cF_{\rm anti}(Q_{\nu_j}), 0 \leq j \leq r} \int_{[0,\beta/2]^{r+1}} 2^r \prod_{j=0}^r |v_{\nu_j}|\,|q_{\nu_j,A_j}(\tau_j)|
    e^{\beta\kappa(R_{\nu_j,A_j})} \notag\\
    &\quad \times\prod_{\{p,q\}\in E(T)} 1\{\nu_p = \nu_q \text{ or } \nu_p \sim \nu_q\} \prod_{j=0}^r d\tau_j \label{eq:mr1}\\
    &\leq \frac{1}{r!(3\Delta)^{r+1}} \sum_{T\in\cT_{r+1}} \frac{a_{\mu_*}}{2}\Delta^r\\
    &\leq  \frac{a_{\mu_*}}{2(3\Delta)^{r+1}} \frac{(r+1)^{r-1}}{r!} \Delta^r, \label{eq:mr3}
\end{align}
where the third inequality follows from Cayley's formula, and the second inequality follows from traversing the tree rooted at 0, applying to each parent label $\mu \in S$ \Cref{lem:graph-sample-source-activity} and the definition of $\Delta$ to 
\begin{align}
    \sum_{\substack{\nu\in\wh S\\ \nu=\mu\text{ or }\nu\sim\mu}} \sum_{A\subseteq \cF_{\rm anti}(Q_\nu)}\int_0^{\beta/2} 2|v_\nu|\,|q_{\nu,A}(\tau)|
    e^{\beta\kappa(R_{\nu,A})} d\tau = \sum_{\substack{\nu\in\wh S\\ \nu=\mu\text{ or }\nu\sim\mu}} a_\nu \leq a_\mu + \sum_{\substack{\nu\in\cG\setminus\{\mu\}\\\nu\sim\mu}} a_\nu \leq \Delta
\end{align}
for all $r$ non-root vertices, and applying \Cref{lem:graph-sample-source-activity} to the root at $\nu_0 = \mu_*$ to obtain
\begin{align}
    \sum_{A_0\subseteq \cF_{\rm anti}(Q_{\mu_*})} \int_0^{\beta/2} |v_{\mu_*}|\,|q_{\mu_*,A_0}(\tau)|
    e^{\beta\kappa(R_{\mu_*,A_0})} d\tau = \frac{a_{\mu_*}}{2}.
\end{align}
To show the first inequality \cref{eq:mr1}, we observe that $\ad_{R_{\nu_r,A_r}}\cdots \ad_{R_{\nu_1,A_1}}(R_{\nu_0,A_0}) \neq 0$ implies that every partial nested commutator is nonzero, and thus for every $j \geq 1$, the Pauli $R_{\nu_j, A_j}$ must anticommute with some $R_{\nu_\ell, A_\ell}$ with $\ell < j$. Decomposing $R_{\nu_j, A_j} = \prod_{\sigma \in A_j} P_{\sigma} Q_{\nu_j}$ and similarly for $\ell$, this anticommutation must originate from either $\{Q_{\nu_j}, Q_{\nu_\ell}\} = 0$, or anticommutation between $Q_{\nu_j}$ and some $P_\sigma$ for $\sigma \in A_\ell$ (or vice versa with $j \leftrightarrow \ell$). This ensures that either $\nu_j = \nu_\ell$ or $\nu_j \sim \nu_\ell$, implying that the graph on $\{0,\dots,r\}$ with an edge $\{j,\ell\}$ whenever $\nu_j = \nu_\ell$ or $\nu_j \sim \nu_\ell$ is connected. We sum over spanning trees to obtain
\begin{align}
    1\{\ad_{R_{\nu_r,A_r}}\cdots \ad_{R_{\nu_1,A_1}}(R_{\nu_0,A_0}) \neq 0\} \leq \sum_{T\in\cT_{r+1}}\prod_{\{p,q\}\in E(T)} 1\{\nu_p = \nu_q \text{ or } \nu_p \sim \nu_q\} ,
\end{align}
which gives \cref{eq:mr1} after introducing $1/r!$ for the volume of the integral, which we enlarged from $[0,\tau_0]^r$ to $[0,\beta/2]^r$.

Given \cref{eq:mr3}, we can finally show $\int dp \leq 1$. Using $a_{\mu_*} = \theta \Delta$
\begin{align}
    \sum_{r \geq 0} \frac{a_{\mu_*}}{2(3\Delta)^{r+1}} \frac{(r+1)^{r-1}}{r!} \Delta^r \leq \frac{\theta}{6} \sum_{r \geq 0} \frac{(r+1)^{r-1}}{r! 3^r} = \frac{\theta T(1/3)}{2} < \frac{\theta}{2}
\end{align}
by \Cref{fac:tfunc}. Hence,
\begin{align}
    \int dp &\leq \frac{1}{\theta} \sum_{m\geq1} \sum_{r_1,\ldots,r_m\geq0} \prod_{i=1}^m \left[ \frac{a_{\mu_*}}{2(3\Delta)^{r_i+1}}  \frac{(r_i+1)^{r_i-1}}{r_i!} \Delta^{r_i} \right]\\
    &= \frac{1}{\theta} \sum_{m\geq1} \left[ \sum_{r\geq0} \frac{a_{\mu_*}}{2(3\Delta)^{r+1}} \frac{(r+1)^{r-1}}{r!} \Delta^r \right]^m\\
    &< \frac{1}{\theta} \sum_{m\geq1} \left(\frac{\theta}{2}\right)^m\\
    &= \frac{1}{2(1-\theta/2)} \leq 1\,,
\end{align}
where the last inequality uses $0 < \theta \leq 1$.
\end{proof}

\subsection{Pinning procedure}
We give the main procedure for obtaining the stabilizer expansion.

\begin{lemma}[Pinning procedure]
\label{lem:stab-lower}
Let
\begin{align}
    H=H_0+V = \sum_{\sigma\in\cF}u_\sigma P_\sigma + \sum_{\mu\in\cG}v_\mu Q_\mu,
\end{align}
where the $P_\sigma$ commute pairwise. If
\begin{align}
    \Delta(\beta)\leq \frac1{72},
\end{align}
then
\begin{align}
    \rho_\beta(H)\in\STAB_n.
\end{align}
\end{lemma}
\begin{proof}
It is enough to show that $e^{-\beta H}\in\StabCone_n$. If $\beta=0$,
then $e^{-\beta H}=I\in\StabCone_n$. Hence assume $\beta>0$. If
$\Delta(\beta)=0$, then $\cG=\emptyset$, so $V=0$.
In that case
\begin{align}
    e^{-\beta H_0} = \prod_{\sigma\in\cF} \cosh(\beta u_\sigma)\lr{I-\tanh(\beta u_\sigma)P_\sigma} \in\StabCone_n
\end{align}
by \Cref{lem:stabcrit}. Thus assume
$0<\Delta(\beta)\leq1/72$.

We define a randomized pinning procedure. At every stage the procedure
keeps a subset $S\subseteq\cG$ and a finite list
\begin{align}
    \chi=((\lambda_1,X_1),\ldots,(\lambda_m,X_m)),
\end{align}
where each $X_j$ is either $0$ or a signed Hermitian Pauli string. The
procedure keeps invariant the following properties:
\begin{enumerate}
    \item the thermal state (unnormalized),
        \begin{align}
        e^{-\beta H} = \E{ e^{-\beta H_0/2} \cU_S(\beta/2,0) \prod_{j=1}^m\lr{I+\lambda_jX_j} \cU_S(0,-\beta/2) e^{-\beta H_0/2} }, \label{eq:pinning-operator-invariant}
    \end{align}
    \item either $X_j = 0$ or
    \begin{align}
        X_j\in \pm\langle P_\sigma,Q_\mu:\sigma\in\cF,\mu\in\cG\rangle, \label{eq:pinning-generated-invariant}
    \end{align}
    \item whenever $X_i, X_j$ are both nonzero and $i \neq j$,
    \begin{align}
        X_i\not\sim X_j, \label{eq:pinning-comm-invariant}
    \end{align}
    \item all but possibly the last factor are incompatible with every unpinned perturbing label,
    \begin{align}
        \Phi_S(X_j)=0 \qquad \lr{j<m}, \label{eq:pinning-finished-invariant}
    \end{align}
    \item every coefficient satisfies
    \begin{align}
        |\lambda_j|\leq 2^{-\Phi_S(X_j)}e^{-\beta\kappa(X_j)} \qquad \lr{1\leq j\leq m}. \label{eq:pinning-coeff-invariant}
    \end{align}
\end{enumerate}
Assume that $S\ne\emptyset$. If $m\geq1$ and
$\Phi_S(X_m)>0$, choose a label $\mu_*\in S$ such that
\begin{align}
    \mu_*\sim X_m.
\end{align}
Otherwise append the inactive factor $(0,I)$ to $\chi$, relabel it as
the last factor, and choose any $\mu_*\in S$.

Set
\begin{align}
    \wh S=S\setminus\set{\mu_*}, \qquad \theta=\frac{|v_{\mu_*}|\,\eta_\beta(\kappa(Q_{\mu_*}))}{\Delta(\beta)}.
\end{align}
Using \Cref{lem:graph-sample}, independently sample
$(b_1,E_1,D_1)$ and $(b_2,E_2,D_2)$ such that
\begin{align}\label{eq:expbe}
    \E{I+b_iE_i} = \cU_{\wh S}(0,\beta/2)\cU_S(\beta/2,0).
\end{align}
Whenever $b_i\ne0$ and $D_i$ has length $t_i$, the sampling lemma gives
\begin{align}
    b_i e^{\beta\kappa(E_i)} \leq \theta\lr{3\Delta(\beta)}^{t_i}, \qquad \Phi_S(E_i)\leq t_i. \label{eq:sampled-propagator-bounds}
\end{align}
Replace the last factor $(\lambda_m,X_m)$ by a new factor $(\wh\lambda,\wh X)$ chosen from the following seven branches. For convenience, we write $\lambda = \lambda_m$ and $X = X_m$.
Branch $J=1$ is chosen with probability $2^{-\theta}$, and each branch
$J=2,\ldots,7$ is chosen with probability $(1-2^{-\theta})/6$:
\begin{align}
    \begin{array}{c|c|c} J&\wh\lambda&\wh X\\ \hline 1&\lambda/2^{-\theta}&X\\
    2&\dfrac6{1-2^{-\theta}}b_1&\dfrac{E_1+E_1^\dagger}{2}\\
    3&\dfrac6{1-2^{-\theta}}b_2&\dfrac{E_2+E_2^\dagger}{2}\\
    4&\dfrac6{1-2^{-\theta}}b_1\lambda&\dfrac{E_1X+XE_1^\dagger}{2}\\
    5&\dfrac6{1-2^{-\theta}}b_2\lambda&\dfrac{E_2X+XE_2^\dagger}{2}\\
    6&\dfrac6{1-2^{-\theta}}b_1b_2&\dfrac{E_1E_2^\dagger+E_2E_1^\dagger}{2}\\
    7&\dfrac6{1-2^{-\theta}}b_1b_2\lambda& \dfrac{E_1XE_2^\dagger+E_2XE_1^\dagger}{2}. \end{array}
\end{align}

\noindent\textbf{First invariant} \cref{eq:pinning-operator-invariant}.
Averaging only over the branch
$J$ gives
\begin{align}
    \EE_J\!\lr{I+\wh\lambda\wh X} = \frac12 \Big[ (I+b_1E_1)(I+\lambda X)(I+b_2E_2)^\dagger + (I+b_2E_2)(I+\lambda X)(I+b_1E_1)^\dagger \Big].
\end{align}
Averaging over the two independent samples gives by \cref{eq:expbe}
\begin{align}
    \E{I+\wh\lambda\wh X} = R(I+\lambda X)R^\dagger, \qquad R= \cU_{\wh S}(0,\beta/2)\cU_S(\beta/2,0).
\end{align}
Let
\begin{align}
    A=\prod_{j=1}^{m-1}(I+\lambda_jX_j)
\end{align}
be the product of the previously finished factors. By
\cref{eq:pinning-finished-invariant}, every $X_j$ with $j<m$ is
incompatible with every perturbing label in $S$, hence commutes with $R$
by \Cref{lem:propagator-identities}. By
\cref{eq:pinning-comm-invariant}, $A$ also commutes with $X$. By
\Cref{lem:propagator-identities},
\begin{align}
    \cU_S(\beta/2,0) = \cU_{\wh S}(\beta/2,0)R, \qquad \cU_S(0,-\beta/2) = R^\dagger\cU_{\wh S}(0,-\beta/2),
\end{align}
so
\begin{align}
    & e^{-\beta H_0/2} \cU_{\wh S}(\beta/2,0) A\,\E{I+\wh\lambda\wh X}\, \cU_{\wh S}(0,-\beta/2) e^{-\beta H_0/2}  \notag\\
    &= e^{-\beta H_0/2} \cU_S(\beta/2,0) A(I+\lambda X) \,\cU_S(0,-\beta/2) e^{-\beta H_0/2}.
\end{align}
Thus invariant~\cref{eq:pinning-operator-invariant} is preserved after
replacing $S$ by $\wh S$.

\noindent\textbf{Second invariant} \cref{eq:pinning-generated-invariant}. \Cref{lem:graph-sample} implies that, up to an overall phase,
\begin{align}\label{eq:eiprod}
E_i \in \left\langle \bigcup_{a=1}^{t_i} \left( \{Q_{\mu_a^{(i)}}\} \cup \{P_\sigma: \sigma\in\cF_{\rm anti}(Q_{\mu_a^{(i)}})\} \right) \right\rangle,
\end{align}
where $D_i=(\mu_1^{(i)},\ldots,\mu_{t_i}^{(i)})$ is the list supplied by \Cref{lem:graph-sample}.

By the inductive hypothesis, the trivial branch $\wh X = X$ satisfies \cref{eq:pinning-generated-invariant}, i.e., $X \in \pm\langle P_\sigma,Q_\mu:\sigma\in\cF,\mu\in\cG\rangle$. Observe that each of the six nontrivial branches has the form $(Y + Y^\dagger)/2$, where $Y$ is a phased Pauli string that is a product of $X, E_i, E_i^\dagger$. Since $E_i, E_i^\dagger \in \pm\langle P_\sigma,Q_\mu:\sigma\in\cF,\mu\in\cG\rangle$ by \cref{eq:eiprod}, writing $Y = \omega R$ for an unphased Pauli $R$, we have
\begin{align}\label{eq:yr}
    \frac{Y+Y^\dagger}{2} \in \left\{0, +R, -R\right\}.
\end{align}
This shows that $\wh X \in \pm\langle P_\sigma,Q_\mu:\sigma\in\cF,\mu\in\cG\rangle$, i.e., that \cref{eq:pinning-generated-invariant} remains invariant.

\noindent\textbf{Third invariant} \cref{eq:pinning-comm-invariant}. Fix a nonzero old factor $X_j$ with $j < m$. To show $\wh X \not\sim X_j$, we will first show that $X_j \not\sim X$, $X_j\not\sim E_1$ and $X_j\not\sim E_2$. The first claim is immediate: $X_j\not\sim X$ either by the induction
hypothesis \cref{eq:pinning-comm-invariant} or trivially since $X = I$. The second claim, $X_j\not\sim E_i$, is only trivial when $b_i=0$, in which case \Cref{lem:graph-sample} allows us to take $E_i=I$. In the nontrivial case $b_i \neq 0$, it will be useful to use the fact that $\Phi_S(X_j)=0$ implies that $X_j \not\sim Q_\nu$ for all $\nu \in S$. Equivalently, for every $\nu \in S$,
\begin{align}
    [X_j,Q_\nu]=0, \qquad \cF_{\rm anti}(X_j)\cap \cF_{\rm anti}(Q_\nu)=\emptyset. \label{eq:old-factor-incompatible-with-sampled-labels}
\end{align}
We now establish $X_j \not\sim E_i$ by showing that $[X_j, E_i] = 0$ and then that $\cF_{\rm anti}(X_j)\cap \cF_{\rm anti}(E_i)=\emptyset$.
\begin{itemize}
    \item To show that $[X_j, E_i] = 0$, it suffices to show that $X_j$ commutes with every Pauli factor appearing in the product description of $E_i$.
    Since each $\mu^{(i)}_a$ lies in $S$ by \Cref{lem:graph-sample}, we have that $[X_j, Q_{\mu_a^{(i)}}] = 0$ by \cref{eq:old-factor-incompatible-with-sampled-labels}. Similarly, by \cref{eq:old-factor-incompatible-with-sampled-labels}, since $\cF_{\rm anti}(X_j) \cap \cF_{\rm anti}(Q_{\mu_a^{(i)}}) = \emptyset$, we have that $X_j$ commutes with every free Pauli $P_\sigma$ satisfying $\sigma \in \cF_{\rm anti}(Q_{\mu_a^{(i)}})$. Hence, $[X_j, E_i] = 0$.
    \item To show that $\cF_{\rm anti}(X_j)\cap \cF_{\rm anti}(E_i)=\emptyset$, we use the property
    \begin{align}
        \cF_{\rm anti}(W_1 W_2) \subseteq \cF_{\rm anti}(W_1) \cup \cF_{\rm anti}(W_2)
    \end{align}
    for Pauli strings $W_1, W_2$. This gives
    \begin{align}
        \cF_{\rm anti}(E_i) \subseteq \bigcup_{a=1}^{t_i}\cF_{\rm anti}(Q_{\mu^{(i)}_a}),
    \end{align}
    which is disjoint from $\cF_{\rm anti}(X_j)$ by \cref{eq:old-factor-incompatible-with-sampled-labels}.
\end{itemize}
We conclude that $X_j \not\sim E_i$; the same holds for $E_i^\dagger$ since the adjoint only changes the phase. Hence, we find that $X_j$ is incompatible with every Pauli product that appears in the seven possible branch definitions of $\widehat X$. The first branch is trivial, since we already established $X \not\sim X_j$. To show $X_j \not\sim \wh X$ in the six nontrivial branches, we write each branch as $(Y+Y^\dagger)/2$ and reuse \cref{eq:yr}: since we already have $X_j \not\sim Y$ and the compatibility relation is independent of phase, we conclude that $X_j \not\sim \wh X$, completing the proof of \cref{eq:pinning-comm-invariant}.

\noindent\textbf{Fourth invariant} \cref{eq:pinning-finished-invariant}. Since all coefficients of $\Phi$ are nonnegative, this follows directly from the fact that $\wh S$ is smaller than $S$, i.e.,
\begin{align}
    \Phi_{\widehat S}(X_j) = \Delta(\beta)^{-1} \sum_{\nu\in\widehat S:\,\nu\sim X_j} |v_\nu|\,\eta_\beta(\kappa(Q_\nu)) \leq \Delta(\beta)^{-1} \sum_{\nu\in S:\,\nu\sim X_j} |v_\nu|\,\eta_\beta(\kappa(Q_\nu)) = \Phi_S(X_j) = 0.
\end{align}

\noindent\textbf{Fifth invariant} \cref{eq:pinning-coeff-invariant}.  Since
$\Delta(\beta)\leq1/72$ and $0<\theta\leq1$, for every integer
$T\geq1$,
\begin{align}
    \frac{6\theta}{1-2^{-\theta}} \lr{3\Delta(\beta)}^T \leq 12\lr{\frac{1}{24}}^T \leq 2^{-T}. \label{eq:theta-decay-bound}
\end{align}
We now check case by case that
\begin{align}
    |\wh\lambda| \leq 2^{-\Phi_{\wh S}(\wh X)} e^{-\beta\kappa(\wh X)}
\end{align}
regardless of which branch $J$ was selected.

\begin{itemize}
    \item Suppose first that branch $1$ occurs. If the last factor was appended
    in the inactive case, then $\lambda=0$ and the bound is trivial. If the
    last factor was active, then $\mu_*\sim X$, so
    \begin{align}
        \Phi_{\wh S}(X) = \Phi_S(X)-\theta.
    \end{align}
    Using the induction hypothesis,
    \begin{align}
        |\wh\lambda| = \frac{|\lambda|}{2^{-\theta}} \leq 2^{-\Phi_S(X)+\theta}e^{-\beta\kappa(X)} = 2^{-\Phi_{\wh S}(X)}e^{-\beta\kappa(X)}.
    \end{align}

    \item Now consider branches $2$ and $3$. If $b_i=0$, the coefficient bound is
    trivial. Otherwise, if $D_i$ has length $t_i$, then
    \cref{eq:sampled-propagator-bounds} and \cref{eq:theta-decay-bound} give
    \begin{align}
        |\wh\lambda| = \frac6{1-2^{-\theta}}\,b_i \leq 2^{-t_i}e^{-\beta\kappa(E_i)}.
    \end{align}
    By \Cref{lem:potential},
    \begin{align}
        \Phi_{\wh S}(\wh X) \leq \Phi_S(E_i) \leq t_i, \qquad \kappa(\wh X)\leq\kappa(E_i).
    \end{align}
    Therefore
    \begin{align}
        |\wh\lambda| \leq 2^{-\Phi_{\wh S}(\wh X)} e^{-\beta\kappa(\wh X)}.
    \end{align}

    \item Branches $4$ and $5$ are similar but include the old factor $X$. Again
    the claim is trivial if $b_i=0$ or $\lambda=0$. Otherwise, by
    \cref{eq:sampled-propagator-bounds}, \cref{eq:theta-decay-bound}, and the
    induction hypothesis,
    \begin{align}
        |\wh\lambda| = \frac6{1-2^{-\theta}}b_i|\lambda| \leq 2^{-t_i-\Phi_S(X)} e^{-\beta(\kappa(E_i)+\kappa(X))}.
    \end{align}
    \Cref{lem:potential} gives
    \begin{align}
        \Phi_{\wh S}(\wh X) \leq \Phi_{\wh S}(E_i)+\Phi_{\wh S}(X) \leq t_i+\Phi_S(X), \qquad \kappa(\wh X)\leq\kappa(E_i)+\kappa(X).
    \end{align}
    Therefore
    \begin{align}
        |\wh\lambda| \leq 2^{-\Phi_{\wh S}(\wh X)} e^{-\beta\kappa(\wh X)}.
    \end{align}

    \item For branch $6$, if either $b_1=0$ or $b_2=0$, the claim is trivial.
    Otherwise let $T=t_1+t_2$. Since $\theta\leq1$,
    \cref{eq:sampled-propagator-bounds} gives
    \begin{align}
        b_1b_2 e^{\beta(\kappa(E_1)+\kappa(E_2))} \leq \theta^2\lr{3\Delta(\beta)}^T \leq \theta\lr{3\Delta(\beta)}^T.
    \end{align}
    Using \cref{eq:theta-decay-bound},
    \begin{align}
        |\wh\lambda| = \frac6{1-2^{-\theta}}b_1b_2 \leq 2^{-T}e^{-\beta(\kappa(E_1)+\kappa(E_2))}.
    \end{align}
    Also by \Cref{lem:potential},
    \begin{align}
        \Phi_{\wh S}(\wh X) \leq \Phi_{\wh S}(E_1)+\Phi_{\wh S}(E_2) \leq t_1+t_2 = T, \qquad \kappa(\wh X)\leq\kappa(E_1)+\kappa(E_2).
    \end{align}
    Thus
    \begin{align}
        |\wh\lambda| \leq 2^{-\Phi_{\wh S}(\wh X)} e^{-\beta\kappa(\wh X)}.
    \end{align}

    \item Finally consider branch $7$, which proceeds similarly to the previous
    case. The claim is trivial if $b_1b_2\lambda=0$. Otherwise, with
    $T=t_1+t_2$, \cref{eq:sampled-propagator-bounds},
    \cref{eq:theta-decay-bound}, and the induction hypothesis give
    \begin{align}
        |\wh\lambda| = \frac6{1-2^{-\theta}}b_1b_2|\lambda| \leq 2^{-T-\Phi_S(X)} e^{-\beta(\kappa(E_1)+\kappa(X)+\kappa(E_2))}.
    \end{align}
    By \Cref{lem:potential},
    \begin{align}
        \Phi_{\wh S}(\wh X) \leq \Phi_{\wh S}(E_1) +\Phi_{\wh S}(X) +\Phi_{\wh S}(E_2) \leq t_1+\Phi_S(X)+t_2 = T+\Phi_S(X),
    \end{align}
    and
    \begin{align}
        \kappa(\wh X) \leq \kappa(E_1)+\kappa(X)+\kappa(E_2).
    \end{align}
    Thus again
    \begin{align}
        |\wh\lambda| \leq 2^{-\Phi_{\wh S}(\wh X)} e^{-\beta\kappa(\wh X)}.
    \end{align}
\end{itemize}
This proves invariant~\cref{eq:pinning-coeff-invariant}.

Having shown all invariants are maintained, we can now finish the proof. Each step removes exactly one label from $S$, so the procedure terminates
after at most $|\cG|$ steps. At termination $S=\emptyset$, and
invariant~\cref{eq:pinning-operator-invariant} gives
\begin{align}
    e^{-\beta H} = \E{ e^{-\beta H_0/2} \prod_{j=1}^m\lr{I+\lambda_jX_j} e^{-\beta H_0/2} }.
\end{align}
Since $S=\emptyset$, invariant~\cref{eq:pinning-coeff-invariant} gives
\begin{align}
    |\lambda_j|\leq e^{-\beta\kappa(X_j)}
\end{align}
for every final factor. By invariant~\cref{eq:pinning-comm-invariant},
the final nonzero $X_j$ are pairwise incompatible. Hence
\Cref{lem:free-slack} implies that every final state lies in
$\StabCone_n$. The number of pinning steps is finite. The final random operator is
integrable and always lies in the closed convex cone $\StabCone_n$, so
its expectation also lies in $\StabCone_n$. Therefore
\begin{align}
    e^{-\beta H}\in\StabCone_n.
\end{align}
After dividing by $\Tr(e^{-\beta H})>0$, we obtain
$\rho_\beta(H)\in\STAB_n$.
\end{proof}

We can now prove the main result of this section.
\begin{proof}[Proof of \Cref{thm:stab}]
Since $\Delta(\beta) \leq w_{\rm pert} \eta_\beta(w_{\rm free})$, we can define
\begin{align}
    \beta_{\rm stab}(w_{\rm pert}, w_{\rm free}) = \sup\left\{\beta \geq 0 \,:\, w_{\rm pert} \eta_\beta(w_{\rm free}) \leq \frac{1}{72}\right\}
\end{align}
and apply \Cref{lem:stab-lower} to obtain $\rho_\beta(H) \in \STAB_n$ for all $\beta \leq \beta_{\rm stab}$. Simplifying gives
\begin{align}
    \beta_{\rm stab}(w_{\rm pert}, w_{\rm free}) \geq \begin{cases} +\infty,&w_{\rm pert}=0,\\[0.8em]
    \displaystyle \frac{1}{72w_{\rm pert}}, &w_{\rm pert}>0,\ w_{\rm free}=0,\\[1.1em]
    \displaystyle \frac1{w_{\rm free}} \log\lr{ \frac{1+\sqrt{1+w_{\rm free}/(18w_{\rm pert})}}{2} }, &w_{\rm pert},w_{\rm free}>0. \end{cases}
\end{align}
In particular, when $0 < w_{\rm pert} \ll w_{\rm free}$, this gives $\beta_{\rm stab} = \Omega\lr{\frac{1}{w_{\rm free}} \log \frac{w_{\rm free}}{w_{\rm pert}}}$ and thus when $H$ is $\epsilon$-close to commuting, $\beta_{\rm stab} = \Omega\lr{\log(1/\epsilon)/sk}$.
Note that we can apply this to any Hamiltonian by setting $H_0=0$ and $V=H$ so $w_{\rm free}=0$ and $\eta_\beta(w_{\rm free})=\beta$.
Moreover,
\begin{align}
    w_{\rm pert} &\leq \max_{a\in\cA}\lr{ |c_a|+ \sum_{b\ne a:\,P_bP_a=-P_aP_b}|c_b|} \leq \max_{a\in\cA} \sum_{b:\,\supp(P_b)\cap \supp(P_a)\ne\emptyset}|c_b|  \\
    &\leq \max_{a\in\cA} \sum_{x\in \supp(P_a)}\sum_{b:\,x\in \supp(P_b)}|c_b| \leq \max_{a\in\cA}|\supp(P_a)|s \leq ks.
\end{align}
\end{proof}

\section{Infinite-temperature phase}
\label{sec:expect}
In \Cref{subsec:zf}, we will show the lower bound on the zero-free disk radius $\beta_{\rm phase}$ reported in \Cref{thm:zf}; the upper bound is given in \Cref{thm:zf-upper} (\Cref{sec:upper}). We will also show \Cref{thm:decay} in \Cref{subsec:correlation-decay} and \Cref{thm:expect} in \Cref{subsec:alg}. These results are all based on a polymer representation of the partition function (or of $\Tr(O e^{-\beta H})$) constructed similarly to prior work~\cite{hastings2006solving,kliesch2014locality,harrow2020classical,mann2021efficient}. Expanding the normalized partition function
\begin{align}
    \tr(e^{-\beta H})
    =
    \sum_{m\geq0}\frac{(-\beta)^m}{m!}
    \sum_{a_1,\ldots,a_m}
    \left(\prod_{j=1}^m c_{a_j}\right)
    \tr(P_{a_1}\cdots P_{a_m}),
\end{align}
we identify the connected components of the overlap graph of the supports of $P_{a_1},\cdots, P_{a_m}$, which decomposes $\Tr(P_{a_1}\cdots P_{a_m})$ into a product of traces. Each connected component is a polymer (\Cref{lem:polymer-rep}) and thus $\log Z$ can be controlled by the Kotecky-Preiss criterion~\cite{kotecky1986cluster}.

In the na\"{i}ve argument, one applies Kotecky-Preiss by adding Hamiltonian terms successively to construct polymers: whenever a new Hamiltonian term is attached, it can overlap with one of at most $k$ sites in the current term and costs one interaction incident to that site. For a $d$-degree Hamiltonian, this bounds the order-$m$ polymer mass by $(C|\beta|dk)^m$ and converges for $\beta \lesssim 1/(dk)$. In the long-range setting, the interactions instead can sum up to $s$, giving a polymer mass of $(C|\beta|sk)^m$ and convergence for $\beta \lesssim 1/(sk)$.

We improve upon this na\"{i}ve argument by observing that $\tr(P_{a_1}\cdots P_{a_m}) = 0$ unless $P_{a_1}\cdots P_{a_m}$ is proportional to the identity on every qubit. During the construction of a polymer, call a qubit \emph{defective} when the current product acts on it by a nonidentity Pauli. Whenever a defect is present, we can force the next Hamiltonian term to fix the defective qubit rather than allowing it to overlap on any of the $k$ sites. Intuitively, if one thinks of $s$ as the degree of the interaction, there are at most $s$ choices of the Hamiltonian terms to repair the qubit. In the special case where repairing the defective qubit also repairs \emph{all other} defective qubits, we don't have $s$ choices: there must be a unique Hamiltonian term. After this repair, there are no longer any defects, so we append a new Hamiltonian term that \emph{must} create at least one defect. Hence, defects occur in at least half the steps (i.e., $m/2$). For $b \leq m/2$ counting the number of times there are no defects, this argument bounds the polymer mass to be at most $s^{m-2b} (sk)^b$, which ranges between $s$ (when $b=0$) and $\sqrt{sk}$ (when $b=m/2$).

Once we show that the polymer representation of the partition function converges with Kotecky-Preiss, we can conclude that truncating it after $O(\log n/\epsilon)$ terms results in an additive $\epsilon$-approximation to $\log Z$. For geometrically local Hamiltonians, this truncation also ensures that local observables are independent from distant terms in the Hamiltonian, leading to correlation decay. To algorithmically estimate the partition function, enumerating all these contributing polymers requires adding $n^{O(\log n/\epsilon)}$ terms. To turn this into a polynomial-time algorithm, we construct a random variable that instead samples polymers. This proceeds by writing a transcript that corresponds to the aforementioned procedure for choosing Hamiltonian terms and sites at which their overlaps must intersect (\Cref{lem:polymer-samp}). By generating random transcripts, we can generate random polymers. We show this results in an unbiased estimator of the $O(\log n/\epsilon)$-truncated partition function with variance controlled by Hoeffding's inequality.

\subsection{Preliminaries}
In addition to an $(s,k)$-long-range Pauli Hamiltonian $H = \sum_{a \in \cA} c_a P_a$ with $\max_{a\in\cA} |c_a| = J$, we consider in this section an observable
\begin{align}
    O = \sum_{b \in \cA_O} b_b Q_b, \quad B_O = \sum_{b\in\cA_O}|b_b|, \quad |\supp(Q_b)| \leq k_O, \qquad k_* = \max\{k,k_O\},
\end{align}
where we take all $Q_b$ to be nonidentity. We denote the unnormalized partition function by $Z(\beta) = \Tr(e^{-\beta H})$. We will also use the notation, for qubit $x \in [n]$, that
\begin{align}\label{eq:axsx}
    \cA(x) = \{a\in\cA:x\in\supp(P_a)\}, \quad S(x)=\sum_{a\in\cA(x)}|c_a| \leq s.
\end{align}
We also require some elementary properties of graphs. For a graph $G=(V,E)$, set
\begin{align}
    \varphi(G) = \sum_{F\subseteq E : (V,F) \text{ connected}} (-1)^{|F|}, \qquad \tau(G) = \#\text{ of spanning trees of }G
\end{align}
with the convention $\varphi(G) = \tau(G) = 1$ when $|V|=1$.
\begin{fact}\label{fac:graph}
The following properties hold.
    \begin{itemize}
        \item If $G$ is disconnected and has $|V| \geq 2$, $\varphi(G) = 0$.
        \item $|\varphi(G)| \leq \tau(G)$.
        \item $\varphi(G)$ can be evaluated with $O(3^{|V|})$ operations.
        \item $\tau(G)$ can be evaluated with $O(|V|^3)$ operations.
    \end{itemize}
\end{fact}
\begin{proof}
    The first claim immediately follows from the definition of $\varphi$: $F\subseteq E$ deletes edges and cannot make disconnected $G$ connected.
    
    To show the second claim, it is enough to consider connected $G$. Kruskal's algorithm assigns to every connected spanning subgraph $F\subseteq G$ its minimal spanning tree $T(F)$. For a spanning tree
    $T\subseteq G$, let $R(T)$ be obtained from $T$ by adding edges $e\in E(G)\setminus T$ whose endpoints are connected in $T$ by edges strictly smaller than $e$ (given some fixed ordering on $E$). The connected spanning subgraphs of $G$ are partitioned into the intervals $\{F:T\subseteq F\subseteq R(T)\}$ as $T$ ranges over the
    spanning trees of $G$. Hence
    \begin{align}
        \varphi(G) = \sum_{T\subseteq G\text{ tree}} \sum_{F:T\subseteq F\subseteq R(T)}(-1)^{|F|} = \sum_{T\subseteq G\text{ tree}}(-1)^{|T|} (1-1)^{|E(R(T))|-|E(T)|}.
    \end{align}
    Each inner contribution is $0$ or $\pm1$, so $|\varphi(G)|\leq \tau(G)$.

    The third claim follows from evaluating the following recurrence over all subsets $S \subseteq V$ with lexicographically smallest vertex $r$
    \begin{align}
        \varphi(G[S]) = \prod_{\{i,j\}\subseteq S} \lr{1-\mathbf 1_{\{i,j\}\in E(G)}} - \sum_{A \subsetneq S:r\in A} \varphi(G[A]) \prod_{\{i,j\}\subseteq S \setminus A} \lr{1-\mathbf 1_{\{i,j\}\in E(G)}}
    \end{align}
    where $\varphi(G[\{r\}]) = 1$.

    The fourth claim follows from using Gaussian elimination to compute the determinant of the graph Laplacian, which gives the number of spanning trees by the matrix-tree theorem.
\end{proof}

\begin{fact}
\label{fac:cayley}
Let $\ell_0,\ldots,\ell_k>0$. For a tree $T$ on $\{0,\ldots,k\}$, orient $T$
away from $0$ and let $p_T(i)$ be the parent of $i$. Then, for $k\geq1$,
\begin{align}
    \sum_T\prod_{i=1}^k\ell_{p_T(i)} =\ell_0(\ell_0+\cdots+\ell_k)^{k-1}.
\end{align}
For $k=0$, the corresponding sum is $1$.
\end{fact}
\begin{proof}
If $d_j$ is the degree of $j$ in $T$, then $\prod_{i=1}^k \ell_{p_T(i)} = \ell_0^{d_0}\prod_{j=1}^k \ell_j^{d_j-1}$. In the Prüfer sequence labeling $T$, vertex $j$ appears $d_j-1$ times. Hence, the left-hand side is
\begin{align}
    \ell_0\sum_{\text{Prüfer }w}\prod_{r=1}^{k-1}\ell_{w_r} = \ell_0(\ell_0+\cdots+\ell_k)^{k-1}.
\end{align}
The case $k=0$ is immediate.
\end{proof}

\subsection{Polymer representation}
We provide the polymer representation for the thermal expectation of $O$ and for the partition function (which replaces $O$ with $I$). To give a succinct definition of the polymer, we define it in terms of $\wt O$ which may include the identity.

\begin{definition}[Polymer]\label{def:polymer}
    Let $\wt O = \sum_{b \in \cA_{\wt O}} b_b Q_b$, where $\cA_{\wt O}$ potentially includes $*$ with $b_* = 1, Q_* = I$.
    A polymer $\gamma = (b; a_1,\dots,a_m)$ with $b \in \cA_{\wt O}$ satisfies the following conditions.
    \begin{itemize}
        \item $\tr(Q_b P_{a_1}\cdots P_{a_m}) \neq 0$.
        \item The sets of $\supp(P_{a_j})$ and (if $Q_b \neq I$) $\supp(Q_b)$ are connected.
    \end{itemize}
    We write $|\gamma| = m$. The support $\supp(\gamma)$ is the union of the supports of $Q_b$ and $P_{a_1},\dots,P_{a_m}$. Two polymers are compatible ($\gamma \sim \eta$) if they have disjoint support. The activity of a polymer is defined as
    \begin{align}
        w_{\wt O}(\gamma) = \frac{(-\beta)^m}{m!} b_b \lr{\prod_{i=1}^m c_{a_i}} \tr(Q_bP_{a_1} \cdots P_{a_m})
    \end{align}
    and the set of polymers is denoted by $\cP(\wt O)$.
\end{definition}

\begin{lemma}[Polymer representation]\label{lem:polymer-rep}
    The partition function and thermal expectations have polymer representations
    \begin{align}
        \tr(e^{-\beta H}) = \sum_{\Gamma \subseteq \cP(I) \;\mathrm{compatible}} \prod_{\eta\in\Gamma} w_I(\eta), \qquad \tr(Oe^{-\beta H}) = \sum_{\gamma \in \cP(O)} w_O(\gamma) \sum_{\substack{\Gamma \subseteq \cP(I) \;\mathrm{compatible}\\ \eta\sim\gamma \,\forall \eta\in\Gamma}} \prod_{\eta\in\Gamma} w_I(\eta).\label{eq:polymer-rep}
    \end{align}
\end{lemma}
\begin{proof}
    We expand
    \begin{align}\label{eq:qebh}
        \tr(Q e^{-\beta H}) = \sum_{m \geq 0} \frac{(-\beta)^m}{m!} \sum_{a_1,\dots,a_m \in \cA} \lr{\prod_{i=1}^m c_{a_i}} \tr(QP_{a_1}\cdots P_{a_m})
    \end{align}
    and note that each term has magnitude at most $(|\beta|\sum_{a \in \cA}|c_a|)^m/m!$ and thus converges absolutely for all complex $\beta$. Fix an ordered word $(a_1,\ldots,a_m)$ with $\tr(QP_{a_1}\cdots P_{a_m})\neq 0$. The trace factorizes over the connected components since different connected components have disjoint supports; each component is thus a polymer in $\cP(I)$ or $\cP(O)$, and is pairwise compatible with all other contributing polymers. From the definition of the activities in \Cref{def:polymer}, we see that $\tr(Q e^{-\beta H})$ can be written in the form of \cref{eq:polymer-rep}: if $\Gamma = \{\eta_1,\dots,\eta_k\}$ is a compatible family of polymers in $\cP(I)$ with lengths $m_1,\dots,m_k$ such that $m = m_1+\cdots+m_k$, then there are $m!/m_1! \cdots m_k!$ ways to choose $a_1,\dots,a_m$; the $m!$ in the numerator cancels the $1/m!$ in \cref{eq:qebh} and each activity $w(\eta_j)$ contributes the new factor of $1/m_j!$.
\end{proof}

\begin{lemma}[Sampling polymers]\label{lem:polymer-samp}
    Given a set $W \subseteq [n]$ and integer $m \geq 1$, there is a randomized sampler that returns a polymer $\gamma \in \cP(I)$ and a complex number $Y$ (or ``fail'' and $Y=0$) such that for every test function $f$,
    \begin{align}
        \E{Y f(\gamma)} = \sum_{\substack{\gamma \in \cP(I):|\gamma|=m\\ \supp(\gamma)\cap W \neq \emptyset}} w_I(\gamma) f(\gamma), \qquad |Y| \leq (4|\beta|\max\{s,\sqrt{Js\max\{k, |W|\}}\})^m.
    \end{align}
    Analogously, given $m \geq 0$ there is a randomized sampler that returns $\gamma, Y$ (or ``fail'' and $Y=0$) such that
    \begin{align}
        \E{Y f(\gamma)} = \sum_{\gamma \in \cP(O):|\gamma|=m} w_O(\gamma) f(\gamma), \qquad |Y| \leq B_O(4|\beta|\max\{s,\sqrt{Jsk_*}\})^m.
    \end{align}
    Additionally, the following bounds hold:
    \begin{align}\label{eq:wabs}
    \sum_{\substack{\gamma \in \cP(I):|\gamma|=m\\ \supp(\gamma)\cap W \neq \emptyset}} \abs{w_I(\gamma)} &\leq (4|\beta|\max\{s,\sqrt{Js\max\{k, |W|\}}\})^m\\
    \sum_{\gamma \in \cP(O):|\gamma|=m} \abs{w_O(\gamma)} &\leq B_O(4|\beta|\max\{s,\sqrt{Jsk_*}\})^m.
\end{align}
\end{lemma}
\begin{proof}
    To construct the sampler, we start with an initial operator $Q$, coefficient $B$, and support $V \subseteq [n]$. In the identity case, we set $Q = I$, $B = 1$ and $V = W$; if we have an observable, we sample $b$ from $\cA_O$ with probability $|b_b|/B_O$ and set $Q=Q_b$, $B=B_O$ and $V = \supp(Q_b)$.

    For Pauli string $P_a$, we write $P_a^x$ to denote the Pauli operator acting on qubit $x \in [n]$. We will keep track of an ordered list $S$ of Hamiltonian terms (possibly with repeating elements) and denote the \emph{defective} qubits of $S$ by
    \begin{align}
        D(S) = \left\{x\,:\, \tr(Q^x\prod_{a \in S}P_a^x) = 0\right\}.
    \end{align}
    Hence, $D(S) = \emptyset$ means the current product of $Q$ with terms in $S$ is proportional to the identity.
    The sampler will produce $S$ that it accepts as a polymer $\gamma$ with $|\gamma|=m$ and
    $\supp(\gamma)\cap V\neq\emptyset$. In the observable case
    $\gamma=(b;a_1,\dots,a_m)$, while in the identity case the first index is
    irrelevant and the choice $V=W$ enforces $\supp(\gamma)\cap W\neq\emptyset$.
    
    Initialize $S = \emptyset$, a transcript $r$, and a stack $\omega$. Put $Q$ on the stack with its support defined as $\supp(Q) = V$. For some fixed polymer $\gamma$, repeat the following procedure until $S = \gamma$.
    \begin{enumerate}
        \item \emph{Repair.} Suppose $S \neq \gamma$ and $S$ has at least one defective qubit. Let $x \in D(S)$ be the defective qubit chosen by a fixed deterministic rule. Because $x$ is defective and the final accepted Pauli string has nonzero trace, some new Hamiltonian term must act on $x$. We ``repair'' this defect by
        finding such a Hamiltonian term to act on $x$. Choose some $b \in \gamma \setminus S$ such that $x \in \supp(P_b)$, and add $b$ to $S$ and $P_b$ to the top of the stack $\omega$. Append to the transcript $r$ the symbol ${\rm repair}(b)$.
        \item \emph{Birth.} Suppose $S \neq \gamma$ and $S$ has no defective qubits. Repeat the following procedure: let $P_a$ be the top of the stack $\omega$; if there does not exist $b \in \gamma \setminus S$ such that $\supp(P_a) \cap \supp(P_b) \neq \emptyset$, pop $P_a$ and append to the transcript $r$ the symbol ${\rm pop}(a)$. After this procedure, $P_a$ at the top of the stack has some $x \in \supp(P_a)$ such that some $b \in \gamma \setminus S$ satisfies $x \in \supp(P_a) \cap \supp(P_b)$; this ensures that the new Hamiltonian term maintains a connected support. Choose such an $x$ and then choose such a $b$. Add $b$ to $S$ and push $P_b$ onto the stack; append to the transcript $r$ the symbol ${\rm birth}(x, b)$.
    \end{enumerate}
    This procedure assigns to every ordered polymer $\gamma$ a canonical transcript $r(\gamma)$. The transcript records the order in which occurrences are discovered by the above procedure; it doesn't store the ordering $\gamma=(a_1,\dots,a_m)$ that is used in the polymer activity $w(\gamma)$ to determine the order of the product $\tr(Q_b P_{a_1}\cdots P_{a_m})$.
    
    We now construct a randomized sampler that generates a proposed transcript and chooses an ordering for $\gamma$; we accept the proposal only if running the procedure above on the ordered $\gamma$ produces the same transcript.

    As before, initialize $S = \emptyset$, a transcript $r$, and a stack $\omega$. Put $Q$ on the stack with its support defined as $\supp(Q) = V$. Repeat the following procedure until $|S| = m$. When we say ``reject'' below, recall that we mean to abort the procedure and return ``fail'' with $Y=0$.
    \begin{enumerate}
        \item \emph{Repair.} Suppose $D(S)\neq\emptyset$. Let
        $x\in D(S)$ be the defective qubit chosen by a deterministic rule. Flip a fair coin. If heads, choose $b\in\cA$ to be the unique Hamiltonian label that makes $D(S)$ empty; if no such label exists, reject. On tails, if $S(x)=0$, reject; otherwise, sample $b\in\cA(x)$ with probability $|c_b|/S(x)$ and reject if $b$ makes $D(S)$ empty. Add this occurrence to $S$, and push $P_b$ onto the stack $\omega$. Append to the transcript $r$ the symbol ${\rm repair}(b)$.
        \item \emph{Birth.} Suppose $S$ has no defective qubits. Let $P_a$ be the top of the stack $\omega$; while $P_a \neq Q$ (i.e., the top is not the initial stack entry), repeat the following procedure. Flip a fair coin and if it lands on heads, pop $P_a$ and append to the transcript $r$ the symbol ${\rm pop}(a)$, updating $P_a$ to the new top. Stop the procedure if the coin lands on tails. Then for $P_a$ at the top of the stack, reject if $\sum_{y\in\supp(P_a)} S(y)=0$; otherwise, sample $x \in \supp(P_a)$ with probability $S(x) / \sum_{y\in\supp(P_a)} S(y)$ and then sample $b \in \cA(x)$ with probability $|c_b|/S(x)$. Add $b$ to $S$ and push $P_b$ onto the stack; append to the transcript $r$ the symbol ${\rm birth}(x, b)$.
    \end{enumerate}
    Prepare a candidate polymer $\gamma$ by selecting a random ordering of $S$; explicitly, if the transcript $r$ says that $S$ has multiplicities $\{n_a\}$, choose $\gamma$ uniformly at random from the $N_r=\frac{m!}{\prod_a n_a!}$ ordered possibilities. Accept if $\gamma$ is a valid polymer per \Cref{def:polymer} and if the first procedure, when run on $\gamma$, produces the same transcript $r(\gamma)$; otherwise, return ``fail''.
    
    For an accepted transcript $r$ producing polymer $\gamma$, set
    \begin{align}
        Y(r,\gamma) = \frac{w(\gamma)}{\pr{r}/N_r},
    \end{align}
    where $\pr{r} = \pr{r\mid b}|b_b|/B_O$ is the probability of transcript $r$, and $w$ is either $w_I$ or $w_O$ of \Cref{def:polymer}. If the transcript was rejected, set $Y = 0$. By construction, since each $\gamma$ has a unique accepted transcript $r$, this satisfies
    \begin{align}
        \E{Yf(\gamma)} = \sum_{r\ {\rm accepted}} \sum_b \frac{|b_b|}{B_O} \frac{\pr{r\mid b}}{N_r} \frac{w(\gamma)}{|b_b|\pr{r\mid b}/B_ON_r} f(\gamma) = \sum_{\gamma} w(\gamma)f(\gamma),
    \end{align}
    where the expectation is taken over the random ordering (chosen with probability $1/N_r$) and the random transcript (chosen with probability $\pr{r}$).
    It remains to bound $|Y|$. In a repair such that $D(S)$ is not empty, for fixed $x$ we have $|c_b|/2S(x) \geq |c_b|/2s$; if $D(S)$ becomes empty, we have $1/2 \geq |c_b|/2J$. Each pop occurs with probability $1/2$; there are at most $m$ pops since the stack only receives an element when something is appended to $S$. In the birth step, for fixed $P_a$ we have $S(x) / \sum_{y\in\supp(P_a)} S(y) \geq S(x) / \wt k s$; here, $\wt k = \max\{k,|W|\}$ or $k_*$ in the identity or observable cases. A birth can occur only when $D(S) = \emptyset$; crucially, immediately after a birth, the product has vanishing trace. Since the final $S$ is a polymer of size $m$ with nonvanishing trace, at most $m/2$ births can occur. If there are $b$ births, there are $m-b$ repairs, since each repair adds an element to $S$. Moreover, the number of repairs $r$ that make $D(S)$ empty is at least $r \geq b$, since a birth only occurs when $D(S) = \emptyset$.

    To upper-bound $|Y|$, we start by lower-bounding $\pr{r}$. The probability of choosing the entry ${\rm birth}(x,b)$ on the transcript, conditioned on $P_a$ being at the top of the stack, is at least
    \begin{align}
        \frac{S(x)}{\sum_{y\in\supp(P_a)}S(y)} \cdot \frac{|c_b|}{S(x)} = \frac{|c_b|}{\sum_{y\in\supp(P_a)}S(y)} \geq \frac{|c_b|}{\wt k s}.
    \end{align}
    Since $N_r \leq m!$, we have (combined with the factors of 2 from the coins associated with at most $m-b$ repairs and, from the birth step, $m$ pops and $b$ stops) that
    \begin{align}
        \frac{\pr{r}}{N_r} &\geq \frac{p}{N_r}\lr{\prod_{a \in \text{repair}\,:\,D(S)\neq\emptyset} \frac{|c_a|}{2s}}\lr{\prod_{a \in \text{repair}\,:|\,D(S)=\emptyset} \frac{|c_a|}{2J}} \lr{\frac{1}{2}}^m \lr{\prod_{a \in \text{birth}} \frac{|c_a|}{2s\wt k}} \\
        &\geq \frac{p}{m!} 2^{-2m} s^{-(m-b-r)} J^{-r} (s\wt k)^{-b} \prod_{a \in \gamma} |c_a|,
    \end{align}
    where $p=1$ for the identity case and $p=|b_b|/B_O$ for the observable case; the products run over transcript entries, where the repair entries are partitioned into those that makes $D(S)$ empty and those that do not. Note $b = |{\rm birth}| \leq m/2$, $|{\rm repair}\,:\,D(S)\neq\emptyset| = m-b-r$ and $|{\rm repair}\,:\,D(S)=\emptyset| = r$. The activity satisfies
    \begin{align}
        |w(\gamma)| = \frac{|\beta|^m}{m!} Bp \prod_{a \in \gamma} |c_a|.
    \end{align}
    Hence, using $J \leq s$ and $r \geq b$,
    \begin{align}
        |Y(r,\gamma)| \leq B\lr{4|\beta|}^m (s\wt k)^b J^r s^{m-b-r} \leq B\lr{4|\beta|\max\left\{s,\sqrt{Js\wt k}\right\}}^m
    \end{align}
    To show the remaining bounds \cref{eq:wabs}, we use the same construction as above but replace $w$ with its magnitude.
\end{proof}

\subsection{Cluster expansion}
\label{subsec:zf}
For a tuple of polymers $(\gamma_1,\ldots,\gamma_k)$, we define the \emph{incompatibility graph} $G(\gamma_1,\ldots,\gamma_k)$ with vertices $[k]$ and an edge $\{i,j\}$ whenever $\gamma_i\not\sim \gamma_j$, i.e. whenever their supports intersect. We will apply the Kotecky-Preiss criterion~\cite{kotecky1986cluster}; in fact, we will later use a slightly stronger statement derived in~\cite{ueltschi2003cluster}, so we state the standard KP condition as a corollary of it.

\begin{lemma}[Cluster expansion~\cite{ueltschi2003cluster}, Theorem~1]
\label{lem:ueltschi}
Let $(\cP,\mathcal F,\mu)$ be a measure space, where $\mu$ is a complex
measure satisfying $|\mu|(\cP)<\infty$, and let $\zeta:\cP\times\cP\longrightarrow\mathbb C$ be a measurable symmetric function. Define
\begin{align}
    Z_\mu
    &=
    \sum_{k\geq0}\frac1{k!}
    \int_{\cP}\dd\mu(\gamma_1)\cdots
    \int_{\cP}\dd\mu(\gamma_k)
    \prod_{1\leq i<j\leq k}
    \lr{1+\zeta(\gamma_i,\gamma_j)},
    \label{eq:ueltschi-partition-function}
\end{align}
where the term with $k=0$ is understood to be $1$. For $k\geq1$, let
\begin{align}
    \varphi_\zeta(\gamma_1,\ldots,\gamma_k)
    &=
    \begin{cases}
        1,
        & k=1,\\[4pt]
        \displaystyle
        \frac1{k!}
        \sum_{F\in\mathcal C_k}
        \prod_{\{i,j\}\in E(F)}
        \zeta(\gamma_i,\gamma_j),
        & k\geq2,
    \end{cases}
    \label{eq:ueltschi-cluster-coefficient}
\end{align}
where $\mathcal C_k$ is the set of connected graphs with vertex set
$[k]$.
Suppose that $\abs{1+\zeta(\gamma,\eta)} \leq 1$ for all $\gamma,\eta\in\cP$, and that there is a nonnegative measurable function
$a:\cP\to[0,\infty)$ such that, for every $\gamma\in\cP$, $\int_{\cP}\dd|\mu|(\eta)\, \abs{\zeta(\gamma,\eta)} e^{a(\eta)} \leq a(\gamma)$ and $\int_{\cP}\dd|\mu|(\gamma)\ e^{a(\gamma)}<\infty$.
Then
\begin{align}
    Z_\mu
    =
    \exp\lr{
        \sum_{k\geq1}
        \int_{\cP}\dd\mu(\gamma_1)\cdots
        \int_{\cP}\dd\mu(\gamma_k)\,
        \varphi_\zeta(\gamma_1,\ldots,\gamma_k)
    },
    \label{eq:ueltschi-cluster-expansion}
\end{align}
and the combined sum and integrals in the exponent converge absolutely.
Moreover, for every $\gamma_1\in\cP$,
\begin{align}
    1
    +
    \sum_{k\geq2} k
    \int_{\cP}\dd|\mu|(\gamma_2)\cdots
    \int_{\cP}\dd|\mu|(\gamma_k)\,
    \abs{\varphi_\zeta(\gamma_1,\ldots,\gamma_k)}
    \leq
    e^{a(\gamma_1)}.
    \label{eq:ueltschi-rooted-bound}
\end{align}
\end{lemma}

\begin{corollary}[Kotecky-Preiss criterion~\cite{kotecky1986cluster}]
\label{lem:kp}
Let $\cP$ be a finite set with a symmetric incompatibility relation
$\not\sim$, including self-incompatibility $\gamma\not\sim\gamma$, and
let $w:\cP\to\mathbb C$. For $\Lambda\subseteq\cP$, define
\begin{align}
    \Xi_{\Lambda}
    =
    \sum_{\substack{\Gamma\subseteq\Lambda\\
                    \text{pairwise compatible}}}
    \prod_{\gamma\in\Gamma}w(\gamma),
    \label{eq:xi-kp}
\end{align}
where pairwise compatibility is imposed only on distinct elements of
$\Gamma$. Suppose there are numbers $a_\gamma>0$ such that, for every
$\gamma\in\cP$,
\begin{align}
    \sum_{\eta\not\sim\gamma}
    |w(\eta)|e^{a_\eta}
    \leq a_\gamma.
    \label{eq:kp-criterion}
\end{align}
Then $\Xi_{\Lambda}\neq0$ for every $\Lambda\subseteq\cP$. Moreover,
with $G(\gamma_1,\ldots,\gamma_k)$ denoting the incompatibility graph
of the tuple,
\begin{align}
    \log\Xi_{\Lambda}
    =
    \sum_{k\geq1}\frac1{k!}
    \sum_{\gamma_1,\ldots,\gamma_k\in\Lambda}
    \varphi\lr{G(\gamma_1,\ldots,\gamma_k)}
    \prod_{i=1}^k w(\gamma_i),
    \label{eq:finite-cluster-expansion}
\end{align}
where
\begin{align}
    \varphi(G)
    :=
    \sum_{\substack{F\subseteq G\\
                    F\text{ connected and spanning}}}
    (-1)^{|E(F)|},
    \qquad
    \varphi(K_1):=1.
    \label{eq:ursell-graph-function}
\end{align}
The expansion in \eqref{eq:finite-cluster-expansion} is absolutely
convergent, and the logarithm is the branch obtained by analytic
continuation from zero activities.
\end{corollary}
\begin{proof}
    Apply \Cref{lem:ueltschi} on $\Lambda$ with $\mu(\{\gamma\}) = w(\gamma)$ and $\zeta(\gamma,\eta) = -1\{\gamma\nsim\eta\}$; then $Z_\mu = \Xi_\Lambda$ and $\varphi_\zeta = \varphi(G)/k!$.
\end{proof}

\begin{proof}[Proof of lower bound in \Cref{thm:zf}]
Set $q=4e|z|\max\{s,\sqrt{Jsk}\}$.
Since every Hamiltonian term is a nonidentity Pauli, every
polymer in $\cP(I)$ has length at least two.
By \Cref{lem:polymer-samp}, for every $W$ with
$1\leq |W|\leq k$ and $q<1$,
\begin{align}
    \sum_{\eta:\supp(\eta)\cap W\neq\emptyset}
    |w_I(\eta)|e^{|\eta|}
    \leq \sum_{m\geq2}q^m
    =\frac{q^2}{1-q}.
\end{align}
For $\gamma=(*;a_1,\ldots,a_\ell)$, every incompatible
polymer intersects the support of at least one $P_{a_j}$.
Consequently,
\begin{align}
    \sum_{\eta\not\sim\gamma}|w_I(\eta)|e^{|\eta|}
    &\leq
    \sum_{j=1}^{\ell}
    \sum_{\eta:\supp(\eta)\cap\supp(P_{a_j})\neq\emptyset}
    |w_I(\eta)|e^{|\eta|} \leq |\gamma|\frac{q^2}{1-q}
    \leq |\gamma|
\end{align}
whenever $q<3/5$. Let $\Xi_M(z)$ be the polymer partition function restricted to
$\cP_{\leq M}(I)$. Applying \Cref{lem:kp} with
$a_\gamma=|\gamma|$ gives $\Xi_M(z)\neq0$ throughout
$|z|<3/(20e\max\{s,\sqrt{Jsk}\})$, for every $M$. To pass to the limit, observe by \Cref{lem:polymer-rep} that $\Xi_M(z)\to 2^{-n}Z(z)$ locally uniformly. Since $2^{-n}Z(0)=1$, Hurwitz's theorem shows that
$Z(z)\neq0$ on $|z|<3/(20e\max\{s,\sqrt{Jsk}\})$.
In particular,
\begin{align}
    Z(z)\neq0
    \qquad\text{for all}\qquad
    |z|\leq\frac{3}{20e\max\{s,\sqrt{Jsk}\}}.
\end{align}
\end{proof}

\subsection{Decay of correlations}
\label{subsec:correlation-decay}

For geometrically local Hamiltonians in any dimension, we show here that a cluster expansion of a local thermal expectation (\Cref{lem:clustertrunc}) leads to an exponential decay of correlation between any two observables, regardless of how far apart they are. This improves upon the constraint of~\cite{harrow2020classical} that the observables must be $\Omega(\log n)$ apart (except in 1D or if the Hamiltonian is commuting), resolving open question 1(a) of the paper.

We use $\Lambda \subset \mathbb Z^D$ to denote the lattice of dimension $D$, and we use $d$ to denote Euclidean distance on $\Lambda$. For regions $U, V \subseteq \Lambda$, we write $d(U, V)$ as shorthand for $\min_{x\in U, y\in V}d(x,y)$, and similarly for $d(x,U)$. A Hamiltonian $H=\sum_{a\in\cA} c_a P_a$ is geometrically local with interaction range $R$ if for every $a\in \cA$,
\begin{align}
    \max_{x,y \in \supp(P_a)} d(x,y) \leq R.
\end{align}
Our decay of correlations result is based on \Cref{lem:clustertrunc} and a corollary.

\begin{lemma}[Cluster expansion of thermal expectation]\label{lem:clustertrunc}
    The cluster expansion
    \begin{align}\label{eq:oexpansion}
    \langle O \rangle_\beta = \sum_{k \geq 0} \frac{1}{k!} \sum_{\gamma_0 \in \cP(O)} \sum_{\gamma_1,\dots,\gamma_k\in\cP(I)} \varphi(G(\gamma_0,\dots,\gamma_k)) w_O(\gamma_0) \prod_{i=1}^k w_I(\gamma_i)
\end{align}
    is absolutely convergent and holds for all $|\beta| \max\{s, \sqrt{Jsk_*}\} \leq 1/128e$. Moreover, the terms with $|\gamma_0| + \cdots + |\gamma_k| > L$ contribute at most $B_O 8^{-L}$. Similarly,
    \begin{align}\label{eq:logzexpansion}
    \log \tr(e^{-\beta H}) = \sum_{k \geq 1} \frac{1}{k!} \sum_{\gamma_1,\dots,\gamma_k\in\cP(I)} \varphi(G(\gamma_1,\dots,\gamma_k)) \prod_{i=1}^k w_I(\gamma_i)
\end{align}
    is absolutely convergent and holds for all $|\beta| \max\{s, \sqrt{Jsk}\} \leq 1/128e$, and the terms with $|\gamma_1| + \cdots + |\gamma_k| > L$ contribute at most $n 8^{-L}$. In particular, this implies $Z(\beta) \neq 0$.
\end{lemma}
\begin{proof}
    Define
    \begin{align}
        \cP_{\leq M}(I) = \{\gamma\in\cP(I)\,:\,|\gamma| \leq M\}
    \end{align}
    and similarly for $\cP_{\leq M}(O)$.
    
    Introduce $\Xi_M(u, t)$ defined by \cref{eq:xi-kp} on a combined set of polymers $\cP$ containing both the polymers in $\cP_{\leq M}(I)$ and in $\cP_{\leq M}(O)$. The compatibility relation is defined as before ($\gamma\sim\eta$ means $\supp(\gamma) \cap \supp(\eta) = \emptyset$), except all pairs of polymers in $\cP_{\leq M}(O)$ are considered incompatible. The polymers from $\cP_{\leq M}(I)$ are assigned activities $uw_I(\gamma)$ and the polymers from $\cP_{\leq M}(O)$ are assigned $tw_O(\gamma)$.
    For $|u|$ and $|t|$ sufficiently small, we thus have
    \begin{align}
        \log \Xi_M(u,t) = \sum_{r\geq1}\frac1{r!} \sum_{\delta_1,\ldots,\delta_r\in \cP_{\leq M}(I)\sqcup \cP_{\leq M}(O)} \varphi(G(\delta_1,\ldots,\delta_r)) \prod_{j=1}^r w_{u,t}(\delta_j)
    \end{align}
    for
    \begin{align}
        w_{u,t}(\delta)= \begin{cases} u\,w_I(\delta), & \delta\in\cP_{\leq M}(I),\\ t\,w_O(\delta), & \delta\in\cP_{\leq M}(O). \end{cases}
    \end{align}
    Since all observable polymers are declared mutually incompatible, differentiating
    at $t=0$ keeps exactly those tuples with one observable polymer and $k$ identity
    polymers. Thus
    \begin{align}
        \frac{1}{\Xi_M(u,0)} \left.\frac{\partial}{\partial t}\right|_{t=0}\Xi_M(u,t) &= \sum_{k\geq0}\frac{u^k}{k!} \sum_{\gamma_0\in\cP_{\leq M}(O)} \sum_{\gamma_1,\dots,\gamma_k\in\cP_{\leq M}(I)} \varphi(G(\gamma_0,\dots,\gamma_k)) w_O(\gamma_0)\prod_{i=1}^k w_I(\gamma_i). \label{eq:xiratio}
    \end{align}
    Here the factor $1/k!$ comes from $(k+1)/(k+1)!$, since the distinguished
    observable polymer may occupy any one of the $k+1$ positions.

    Note that for now, \cref{eq:xiratio} only holds for sufficiently small $|u|$. We will shortly show that the left-hand side is holomorphic for $|u|<1$ by checking that $\Xi_M(u,0) \neq 0$; we will then show the right-hand side is also holomorphic on $|u|<1$ by the tail bound proved afterwards. The identity theorem extends \cref{eq:xiratio} to $|u|<1$, and then continuity shows it holds at $u=1$.
    
    We apply \Cref{lem:kp} to the denominator $\Xi_M(u, 0)$ to ensure it is nonzero. For $|u| \leq 1$ and $\gamma \in \cP_{\leq M}(I)$, \Cref{lem:polymer-samp} gives
    \begin{align}
\sum_{\substack{\eta \in \cP_{\leq M}(I)\\ \eta \not\sim \gamma}} \abs{u w_I(\eta)} e^{|\eta|} &\leq \sum_{i=1}^{|\gamma|} \sum_{m=1}^M e^m \sum_{\substack{\eta \in \cP_{\leq M}(I)\\ |\eta|=m\\
        \supp(\eta) \cap \supp(P_{a_i}) \neq \emptyset}} |w_I(\eta)| \leq |\gamma|\sum_{m\geq1}(4e|\beta|\max\{s, \sqrt{Jsk}\})^m.
    \end{align}
    If $4e|\beta|\max\{s, \sqrt{Jsk}\} \leq 1/2$ then we can sum the geometric series to obtain
    \begin{align}
        \sum_{\substack{\eta \in \cP_{\leq M}(I)\\ \eta \not\sim \gamma}} \abs{u w_I(\eta)} e^{|\eta|} &\leq |\gamma|.
    \end{align}
    Applying \Cref{lem:kp} with $a_\gamma = |\gamma|$, we have that the expansion
    \begin{align}\label{eq:logzin}
    \log \Xi_M(u,0) = \log \sum_{\Gamma\subseteq \cP_{\leq M}(I) \text{ compatible}} \prod_{\gamma \in \Gamma} uw_I(\gamma) = \sum_{k \geq 1} \frac{u^k}{k!} \sum_{\gamma_1,\dots,\gamma_k\in\cP_{\leq M}(I)} \varphi(G(\gamma_1,\dots,\gamma_k)) \prod_{i=1}^k w_I(\gamma_i)
\end{align}
    holds for all $|u| \leq 1$; in particular, $\Xi_M(u,0) \neq 0$ on this unit disk.
    We will show tail bounds to obtain by \Cref{lem:polymer-rep} that
    \begin{align}
        \Xi_M(1, 0) \to \tr(e^{-\beta H}) \neq 0, \quad \frac{1}{\Xi_M(1, 0)} \frac{\partial}{\partial t}\Bigg|_{t=0} \Xi_M(1,t) \to \langle O \rangle_\beta,
    \end{align}
    where the zero-freeness follows from \cref{eq:logzin}.

    It now remains to show the claimed tail bounds. We start by bounding the $k$th term of \cref{eq:xiratio}. By \Cref{fac:graph}, we have $|\varphi(G)| \leq \tau(G)$ and thus we may bound $\varphi$ by summing over spanning trees $T$ of the incompatibility graph; we can always assume $G$ is connected (and thus we can find a spanning tree) since $\varphi$ is zero otherwise. Fix a spanning tree $T$ of $G(\gamma_0,\dots,\gamma_k)$ rooted at $\gamma_0$. We bound the contributions of the activities in the $k$th term of \cref{eq:xiratio} using \Cref{lem:polymer-samp}: the $\gamma_0$ term contributes at most $B_O (4|\beta|\max\{s, \sqrt{Jsk_*}\})^{|\gamma_0|}$, and each polymer thereafter contributes $(4|\beta|\max\{s, \sqrt{Jsk_*}\})^{|\gamma_i|}$. To bound the $\varphi$, we use loose counting using the fixed tree $T$. For every child $i$, choose an element $a_j$ from the polymer $\gamma_{p_T(i)}$ for parent $p_T(i)$, such that $a_j$ intersects $\supp(\gamma_i)$. The number of choices for $a_j$ is at most $|\gamma_{p_T(i)}|$ if $p_T(i) \neq 0$ (since then the polymer is in $\cP_{\leq M}(I)$), and if $p_T(i) = 0$ there are $|\gamma_0|+1$ (since then the first index $b$ in the polymer has non-empty support). Hence, the $k$th term of \cref{eq:xiratio} is at most
    \begin{align}
        \frac{B_O (4|\beta|\max\{s, \sqrt{Jsk_*}\})^{|\gamma_0|+\cdots+|\gamma_k|}}{k!} \sum_T \prod_{i=1}^k \ell_{p_T(i)} \quad \text{for} \quad \ell_i = |\gamma_i| + 1\{i=0\}.
    \end{align}
    By \Cref{fac:cayley} and Stirling's bound $k! \geq (k/e)^k$, we have for $k \geq 1$ that
    \begin{align}\label{eq:exph}
    \frac{1}{k!}\sum_T \prod_{i=1}^k \ell_{p_T(i)} = \frac{\ell_0}{k!}(\ell_0 + \cdots + \ell_k)^{k-1} \leq \frac{\ell_0 e^k}{\ell_0 + \cdots + \ell_k} \lr{\frac{\ell_0 + \cdots + \ell_k}{k}}^k \leq \exp[1+|\gamma_0|+\cdots+|\gamma_k|],
\end{align}
    where we used $k \leq |\gamma_0| + \cdots + |\gamma_k|$. For fixed $|\gamma_0|+\cdots+|\gamma_k|=h$, the number of choices for $|\gamma_0| \geq 0$ and $|\gamma_1|, \dots, |\gamma_k| \geq 1$ is $\sum_{k=0}^h \binom{h}{k} = 2^h$; to see this, offset $|\gamma_0|$ by 1 and count the number of ways to sum positive numbers to $h+1$. Hence, the total contribution of length $h$ is $eB_O(8e|\beta|\max\{s,\sqrt{Jsk_*}\})^h$. Summing over all such terms gives remainder
    \begin{align}
        \sum_{h > L} eB_O(8e|\beta|\max\{s, \sqrt{Jsk_*}\})^h \leq eB_O\sum_{h > L} 16^{-h} \leq B_O8^{-L}
    \end{align}
    if $8e|\beta|\max\{s, \sqrt{Jsk_*}\} \leq 1/16$.

    For the partition function, we use a similar argument. For $\gamma_0,\dots,\gamma_k \in \cP(I)$, we use the identity
    \begin{align}\label{eq:1id}
    1 = \sum_{x \in [n]} \frac{1}{\abs{\bigcup_{i=0}^k \supp(\gamma_i)} \cdot \#\{i:x\in\supp(\gamma_i)\}} \sum_{j=0}^k 1\{x \in \supp(\gamma_j)\},
\end{align}
    where summands with no $i$ such that $x\in\supp(\gamma_i)$ are absent. Truncating
    \begin{align}\label{eq:zin}
    \sum_{k \geq 1} \frac{1}{k!} \sum_{\gamma_1,\dots,\gamma_k \in \cP(I)} \varphi(G(\gamma_1,\dots,\gamma_k)) \prod_{i=1}^k w_I(\gamma_i)
\end{align}
    at $|\gamma_1|+\cdots+|\gamma_k| \leq L$ gives using \cref{eq:1id}
    \begin{align}
        &\sum_{h=1}^L \sum_{\substack{r\geq1\\ m_1,\ldots,m_r\geq1\\
        m_1+\cdots+m_r=h}} \frac1{r!} \sum_{\substack{\eta_1,\ldots,\eta_r\in\cP(I)\\ |\eta_i|=m_i}} \varphi(G(\eta_1,\ldots,\eta_r)) \prod_{i=1}^r w_I(\eta_i) \notag\\
        &= \sum_{x=1}^n \sum_{h=1}^L \sum_{\substack{r\geq1\\ m_1,\ldots,m_r\geq1\\
        m_1+\cdots+m_r=h}} \frac1{r!} \sum_{\substack{\eta_1,\ldots,\eta_r\in\cP(I)\\ |\eta_i|=m_i}} \frac{\varphi(G(\eta_1,\ldots,\eta_r))} {\left|\bigcup_{i=1}^r\supp(\eta_i)\right| \cdot \#\{i:x\in\supp(\eta_i)\}} \lr{\sum_{j=1}^r\mathbf 1\{x\in\supp(\eta_j)\}} \prod_{i=1}^r w_I(\eta_i)\\
        &=  \sum_{x=1}^n \sum_{h=1}^L \sum_{k \geq 0} \sum_{\substack{m_0, \dots, m_k \geq 1\\
        m_0 + \cdots + m_k = h}} \frac{1}{k!} \sum_{\substack{\gamma_0,\dots,\gamma_k\in\cP(I)\\
        |\gamma_i| = m_i, \, x\in\supp(\gamma_0)}} \frac{\varphi(G(\gamma_0,\dots,\gamma_k))} {\abs{\bigcup_{i=0}^k \supp(\gamma_i)} \cdot \#\{i:x\in\supp(\gamma_i)\}} \prod_{i=0}^k w_I(\gamma_i),
    \end{align}
    where in the last line we evaluated the sum over $j$ to get a factor of $r$, and we relabeled $r=k+1$ and $\gamma_i = \eta_{i+1}$. We now bound the absolute value of the final expression. Compared to our previous argument, the main novelty is the factor
    \begin{align}
        \lr{ \abs{\bigcup_{i=0}^k\supp(\gamma_i)} \cdot \#\{i:x\in\supp(\gamma_i)\} }^{-1} \leq 1.
    \end{align}
    The remaining argument is similar to before, summing over spanning trees $T$. The only difference now is that the root polymer is required to touch $\{x\}$, and we sum over $x \in [n]$. For a fixed $x$, the contribution of $|\gamma_0| + \cdots + |\gamma_k| = h$ is at most
    \begin{align}
        \frac{(4|\beta|\max\{s, \sqrt{Jsk}\})^h}{k!} \sum_T \prod_{i=1}^k |\gamma_{p_T(i)}| \leq (4e|\beta|\max\{s, \sqrt{Jsk}\})^h,
    \end{align}
    where now we took $\ell_i = |\gamma_i|$ even at $i=0$. The number of decompositions $m_0 + \cdots + m_k = h$ is $2^{h-1}$, so for $8e|\beta|\max\{s, \sqrt{Jsk}\} \leq 1/16$, the terms with $|\gamma_0|+\cdots+|\gamma_k| > L$ contribute at most
    \begin{align}
        \sum_{h > L} n 2^{h-1} (4e|\beta| \max\{s, \sqrt{Jsk}\})^h \leq n \sum_{h > L} 16^{-h} \leq n8^{-L}
    \end{align}
    to \cref{eq:zin}. Since this bound is uniform in $M$, dominated convergence lets us pass from $\cP_{\leq M}(I)$ to $\cP(I)$ and obtain $\tr(e^{-\beta H})$ as claimed in \cref{eq:logzexpansion}.
\end{proof}

\begin{corollary}[Deleting interactions far from an observable]\label{cor:distant-obs}
    Let $H=\sum_{a\in\cA} c_a P_a$ be a geometrically local $(s,k)$-long-range Pauli Hamiltonian of interaction range at most $R$. Define $k'$-local observable $O = \sum_{b\in\cA_O} b_b Q_b$ for coefficients $b_b$ and Pauli strings $Q_b$. Define for some $\cD \subseteq \cA$
    \begin{align}
        H' = H - \sum_{a \in \cD} c_a P_a.
    \end{align}
    If for some $L > 0$, every $a \in \cD$ satisfies $d(\supp(O), \supp(P_a)) \geq LR$ then
    \begin{align}
        \frac{1}{2}\abs{\Tr(O\rho_\beta(H)) - \Tr(O\rho_\beta(H'))} \leq 8^{-L} \sum_{b\in\cA_O} |b_b| \quad \text{ for all}\quad 0 \leq \beta \leq \frac{1}{128e\,\max\{s,\sqrt{Js\max\{k,k'\}}\}}.
    \end{align}
\end{corollary}
\begin{proof}
    For Hamiltonian $K$, let $\Sigma(K)$ denote the truncation of the cluster expansion \cref{eq:oexpansion} to polymers satisfying $|\gamma_0| + \cdots + |\gamma_j| \leq L$. Then by \Cref{lem:clustertrunc},
    \begin{align}
        \abs{\Tr(O\rho_\beta(H)) - \Sigma(H)} \leq 8^{-L} \sum_{b\in\cA_O} |b_b| \quad \text{and} \quad \abs{\Tr(O\rho_\beta(H')) - \Sigma(H')} \leq 8^{-L} \sum_{b\in\cA_O} |b_b| 
    \end{align}
    for all $0 \leq \beta \leq \frac{1}{128e\,\max\{s,\sqrt{Js\max\{k,k'\}}\}}$. It thus suffices to prove that $\Sigma(H) = \Sigma(H')$. Since each polymer in the representation of \Cref{lem:clustertrunc} is a connected component---i.e., can be written such that the supports of $Q_b, P_1, P_2, \dots$ are consecutively intersecting---and since each Hamiltonian term's support has diameter at most $R$, any choice of polymers $\gamma_0,\dots,\gamma_j$ has distance at most $(|\gamma_0| + \cdots + |\gamma_j|)R$ from $\supp(O)$. Since we assumed that $|\gamma_0| + \cdots + |\gamma_j| \leq L$, and since $H, H'$ agree on all terms within $LR$ of $\supp(O)$, we conclude that $\Sigma(H) = \Sigma(H')$.
\end{proof}

Since \Cref{cor:distant-obs} allows us to truncate the Hamiltonian to the local region around an observable, we can now prove correlation decay (\Cref{thm:decay}) fairly straightforwardly.

\begin{proof}[Proof of \Cref{thm:decay}]
    Writing $A = \sum_{b \in \cA_A} a_b Q_b$ for Pauli operators $Q_b$, we have that
    \begin{align}\label{eq:opbound}
        \sum_{b\in\cA_A} |a_b| \leq \sqrt{4^{|X|} \sum_{b\in\cA_A} |a_b|^2} = \sqrt{2^{|X|} \Tr_X(A^2)} \leq 2^{|X|} \norm{A}
    \end{align}
    where the first inequality bounded the number of Paulis by $4^{|X|}$, the equality used the orthogonality of Pauli strings, and the final inequality bounds Frobenius norm by $2^{|X|/2}$ times the operator norm. It thus suffices to show that
    \begin{align}
        \abs{\langle AB \rangle_\beta - \langle A \rangle_\beta \langle B \rangle_\beta} \leq 48 \lr{\sum_{b\in\cA_A} |a_b|}\lr{\sum_{b\in\cA_B} |b_b|} \exp[-\frac{\log 8}{6R}d(X,Y)]
    \end{align}
    for $B = \sum_{b \in \cA_B} b_b Q'_b$. Note that this immediately holds if $d(X,Y) \leq 6R$ since $\abs{\langle AB \rangle_\beta - \langle A \rangle_\beta \langle B \rangle_\beta} \leq 2 \lr{\sum_{b\in\cA_A} |a_b|}\lr{\sum_{b\in\cA_B} |b_b|}$. Hence, we assume $d(X,Y) > 6R$.
    
    Define disjoint regions
    \begin{align}
        \Lambda_X = \{z\in\Lambda:d(z,X) < d(X,Y)/3\}, \quad \Lambda_Y = \{z\in\Lambda:d(z,Y) < d(X,Y)/3\}, \quad \Lambda_0 = \Lambda \setminus (\Lambda_X \cup \Lambda_Y)
    \end{align}
    and corresponding Hamiltonians $H_X, H_Y, H_0$ containing the terms of $H$ that are entirely supported in the corresponding region, e.g., $H_X = \sum_{a \in \cA : \supp(P_a) \subseteq \Lambda_X} c_a P_a$. Since there are no interactions between these regions, the Gibbs state of $H' = H_X + H_Y + H_0$ factorizes over them and trivially satisfies
    \begin{align}
        \Tr(AB \rho_\beta(H')) = \Tr(A \rho_\beta(H')) \Tr(B \rho_\beta(H')).
    \end{align}
    Since each term $a$ in the set $\cD \subset \cA$ of deleted terms in $H-H'$ satisfies $d(\supp(P_a),X\cup Y)\geq d(X,Y)/3-R>d(X,Y)/6$, we can apply \Cref{cor:distant-obs} with $L=d(X,Y)/(6R)$ to obtain
    \begin{align}
        \frac{1}{2}\abs{\Tr(A\rho_\beta(H)) - \Tr(A\rho_\beta(H'))} &\leq 8^{-d(X,Y)/6R} \sum_{b \in \cA_A} |a_b|
    \end{align}
    and similarly for $B$ and $AB$. Finally, since
    \begin{align}
        AB = \sum_{b \in \cA_A} \sum_{b' \in \cA_B} a_b b_{b'} Q_b Q'_{b'}
    \end{align}
    has coefficient mass at most $\lr{\sum_{b \in \cA_A} |a_b|}\lr{\sum_{b \in \cA_B} |b_b|}$, the triangle inequality gives
    \begin{align}
        \abs{\langle AB \rangle_\beta - \langle A \rangle_\beta \langle B \rangle_\beta} &\leq \abs{\Tr(AB \rho_\beta(H)) - \Tr(AB\rho_\beta(H'))}\notag\\
        &\quad + \abs{\Tr(A\rho_\beta(H))\Tr(B\rho_\beta(H)) - \Tr(A\rho_\beta(H'))\Tr(B\rho_\beta(H'))}\\
        &\leq 6\cdot 8^{-d(X,Y)/6R}\lr{\sum_{b \in \cA_A} |a_b|}\lr{\sum_{b \in \cA_B} |b_b|}\\
        &\leq 48 \cdot 2^{|X|+|Y|} \norm{A} \norm{B} \exp[-\frac{\log 8}{6R}d(X,Y)].
    \end{align}
    where in the second line we bounded thermal expectations by 1-norm, e.g.~$\langle A \rangle_\beta \leq \sum_{b \in \cA_A} |a_b|$, and in the third line we used \cref{eq:opbound}.
\end{proof}

\subsection{Polynomial-time sampler}
\label{subsec:alg}
We consider the truncations of the cluster expansions \cref{eq:oexpansion} and \cref{eq:logzexpansion} for $\langle O \rangle_\beta$ and $\log \tr e^{-\beta H}$ to order $\sum_j |\gamma_j| \leq L$; denote these truncations by $\Sigma_O$ and $\Sigma_H$ respectively. We will now create random variables $X_O$ and $X_H$ such that $\EE\,X = \Sigma$ in each case, and $|X|$ is bounded to ensure concentration. The proofs in both cases are very similar, so we will drop the subscript where possible.

The procedure to generate $X$ is as follows. For some parameter $\rho$, sample $h \in [L]$ with probability
\begin{align}
    p(h) = \frac{(1-\rho)\rho^{h-1}}{1-\rho^L}, \quad \rho_O = 128e|\beta|\max\{s, \sqrt{Jsk_*}\}, \quad \rho_H = 128e|\beta|\max\{s, \sqrt{Jsk}\}.
\end{align}
Choose $k$ and $m_0,\dots,m_k$ uniformly from all choices that satisfy
$h=m_0+\cdots+m_k$. In the observable case, $m_0\geq0$ and $m_i\geq1$ for
$i\geq1$; in the partition function case, all $m_i\geq1$. Also set
$\ell_i=m_i$, except in the observable case put $\ell_0=m_0+1$.

Passing $m_0$ to \Cref{lem:polymer-samp}, sample the root polymer; for the partition function case, also pass to the sampler $W = \{x\}$ for uniformly random $x \in [n]$. If the sampling procedure fails, output 0. If $k \geq 1$, sample a rooted tree $T$ on $\{0,\dots,k\}$ with probability proportional to $\prod_{i=1}^k \ell_{p_T(i)}$, where $p_T(i)$ denotes the parent of child $i$. Traverse $T$ away from the root: for each edge from a parent $\gamma_{p_T(i)}$ to a child, choose one term $a_j$ in $\gamma_{p_T(i)}$ uniformly at random. Pass $W=\supp(P_{a_j})$ and $m_i$ to the sampler of \Cref{lem:polymer-samp} to obtain child polymer $\gamma_i$; if it fails, output $0$.

Compute the incompatibility graph $G(\gamma_0,\dots,\gamma_k)$ and its quantities $\varphi$ and $\tau$ by \Cref{fac:graph}. For the observable case, output
\begin{align}
    X_O = \frac{2^h}{p_O(h)} \cdot \frac{\ell_0(\ell_0 + \cdots + \ell_k)^{k-1}}{k!} \cdot \frac{\varphi(G(\gamma_0,\dots,\gamma_k))}{\tau(G(\gamma_0,\dots,\gamma_k))} \lr{\prod_{i=1}^k \#\{a_j \in \gamma_{p_T(i)} \,:\, \supp(P_{a_j}) \cap \supp(\gamma_i)\}}^{-1} \prod_{i=0}^k Y_i.
\end{align}
(We abuse notation so $P_{a_j}$ may include $Q_b$, since in the observable case $\supp(Q_b)$ is not empty.) For the partition function, output
\begin{align}
    X_H &= \frac{n2^{h-1}}{p_H(h)} \cdot \frac{\ell_0(\ell_0 + \cdots + \ell_k)^{k-1}}{k!} \cdot \frac{\varphi(G(\gamma_0,\dots,\gamma_k))}{\tau(G(\gamma_0,\dots,\gamma_k))} \lr{\prod_{i=1}^k \#\{a_j \in \gamma_{p_T(i)} \,:\, \supp(P_{a_j}) \cap \supp(\gamma_i)\}}^{-1} \notag\\
    &\quad \times \lr{\abs{\bigcup_{i=0}^k \supp(\gamma_i)} \#\{0\leq i \leq k\,:\, x\in\supp(\gamma_i)\}}^{-1} \prod_{i=0}^k Y_i.
\end{align}
We now prove that this procedure creates an unbiased estimator of $\Sigma$.

\begin{lemma}[Random estimator]\label{lem:expectestimator}
    Denote the truncations to order $\sum_j |\gamma_j| \leq L$ of the cluster expansions \cref{eq:oexpansion} and \cref{eq:logzexpansion} for $\langle O \rangle_\beta$ and $\log \tr e^{-\beta H}$ by $\Sigma_O$ and $\Sigma_H$, respectively. Assume $\rho_O, \rho_H \leq 1/4$. The above procedure yields random variables $X_O, X_H$ satisfying
    \begin{align}
        \EE\, X_O = \Sigma_O, \quad \EE\, X_H = \Sigma_H, \quad |X_O| \leq B_O, \quad |X_H| \leq n.
    \end{align}
\end{lemma}
\begin{proof}
    We drop subscripts where the proofs coincide. We start by evaluating the expectation of the $Y$ variables, holding everything else fixed: for a test function $f$,
    \begin{align}
        \E{f(\gamma_0,\ldots,\gamma_k) \prod_{i=0}^kY_i} = \sum_{\gamma_0,\ldots,\gamma_k:|\gamma_i|=m_i} w_{\rm root}(\gamma_0) \prod_{i=1}^k w_I(\gamma_i) f(\gamma_0,\ldots,\gamma_k),
    \end{align}
    where $w_{\rm root}$ is $w_O$ in the observable case and $w_I$ in the partition function case. The polymers are constrained such that each child polymer $\gamma_i$ touches the support of the chosen Hamiltonian term (or observable $Q_b$) in $\gamma_{p_T(i)}$; moreover, in the partition function case, the sum is also constrained such that $x \in \supp(\gamma_0)$.

    We proceed to evaluate the expectation over the remaining random choices. The factors $2^h/p_O(h)$ and $n2^{h-1}/p_H(h)$ cancel the probabilities of choosing $h$, a decomposition $m_0 + \cdots + m_k = h$, and (for the partition function case) $x \in [n]$. For a fixed $T \subseteq G(\gamma_0,\dots,\gamma_k)$, the factor $\prod_{i=1}^k \#\{a_j \in \gamma_{p_T(i)} \,:\, \supp(P_{a_j}) \cap \supp(\gamma_i)\}$ counts the number of valid choices of $a_j$ in the traversing procedure; hence, choosing uniformly random $a_j$ cancels this. \Cref{fac:cayley} implies that the probability of $T$ is $(\prod_{i=1}^k \ell_{p_T(i)})/\ell_0(\ell_0 + \cdots + \ell_k)^{k-1}$, which cancels $\ell_0(\ell_0 + \cdots + \ell_k)^{k-1}$. Finally,
    \begin{align}
        \sum_{T\subseteq G(\gamma_0,\ldots,\gamma_k)\ {\rm spanning}} \frac{\varphi(G(\gamma_0,\ldots,\gamma_k))} {\tau(G(\gamma_0,\ldots,\gamma_k))} = \varphi(G(\gamma_0,\ldots,\gamma_k)).
    \end{align}
    Hence, the observable estimator satisfies
    \begin{align}
        \EE\,X_O = \sum_{\substack{k\geq0,\ m_0\geq0,\ m_i\geq1\\
        m_0+\cdots+m_k\leq L}} \frac1{k!} \sum_{\gamma_0\in\cP(O)} \sum_{\gamma_1,\ldots,\gamma_k\in\cP(I)} \mathbf 1_{\{|\gamma_i|=m_i\}} \varphi(G(\gamma_0,\ldots,\gamma_k)) w_O(\gamma_0)\prod_{i=1}^kw_I(\gamma_i) = \Sigma_O
    \end{align}
    and the partition function estimator satisfies
    \begin{align}
        \EE\,X_H = \sum_{x=1}^n \sum_{\substack{k\geq0,\ m_0,\ldots,m_k\geq1\\
        m_0+\cdots+m_k\leq L}} \frac1{k!} \sum_{\substack{\gamma_0,\ldots,\gamma_k\in\cP(I)\\
        |\gamma_i|=m_i,\ x\in\supp(\gamma_0)}} \frac{\varphi(G(\gamma_0,\ldots,\gamma_k))} {\left|\bigcup_{i=0}^k\supp(\gamma_i)\right| \#\{0\leq i\leq k:x\in\supp(\gamma_i)\}} \prod_{i=0}^kw_I(\gamma_i) = \Sigma_H,
    \end{align}
    where the last equality follows from \cref{eq:1id}.

    For boundedness, \Cref{lem:polymer-samp} gives in the observable and partition function cases, respectively, the bounds
    \begin{align}
        \prod_{i=0}^k |Y_i| \leq B_O (4|\beta|\max\{s, \sqrt{Jsk_*}\})^h, \qquad \prod_{i=0}^k |Y_i| \leq (4|\beta|\max\{s, \sqrt{Jsk}\})^h.
    \end{align}
    We use bounds $|\varphi(G)|/\tau(G) \leq 1$ (\Cref{fac:graph}) and
    \begin{align}
        \prod_{i=1}^k \#\{a_j \in \gamma_{p_T(i)} \,:\, \supp(P_{a_j}) \cap \supp(\gamma_i)\} \geq 1, \quad \abs{\bigcup_{i=0}^k \supp(\gamma_i)} \#\{0\leq i \leq k\,:\, x\in\supp(\gamma_i)\} \geq 1.
    \end{align}
    As seen in \cref{eq:exph}, $\ell_0(\ell_0 + \cdots + \ell_k)^{k-1}/k!$ is upper-bounded by $e^{h+1}$ and $e^h$ for the observable and partition function cases. Since $\rho_O, \rho_H \leq 1/4$, we have
    \begin{align}
        p(h)=\frac{(1-\rho)\rho^{h-1}}{1-\rho^L} \geq (1-\rho)\rho^{h-1} \geq \frac12\rho^{h-1}.
    \end{align}
    Plugging in these bounds and the explicit values of $\rho$ yields
    \begin{align}
        |X_O| \leq \frac{2^h}{p_O(h)}\,e^{h+1}\, B_O(4|\beta|\max\{s, \sqrt{Jsk_*}\})^h \leq 2eB_O\rho_O \lr{\frac{8e|\beta|\max\{s, \sqrt{Jsk_*}\}}{\rho_O}}^h = 2eB_O\rho_O\,16^{-h} \leq B_O,
    \end{align}
    where we used $h\geq1$ and $\rho_O\leq1/4$. Similarly,
    \begin{align}
        |X_H| \leq \frac{n2^{h-1}}{p_H(h)}\,e^h\,(4|\beta|\max\{s, \sqrt{Jsk}\})^h  \leq n\rho_H \lr{\frac{8e|\beta|\max\{s, \sqrt{Jsk}\}}{\rho_H}}^h = n\rho_H\,16^{-h} \leq n.
    \end{align}
\end{proof}

\begin{proof}[Proof of \Cref{thm:expect}]
    Set the factor $A$ to be either $B_O$ or $n$ for the observable or the partition function case, and set $\Sigma^*$ to be either $\langle O \rangle_\beta$ or $\log \tr e^{-\beta H}$. Set the truncation level $L$ and number of samples $N$ to be
    \begin{align}
        L = \left\lceil\log_8 \frac{16A}{\epsilon}\right\rceil, \quad N = \left\lceil\frac{32 A^2}{\epsilon^2} \log \frac{2}{\delta}\right\rceil.
    \end{align}
    Since $\Sigma^*$ is real, we bound the truncation error in \Cref{lem:clustertrunc} by
    \begin{align}
        \abs{\Sigma^* - \Re \Sigma} \leq \abs{\Sigma^* - \Sigma} \leq A 8^{-L} \leq \frac{\epsilon}{2}.
    \end{align}
    \Cref{lem:expectestimator} prepares an estimator with $N$ samples
    \begin{align}
        \wh \Sigma = \Re{\frac{1}{N} \sum_{r=1}^N X^{(r)}}.
    \end{align}
    By \Cref{lem:expectestimator} and Hoeffding's inequality, we bound the probability of failure by $\delta$ via
    \begin{align}
        \pr{\abs{\wh \Sigma - \Re \Sigma} > \frac{\epsilon}{2}} \leq 2\exp[-\frac{N\epsilon^2}{8A^2}] \leq \delta.
    \end{align}
    The two errors of $\epsilon/2$ add to $\epsilon$; it only remains to bound the runtime. We preprocess the lists $\cA(x)$ and $S(x)$ in time $O(|\cA|k + |\cA_O|k_O)$. A single draw of the estimator uses polymers of total length at most $L$ and computes one coefficient $\varphi$, whose leading cost is $O(3^L)$ by \Cref{fac:graph}; the remaining operations are polynomial in $L$ and in the support sizes. Repeating everything (except the preprocessing) $N$ times yields the final claim.
\end{proof}

\section{Classical sampling of separable Gibbs states}
\label{sec:prep}
In this section, we prove the algorithmic statement in \Cref{thm:samp} that describes a polynomial-time classical algorithm to prepare quantum Gibbs states of long-range Pauli Hamiltonians. This algorithm will only work at separable temperatures. \Cref{thm:samp} also states that at any asymptotically lower temperature, there are Gibbs states that are a constant distance away from any separable state; this result is proven in \Cref{thm:sep-upper} (\Cref{sec:upper}).

In \Cref{sec:expect} we showed zero-freeness of the partition function and of \emph{local} thermal expectations. To prepare a mixture of product states close to the Gibbs state, we require zero-freeness statements that hold for \emph{global} observables; otherwise, e.g., measurement outcomes would not necessarily be accurately sampled. Unfortunately, \Cref{thm:expect} only gives global observables at inverse temperatures $\beta = o(1)$. We first show how to estimate global observables at $\beta \lesssim 1/sk$; in particular, our strategy will work at separable temperatures although the proof never explicitly invokes separability. We will then combine this with the algorithm from the separability result to ultimately prepare the Gibbs state as a mixture of product states.

For $Y \subseteq [n]$ and $y \in\{0,1\}^Y$, we estimate global observables by repeatedly estimating partition functions of the form $\tr_y(e^{-\beta H})$, where the $\tr_y$ operation is defined as
\begin{align}
    \tr_y(A) = 2^{-(n-|Y|)} \Tr_{[n]\setminus Y} \bra{y}A\ket{y}.
\end{align}
The proofs remain largely unchanged compared to our earlier proof for estimating the partition function, since $\tr_y(A)$ also factorizes over connected components. However, the temperature threshold will go like $1/sk$ instead of $1/\sqrt{sk}$, since the ``repair'' trick used in \Cref{lem:polymer-samp} no longer applies; we briefly elaborate on this. In \Cref{lem:polymer-samp}, we grew a polymer over $m$ steps by combining Hamiltonian terms with intersecting supports. A ``defective'' qubit acted on by a nonidentity Pauli had to be ``repaired'' by choosing a Hamiltonian term that acted as nonidentity on the same qubit; otherwise, it wouldn't contribute to the trace of the product of Hamiltonian terms. This constraint reduced the counting, since we could force the next term to act on the defective qubit instead of choosing any of $k$ qubits: after $m$ steps of growing the polymer, the number of choices grew as $k^{m/2}$ instead of $k^m$. In comparison, replacing $\tr$ by $\tr_y$ can cause defective qubits to contribute: for example, $\tr(Z) = 0$ but $\bra{y}Z\ket{y} = \pm 1$. Hence, we will use the looser $k^m$ counting that yields a zero-free radius going as $1/sk$. In fact, we will show a slightly better pinned zero-free radius of $\log(k)/(sk)$, which we will show is tight with a non-analyticity at real $\beta \sim \log(k)/(sk)$ in \Cref{sec:upper}.

Once we show zero-freeness of $\tr_y(e^{-\beta H})$, we obtain an algorithm that samples measurement outcomes in the computational basis by pinning qubits one at a time. For $y \in \{0,1\}^{j-1}$, we will construct a polynomial-time algorithm that samples a random variable (similarly to \Cref{lem:expectestimator}) that, in expectation, estimates the pinning odds of the $j$th bit,
\begin{align}\label{eq:rj}
    r_j(y) = \log \tr_{y1}(e^{-\beta H}) - \log \tr_{y0}(e^{-\beta H}).
\end{align}
By proceeding one qubit at a time and conditioning on the previous measurement outcomes in $y$, one can use the estimate of $r_j(y)$ to efficiently sample in the computational basis. To estimate the Gibbs expectation of a global Pauli operator $Q$, we conjugate both $Q$ and $H$ by a unitary $U$; this does not change the thermal expectation. By choosing $U$ to be a product of single-qubit Clifford gates such that $UQU^\dagger$ is diagonal in the $Z$ basis, we can estimate this quantity by sampling measurement outcomes in the computational basis. Hence, we have an efficient algorithm for estimating global Pauli expectations for all $\beta \lesssim 1/sk$.

To prepare the Gibbs state, similarly to \cite{bakshi2024high}, we use the algorithm implicit in the pinning procedure of \Cref{lem:sep-pinning} in the separability proof. After $T$ steps, this procedure outputs a succinct description of an unnormalized separable state $\sigma_T$ such that $\E{\sigma_T} = e^{-\beta H}$, where the expectation is over the randomness of the pinning procedure. The normalized Gibbs state can thus be written as
\begin{align}
    \frac{e^{-\beta H}}{\Tr(e^{-\beta H})} = \E{\frac{\Tr(\sigma_T)}{\Tr(e^{-\beta H})} \frac{\sigma_T}{\Tr(\sigma_T)}}.
\end{align}
We can compute the factor $\Tr(\sigma_T)/\Tr(e^{-\beta H})$ using the invariant property \cref{eq:sep-pinning-operator-invariant} of the pinning procedure that at step $t+1$,
\begin{align}
    \E{\sigma_{t+1} \mid \sigma_t} = \sigma_t,
\end{align}
implying the telescoping product
\begin{align}
    \frac{\Tr(\sigma_T)}{\Tr(e^{-\beta H})} = \prod_t \frac{\Tr(\sigma_{t+1})}{\Tr(\sigma_t)}.
\end{align}
Hence, we can enforce the normalization factor $\Tr(\sigma_T)/\Tr(e^{-\beta H})$ by using rejection sampling to bias the $(t+1)$th proposed pinning step by $\Tr(\sigma_{t+1})/\Tr(\sigma_t)$. The unnormalized state $\sigma_t$ is of the form
\begin{align}
    \sigma_t = e^{-\beta H_{S_t}/2} \prod_{j=1}^t (I+\lambda_jX_j) e^{-\beta H_{S_t}/2},
\end{align}
where $H_{S_t}$ only contains the Hamiltonian terms acting on qubits $S_t \subseteq [n]$. These ratios reduce to partition function ratios and global Pauli expectations, which are estimated using \Cref{thm:expect} and \Cref{thm:globalpauli}. Finally, each factor $I+\lambda_jX_j$ is sampled using the explicit separable decomposition in \Cref{lem:sep-pauli-factor}, giving a pure product stabilizer state.

\subsection{Estimating global Pauli expectations}
The main result of this subsection is the following theorem.

\begin{theorem}[Global Pauli expectations]\label{thm:globalpauli}
    For an $(s,k)$-long-range Pauli Hamiltonian $H$ and Hermitian Pauli string $Q$, if $0 \leq \beta \leq 1/(32e\,sk)$, then given $\epsilon,\delta\in(0,1)$ one can output a random number $\wh q$ satisfying
    \begin{align}
        \pr{\abs{\wh q - \frac{\Tr(Q e^{-\beta H})}{\Tr(e^{-\beta H})}} > \epsilon} \leq \delta
    \end{align}
    in time
    \begin{align}
        \wt O\lr{|\cA| k + \frac{k|\supp(Q)|^4 \log |\cA|}{\epsilon^5} \log \frac{1}{\delta}}.
    \end{align}
\end{theorem}

As before, we require a zero-freeness statement to obtain thermal expectations. We use the Kotecky-Preiss criterion of \Cref{lem:ueltschi} (Theorem 1 of~\cite{ueltschi2003cluster}).

\begin{lemma}[Pinned zero-freeness]\label{lem:pin-zf}
    For an $(s,k)$-long-range Pauli Hamiltonian $H$, $\tr_y(e^{-\beta H}) \neq 0$ for all $|\beta| \leq \max\{\log(k+1)/(64esk),1/(32esk)\}$.
\end{lemma}
\begin{proof}
    We use polymers $\gamma = (a_1,\dots,a_m)$ for $m \geq 1$ with activities
    \begin{align}\label{eq:wy}
        w_y(\gamma) = \frac{(-\beta)^m}{m!} \lr{\prod_{i=1}^m c_{a_i}}\tr_y(P_{a_1}\cdots P_{a_m}).
    \end{align}
    A polymer must have connected support; two polymers are compatible if their supports are disjoint.
    This is the same polymer and activity as the previously defined family $\cP(I)$ (\Cref{def:polymer}), except we replace $\tr$ with $\tr_y$. Hence, \Cref{lem:polymer-rep} gives polymer representation
    \begin{align}
        \label{eq:pin-polymer-rep} \tr_y(e^{-\beta H}) = \sum_{\Gamma \text{ compatible}} \prod_{\gamma\in\Gamma} w_y(\gamma),
    \end{align}
    where $\Gamma = \emptyset$ contributes 1; the proof is identical to that of \Cref{lem:polymer-rep} since $\tr_y$ factorizes over disjoint supports exactly as $\tr$ does. Let $\cP_0$ be the set of polymers with nonzero $\tr_y(P_{a_1}\cdots P_{a_m})$ and all $a_i$ such that $\tr_y(P_{a_i}) = 0$. We now show that for any $W \subseteq [n]$ with $|W| \leq k$,
    \begin{align}\label{eq:sumpin}
        \sum_{\substack{\gamma\in\cP_0:|\gamma|=m\\ \supp(\gamma)\cap W \neq \emptyset}} \abs{w_y(\gamma)} \leq (8|\beta| s \sqrt k)^m,
    \end{align}
    following the same argument as \Cref{lem:polymer-samp}. A qubit is called defective when the current product of Hamiltonian terms has Pauli $X$ or $Y$ on a pinned qubit, or a nonidentity Pauli on an unpinned qubit. Since each $a_i$ satisfies $\tr_y(P_{a_i}) = 0$, a birth must be followed by a repair; there are at most $m/2$ births, producing the $\sqrt k$ factor. Explicitly, while a repair at a specified defective qubit contributes at most $s$, a birth contributes at most $ks$. For $b \leq m/2$ births, this is a total contribution of $s^{m-b} (ks)^b$. Since we abort at a transcript with $m$ pushes and $m$ pops, their sequence in the transcript can be chosen in at most $\binom{2m}{m}$ ways; each push can be either a birth or repair, so we have a counting prefactor of $\binom{2m}{m}2^m \leq 8^m$. Summing over \cref{eq:wy} with multiplicities $\{n_a\}$ introduces counting factor $m!/\prod_a n_a! \leq m!$ that cancels the denominator, yielding \cref{eq:sumpin} after maximizing over $0 \leq b \leq m/2$. Setting $A=1+2e|\beta|sk$, we find that if $q=8e^A|\beta|s\sqrt k \leq 1/2$ then
    \begin{align}\label{eq:a12}
        \sum_{\substack{\gamma\in\cP_0\\ \supp(\gamma)\cap W = \emptyset}} |w_y(\gamma)| e^{A|\gamma|} \leq \sum_{m \geq 2} q^m = \frac{q^2}{1-q} \leq \frac{1}{2}.
    \end{align}
    For the assumed bounds on $\beta$, note that $q < 1/2$, implying the final inequality above.
    We now show that this implies by Kotecky-Preiss that $\tr_y(e^{-\beta H}) \neq 0$. 
    
    Introduce variables $t_1,\dots,t_m \in [0,1]$ corresponding to a polymer composed of Hamiltonian terms $P_{a_1}, \dots, P_{a_m}$, as well as $u,v \in [0,1]$. For $\gamma \in \cP_0$ and $a,b \in \{a\in\cA\,:\,\tr_y(P_a) \neq 0\}$, define
    \begin{align}
        d\lambda(\gamma, \vec t) = \frac{(-\beta)^m}{m!} \lr{\prod_{i=1}^m c_{a_i}} \tr_y(\cT_{\vec t} P_{a_1}\cdots P_{a_m}) d\vec t, \qquad d\lambda(a, u) = -\beta c_a \tr_y(P_a) du,
    \end{align}
    where $\cT_{\vec t}$ denotes ordering the product according to decreasing $t_i$. For each $\gamma$, this satisfies
    \begin{align}\label{eq:intt}
        \int d|\lambda|(\gamma, \vec t) = \int_{[0,1]^m} \frac{|\beta|^m}{m!} \lr{\prod_{i=1}^m |c_{a_i}|} \abs{\tr_y(\cT_{\vec t} P_{a_1}\cdots P_{a_m})} d\vec t = \frac{|\beta|^m}{m!} \lr{\prod_{i=1}^m |c_{a_i}|} = |w_y(\gamma)|.
    \end{align}
    Define the symmetric quantity $B$ by
    \begin{align}
        B((\gamma, \vec t), (\eta, \vec u)) = 1\{\supp(\gamma) \cap \supp(\eta) = \emptyset\}, \quad B((a,u), (b,v)) = 1
    \end{align}
    and
    \begin{align}
        B((\gamma, \vec t), (a,u)) = B((a,u), (\gamma, \vec t)) = \prod_{i:t_i > u} \begin{cases}
            +1 & P_a P_{a_i} = P_{a_i} P_a\\
            -1 & P_a P_{a_i} = -P_{a_i} P_a.
        \end{cases}
    \end{align}
    For $a \in \{a\in\cA\,:\,\tr_y(P_a) \neq 0\}$, each $P_a$ has only $I,Z$ on pinned qubits and $I$ on unpinned qubits; hence, they commute with each other and satisfy $\tr_y(P_aR) = \tr_y(P_a) \tr_y(R)$ for every operator $R$. Applying the same decomposition over connected components as in \Cref{lem:polymer-rep}, we find that
    \begin{align}
        \tr_y(e^{-\beta H}) = \sum_{p \geq 0} \frac{1}{p!} \int \prod_{j=1}^p d\lambda(\omega_j) \prod_{i<j} B(\omega_i, \omega_j).
    \end{align}
    We check the conditions of \Cref{lem:ueltschi} with $\zeta = B-1$, $a_{(\gamma,\vec t)} = A|\gamma|$ and $a_{(a,u)}=1$. By \cref{eq:a12} and $|B| \leq 1$, we have that
    \begin{align}
        \int d|\lambda|(\omega) \abs{B((a,u),\omega) - 1} e^{a_\omega} \leq 2 \cdot \frac{1}{2} = 1
    \end{align}
    since $B((a,u), (b,v)) = 1$ ensures that only polymers intersecting $\supp(P_a)$ can contribute to the integral, allowing us to apply \cref{eq:a12} by translating $|w_y(\gamma)|$ to the integral over $d\vec t$ via \cref{eq:intt}. Similarly, for $\gamma \in \cP_0$ consisting of terms $P_{a_i}$, we have
    \begin{align}
        \int d|\lambda|(\omega) \abs{B((\gamma,\vec t),\omega) - 1} e^{a_\omega} &\leq \sum_{i=1}^m \int_{\supp(\eta)\cap\supp(P_{a_i})\neq\emptyset} \hspace{-5em}e^{A|\eta|} d|\lambda|(\eta, \vec u) \hspace{1em} + \sum_{\substack{a\in \cA:\tr_y(P_a) \neq 0\\ \supp(P_a)\cap\supp(\gamma)\neq\emptyset}} \int_0^1 2 e |\beta| |c_a| du\\
        &\leq \sum_{i=1}^m \sum_{\substack{\eta\in\cP_0\\ \supp(\eta) \cap \supp(P_{a_i}) \neq 0}} |w_y(\eta)| e^{A|\eta|} + 2e|\beta|\sum_{\substack{a\in \cA:\tr_y(P_a) \neq 0\\ \supp(P_a) \cap \supp(\gamma) \neq 0}}|c_a|\\
        &\leq \frac{|\gamma|}{2} + 2e|\beta|sk|\gamma| \leq A|\gamma|,
    \end{align}
    where the first term comes from $\omega=(\eta, \vec u)$ from $\cP_0$ (applying \cref{eq:a12} to a union bound over $\supp(\gamma)$), and the second term comes from $\omega=(a,u)$ (since $d|\lambda|(a,u) = |\beta| |c_a| du$ and the KP weight is $e^1$). Finally, \cref{eq:a12} proves the remaining condition of \Cref{lem:ueltschi} is $\int e^{a_\omega} d|\lambda|(\omega) \leq n/2 + e|\beta| \sum_{a:\tr_y(P_a) \neq 0} |c_a| \leq n(1/2+e|\beta|s) < \infty$ at each $n < \infty$; in particular, $|n^{-1} \log \tr_y(e^{-\beta H})| \leq 1/2 + e|\beta|s$. We can thus apply \Cref{lem:ueltschi} to conclude that $\tr_y(e^{-\beta H}) \neq 0$.
\end{proof}

We will use $\cT_m$ to denote the set of labeled trees on $m$ vertices. We will use $G(a_1,\dots,a_m)$ to denote the overlap graph of the supports of $P_{a_1},\dots, P_{a_m}$; i.e., it is a graph on $m$ vertices with edges between $i,j$ such that $\supp(P_{a_i}) \cap \supp(P_{a_j}) \neq \emptyset$. Denoting the set of partition of $[m]$ by $\Pi([m])$, we will use cumulants
\begin{align}
    K_y(a_1,\dots,a_m) = \sum_{\pi\in\Pi([m])} (-1)^{|\pi|-1}(|\pi|-1)! \prod_{B\in\pi} \tr_y\lr{\prod_{i\in B}^{\longrightarrow} P_{a_i}}
\end{align}
which satisfy the moment-cumulant relation (given by a Möbius inversion)
\begin{align}\label{eq:pinmomcum}
    \tr_y\lr{\prod_{i\in S}^{\longrightarrow} P_{a_i}} = \sum_{\pi\in\Pi(S)}\prod_{B\in\pi} K_y\lr{(a_i)_{i\in B}}.
\end{align}
These satisfy the following properties.
\begin{fact}[Tree bound]\label{fac:pincum}
    The following properties hold.
    \begin{itemize}
        \item If $G(a_1,\dots,a_m)$ is disconnected, $K_y(a_1,\dots,a_m) = 0$.
        \item $|K_y(a_1,\dots,a_m)| \leq 2^{m-1} \tau(a_1,\dots,a_m)$.
    \end{itemize}
\end{fact}
\begin{proof}
    To show the first property, suppose $G(a_1,\dots,a_m)$ is disconnected, and suppose this is witnessed by $A \sqcup B = [m]$. Due to the factorization of trace across disjoint supports, \cref{eq:pinmomcum} implies that mixed cumulants across the $A,B$ vanish. Hence, $K_y(a_1,\dots,a_m) = 0$.

    We observe that the second property holds trivially if $G$ is disconnected due to the first property, so we assume $G$ is connected. For partition $\pi \in \Pi([m])$, define $\phi_G(\pi)$ to be the partition whose blocks are the vertex sets of the connected components of the induced subgraph $G[B]$, as $B$ ranges over the blocks of $\pi$. By the factorization of $\tr_y$ over disjoint supports, we can rewrite
    \begin{align}
        \prod_{B \in \pi}\tr_y\lr{\prod_{i\in B}^{\longrightarrow}P_{a_i}} = \prod_{B \in \phi_G(\pi)}\tr_y\lr{\prod_{i\in B}^{\longrightarrow}P_{a_i}}.
    \end{align}
    Hence, the terms of the cumulant can be regrouped as
    \begin{align}
        K_y(a_1,\dots,a_m) = \sum_{\pi' \in \Pi([m])}\lr{\sum_{\pi:\phi_G(\pi)=\pi'}(-1)^{|\pi|-1}(|\pi|-1)!}\prod_{B \in \pi'}\tr_y\lr{\prod_{i\in B}^{\longrightarrow}P_{a_i}}.
    \end{align}
    Since every block $C$ of $\phi_G(\pi)$ is the vertex set of a connected component of $G[B]$, $\phi_G(\pi)$ is by definition in the set
    \begin{align}
        \Pi_{\rm conn}(G) = \left\{ \rho\in\Pi([m]): G[B]\text{ is connected for every }B\in\rho \right\}.
    \end{align}
    For $\rho \in \Pi_{\rm conn}(G)$, let $G/\rho$ be the quotient graph whose vertices are the blocks of $\rho$ and whose edges connect two blocks whenever an edge of $G$ joins them. Given $\rho$ we can construct $\pi$ such that $\phi_G(\pi)=\rho$ by grouping blocks of $\rho$ such that no two blocks in the same group are adjacent in $G/\rho$. Hence, we have
    \begin{align}
        \sum_{\pi:\phi_G(\pi)=\rho} (-1)^{|\pi|-1}(|\pi|-1)! &= \!\!\!\!\sum_{\pi\in\Pi(V(G/\rho))} (-1)^{|\pi|-1}(|\pi|-1)! \!\!\!\!\!\prod_{\{u,v\}\in E(G/\rho)}\!\!\!\!\! \mathbf 1\{u,v\text{ lie in distinct blocks of }\pi\}\\
        &= \sum_{\substack{F\subseteq E(G/\rho)\\(V(G/\rho),F)\ {\rm connected}}}
        (-1)^{|F|}\\
        &= \varphi(G/\rho).
    \end{align}
    By \Cref{fac:graph} and the fact that every normalized trace of Paulis has magnitude at most 1, this gives
    \begin{align}
        \abs{K_y(a_1,\dots,a_m)} \leq \sum_{\rho\in\Pi_{\rm conn}(G)} \abs{\varphi(G/\rho)} \leq \sum_{\rho\in\Pi_{\rm conn}(G)} \tau(G/\rho).
    \end{align}
    Since $G$ is connected, we can finish the proof using a similar spanning tree counting argument as in the proof of \Cref{fac:graph} to evaluate the sum over $\rho\in\Pi_{\rm conn}(G)$. For $B$ in the vertex set of $G$, let $T_B$ denote a spanning tree of $G[B]$. Let $S = \bigcup_{B\in\rho}T_B$; the connected components of $S$ are by construction the blocks of $\rho$. To construct a spanning tree $T$ of $G$, we add to $S$ the edges $\{B,B'\}$ in the edge set of a spanning tree of $G/\rho$. This forms an injective map from $\rho$ and a spanning tree of $G/\rho$ to a spanning tree of $G$ and the set $S$. Since $T$ must have $m-1$ edges, there are $2^{m-1}$ possible subsets $S$, giving
    \begin{align}
        \sum_{\rho\in\Pi_{\rm conn}(G)} \tau(G/\rho) \leq 2^{m-1}\tau(G).
    \end{align}
\end{proof}

\begin{lemma}[Truncating pinning odds]\label{lem:rtrunc}
    For $0 \leq \beta \leq 1/(32e sk)$, the quantity
    \begin{align}\label{eq:rl}
    r_j^L(y) = \sum_{m=1}^L \frac{(-\beta)^m}{m!} \sum_{a_1,\dots,a_m\in\cA} \lr{\prod_{i=1}^m c_{a_i}}\left[K_{y1}(a_1,\dots,a_m) - K_{y0}(a_1,\dots,a_m)\right]
\end{align}
    satisfies
    \begin{align}
        \abs{r_j(y) - r_j^L(y)} \leq 8^{-L}.
    \end{align}
    Moreover, any nonzero term in \cref{eq:rl} has connected supports $\supp(P_{a_1}), \dots, \supp(P_{a_m})$ and at least one $\supp(P_{a_i})$ acting on $j$.
\end{lemma}
\begin{proof}
    By \Cref{lem:pin-zf}, $\log\tr_{y1}(e^{-\beta H})$ and $\log\tr_{y0}(e^{-\beta H})$ are analytic in the stated disk. Expanding the exponential into moments and then moving to cumulants with \cref{eq:pinmomcum} gives
    \begin{align}
        \log \tr_y(e^{-\beta H}) = \sum_{m \geq 1} \frac{(-\beta)^m}{m!} \sum_{a_1,\dots,a_m\in\cA}\lr{\prod_{i=1}^m c_{a_i}} K_y(a_1,\dots,a_m),
    \end{align}
    and thus
    \begin{align}\label{eq:rjexpansion}
    r_j(y) = \sum_{m\geq 1} \frac{(-\beta)^m}{m!} \sum_{a_1,\dots,a_m\in\cA} \lr{\prod_{i=1}^m c_{a_i}}\left[K_{y1}(a_1,\dots,a_m) - K_{y0}(a_1,\dots,a_m)\right].
\end{align}
    If the supports on $a_1,\dots,a_m$ are disconnected, then $K_{y1} = K_{y0} = 0$ by \Cref{fac:pincum}; if no Hamiltonian term $P_{a_i}$ acts on $j$, then for every $B\subseteq[m]$ we have $\tr_{y1}\!\left(\prod_{i\in B}^{\longrightarrow}P_{a_i}\right) = \tr_{y0}\!\left(  \prod_{i\in B}^{\longrightarrow}P_{a_i} \right)$, and hence $K_{y1}(a_1,\ldots,a_m)  = K_{y0}(a_1,\ldots,a_m)$ by the moment-cumulant relation \cref{eq:pinmomcum}.

    It remains to show the tail bound. We enforce the above two facts and use \Cref{fac:pincum} to bound the order-$m$ contribution to \cref{eq:rjexpansion} by
    \begin{align}\label{eq:mcont}
    \frac{\beta^m}{m!} 2^m \sum_{a_1,\dots,a_m\in \cA} \lr{\prod_{i=1}^m |c_{a_i}|} \lr{\sum_{\ell=1}^m 1\{j \in \supp(P_{a_\ell})\}} \sum_{T \in \cT_m} 1\{T\subseteq G(a_1,\dots,a_m)\},
\end{align}
    where $2^{m-1}$ and the sum over the family $\cT_m$ of trees of size $m$ come from \Cref{fac:pincum}, an additional factor of $2$ comes from the triangle inequality on $K_{y1}$ and $K_{y0}$, and the indicator on $j$ avoids the vanishing terms where no Hamiltonian term acts on $j$. To show an upper bound, root a tree at $\ell$ and count the contributions of $|c_{a_i}|$; the root has weight at most $s$ and every child intersecting $\supp(P_a)$ has weight at most $ks$. Hence, \cref{eq:mcont} is at most
    \begin{align}
        \frac{\beta^m}{m!} 2^m |\cT_m| m s (ks)^{m-1} \leq \frac{(2e\beta ks)^m}{k}
    \end{align}
    by Cayley's formula and $m^{m-1}/m! \leq e^m$. The tail of $m > L$ terms in \cref{eq:rjexpansion} thus contributes at most
    \begin{align}
        \sum_{m > L} \frac{(2e\beta ks)^m}{k} \leq 8^{-L}
    \end{align}
    using $k \geq 1$ and assuming $2e\beta ks \leq 1/16$.
\end{proof}

\begin{lemma}[Random estimator]\label{lem:restimator}
    For $0 \leq \beta \leq 1/(32esk)$, there is a random variable $R_j^L(y)$ such that
    \begin{align}
        \EE\, R_j^L(y) = r_j^L(y), \qquad \abs{R_j^L(y)} \leq 1.
    \end{align}
    After $O(|\cA|k)$ of preprocessing, one draw is computable in time
    $\wt O(3^L k + L \log |\cA|)$.
\end{lemma}
\begin{proof}
    As before in \cref{eq:axsx}, let $\cA(x) = \{a:x\in\supp(P_a)\}$ and $S(x)=\sum_{a\in\cA(x)}|c_a| \leq s$. If $S(j) = 0$ or $\beta = 0$, output 0. Otherwise, put $\rho = 8\beta e s k$ and sample $m \in [L]$ with probability
    \begin{align}
        p(m) = \frac{(1-\rho)\rho^{m-1}}{1-\rho^L}.
    \end{align}
    Sample $\ell \in [m]$ and $T \in \cT_m$ uniformly at random; set the root of $T$ at $\ell$. Sample $a_\ell \in \cA(j)$ with probability $|c_{a_\ell}|/S(j)$. For each child $i\ne \ell$, choose a qubit $x \in \supp(P_{a_{p_T(i)}})$ with probability $S(x)/\sum_{q \in \supp(P_{a_{p_T(i)}})}S(q)$ and choose $a_i \in \cA(x)$ with probability $|c_{a_i}|/S(x)$. Output $0$ if $T\not\subseteq G(a_1,\dots,a_m)$, and otherwise output the real part of
    \begin{align}\label{eq:rwt}
    \wt R &= \frac{(-\beta)^m}{m!} \frac{K_{y1}(a_1,\dots,a_m) - K_{y0}(a_1,\dots,a_m)}{\#\{i\,:\,j\in\supp(P_{a_i})\} \prod_{i \neq \ell} \abs{\supp(P_{a_{p_T(i)}}) \cap \supp(P_{a_i})}} \frac{m |\cT_m| S(j)}{p(m) \tau(G(a_1,\dots,a_m))} \notag\\
    &\quad \times \lr{\prod_{i\neq\ell} \sum_{x\in\supp(P_{a_{p_T(i)}})} S(x)} \lr{\prod_{i=1}^m \sgn(c_{a_i})}.
\end{align}
    We now evaluate $\EE\, \wt R$. We sum over the probabilities
    \begin{align}
        \frac{S(x)}{\sum_{q \in \supp(P_{a_{p_T(i)}})}S(q)} \cdot \frac{|c_{a_i}|}{S(x)} = \frac{|c_{a_i}|}{\sum_{q \in \supp(P_{a_{p_T(i)}})}S(q)}.
    \end{align}
    Summing over the possible intersecting choices of $x$ produces the factor $\abs{\supp(P_{a_{p_T(i)}})\cap\supp(P_{a_i})}$ that is canceled by \cref{eq:rwt}. We sum over the choices of $a_\ell$ with probability $|c_{a_\ell}|/S(j)$; these cancel factors in \cref{eq:rwt} to get
    \begin{align}
        \frac{(-\beta)^m}{m!} \frac{K_{y1}(a_1,\dots,a_m) - K_{y0}(a_1,\dots,a_m)}{\#\{i\,:\,j\in\supp(P_{a_i})\}} \frac{m |\cT_m|}{p(m) \tau(G(a_1,\dots,a_m))} \lr{\prod_{i=1}^m c_{a_i}}.
    \end{align}
    Summing over the sampling of $m$, $\ell$ such that $j \in \supp(P_{a_\ell})$ (otherwise the term does not contribute), and $T \in \cT_m$ such that $T \subseteq G$ (otherwise $G$ is disconnected and the terms don't contribute) cancels the factors
    \begin{align}
        \frac{1}{\#\{i\,:\,j\in\supp(P_{a_i})\}} \frac{m |\cT_m|}{p(m) \tau(G(a_1,\dots,a_m))}.
    \end{align}
    leaving
    \begin{align}
        \EE\,\wt R = \sum_{m=1}^L \frac{(-\beta)^m}{m!} \sum_{a_1,\dots,a_m\in\cA} \lr{\prod_{i=1}^m c_{a_i}} \left[K_{y1}(a_1,\dots,a_m)-K_{y0}(a_1,\dots,a_m)\right] = r_j^L(y).
    \end{align}
    Since $r_j^L$ is real, setting $R_j^L = \Re \wt R$ does not change the expectation.

    To show the upper bound on $|R_j^L|$, note that when the output is nonzero,
    \begin{align}
        \#\{i\,:\,j\in\supp(P_{a_i})\}\geq 1,\qquad \prod_{i\ne \ell} \abs{\supp(P_{a_{p_T(i)}})\cap\supp(P_{a_i})}\geq 1,\qquad \tau(G(a_1,\dots,a_m))\geq 1.
    \end{align}
    By \Cref{fac:pincum}, $|K_{y1}(a_1,\dots,a_m)-K_{y0}(a_1,\dots,a_m)| \leq 2^m \tau(G(a_1,\dots,a_m))$.
    Using $S(j)\leq s$, $\sum_{x\in\supp(P_a)}S(x)\leq ks$, Cayley's formula $|\cT_m|=m^{m-2}$, and $m^{m-1}/m!\leq e^m$, we obtain
    \begin{align}
        |\wt R| &\leq \frac{\beta^m}{m!}\,2^m\tau(G) \frac{m|\cT_m|s}{p(m)\tau(G)}(ks)^{m-1} \leq \frac{2e\beta s(2e\beta sk)^{m-1}}{p(m)}.
    \end{align}
    Since $\rho \leq 1/4$ by assumption, $p(m) \geq \frac{1}{2}\rho^{m-1}$ and hence $|\wt R| \leq 1/8k \leq 1$. The runtime follows from preparing $\cA(x)$ and $S(x)$ in time $O(|\cA|k)$, and the main cost of the sampling comes from computing the cumulants in time $O(3^L)$ following the same recurrence as \Cref{fac:graph}.
\end{proof}

\begin{proof}[Proof of \Cref{thm:globalpauli}]
    If $Q = \pm I$, output $\pm 1$. Otherwise, apply Clifford $U$ such that $UQU^\dagger = \sigma \prod_{j\in\supp(Q)} Z_j$ for sign $\sigma \in \{\pm 1\}$, where we choose $U$ to be a product of single-qubit Clifford gates. Hence, $H' = UHU^\dagger$ is an $(s,k)$-long-range Pauli Hamiltonian. It thus suffices to estimate the thermal expectation of $Q' = \prod_{j\in\supp(Q)} Z_j$ on $H'$; we do this by sampling computational basis measurement outcomes on $\supp(Q') = \supp(Q)$.

    Denote the qubits we're measuring by $\supp(Q) = \{i_1,\dots,i_q\}$. For a prefix $y \in \{0,1\}^{\ell-1}$ on $i_1,\dots,i_{\ell-1}$, the conditional probability of the next bit is
    \begin{align}
        \pr{X_{i_\ell} = 0 \mid X_{i_1}\cdots X_{i_{\ell-1}} = y} = \frac{1}{1 + \exp[r_{i_\ell}(y)]},
    \end{align}
    where $r_j$ is given by \cref{eq:rj} but is defined with respect to $H'$. To sample from this conditional distribution, we use $N$ samples of the random variable $R_j^L$ in \Cref{lem:restimator}; we repeat the entire procedure to get $M$ samples of measurement outcomes.
    Set
    \begin{align}
        M=\left\lceil \frac{8}{\epsilon^2}\log\frac{4}{\delta}\right\rceil, \qquad L=\left\lceil \log_8\frac{8|\supp(Q)|}{\epsilon}\right\rceil, \qquad N=\left\lceil \frac{128|\supp(Q)|^2}{\epsilon^2} \log\frac{4M|\supp(Q)|}{\delta}\right\rceil.
    \end{align}
    For each of the $M$ measurement samples and each prefix $y$, estimate
    $r_{i_\ell}(y)$ by
    \begin{align}
        \wh r_{i_\ell}(y) = \frac1N\sum_{t=1}^N R_{i_\ell}^{L,(t)}(y),
    \end{align}
    where the random variables are drawn using \Cref{lem:restimator}
    for the Hamiltonian $H'$. Then sample $X_{i_\ell}$ with probability $1/(1+\exp[\wh r_{i_\ell}(y)])$. After repeating $M$ times, output the estimate
    \begin{align}\label{eq:qh}
    \wh q=\frac1M\sum_{t=1}^M \sigma(-1)^{\sum_{\ell=1}^q X_{i_\ell}^{(t)}}.
\end{align}
    By \Cref{lem:rtrunc}, the choice of $L$ gives
    $|r_j(y)-r_j^L(y)|\leq 8^{-L}\leq \epsilon/(8|\supp(Q)|)$. By
    \Cref{lem:restimator} and Hoeffding's inequality,
    \begin{align}
        \pr{\abs{\wh r_j(y)-r_j^L(y)}>\frac{\epsilon}{8|\supp(Q)|}} \leq 2\exp[-\frac{N\epsilon^2}{128|\supp(Q)|^2}] \leq \frac{\delta}{2M|\supp(Q)|}.
    \end{align}
    A union bound over all at most $M|\supp(Q)|$ estimates of $r_j$ used by the
    algorithm shows that, with probability at least $1-\delta/2$, every
    estimate satisfies $|\wh r_j(y)-r_j(y)|\leq \epsilon/(4|\supp(Q)|)$.

    Conditioned on this event, we bound the resulting total variational distance in measurement outcomes via a coupling argument. Consider drawing from the true measurement outcome and the distribution given by $\wh r_j$ using the same randomness; if they agree on $y \in \{0,1\}^{\ell-1}$, the probability of disagreeing at the $\ell$th step is at most $\epsilon/(4|\supp(Q)|)$. Union bounding over all $|\supp(Q)|$ steps gives a total variational distance of $\epsilon/4$. Hence, the expectation of each term of the sum in \cref{eq:qh} differs by at most $\epsilon/2$ from the true thermal expectation. Since each term is also bounded in $[-1, 1]$, Hoeffding's inequality gives
    \begin{align}
        \pr{\abs{\wh q-\EE[\wh q\mid \{r_j\}]}>\epsilon/2} \leq 2\exp[-M\epsilon^2/8] \leq \delta/2.
    \end{align}
    Combining each of the failures that occur with probability $\delta/2$ gives the final result.

    It remains to bound the runtime. The preprocessing and single-qubit Clifford conjugation cost $O(|\cA|k)$. The algorithm invokes \Cref{lem:restimator} at most $MN|\supp(Q)|$ times, and each invocation costs $\wt O(3^L k+L\log|\cA|)$. Substituting the above choices of $M,N,L$ gives the stated runtime (after loosening the exponents to become integers).
\end{proof}

\subsection{Sampling product states}
We now turn the separable pinning proof into a sampler for the Gibbs state that outputs a pure product state; averaging over the algorithm's randomness approximates the Gibbs state in trace distance. As before, for $S \subseteq [n]$ we will write
\begin{align}
    H_S = \sum_{a\in \cA:\supp(P_a)\subseteq S} c_a P_a.
\end{align}
To construct the sampler, we use the product states constructed in the separability proof. \Cref{lem:sep-pinning} creates a distribution over $\chi=((\lambda_1,X_1),\dots,(\lambda_m,X_m))$ such that $e^{-\beta H} = \EE\,\prod_{j=1}^m(I + \lambda_j X_j)$, where each $X_j$ is a Pauli string and each $\lambda_j$ is small enough to ensure that $e^{-\beta H}$ is separable. The main tool to generate this distribution is the propagator established in \Cref{lem:sep-propagator-expansion} and \Cref{lem:sep-propagator-sampling}, which is repeatedly used to propose the next $\lambda_j, X_j$ in $\chi$ via the pinning procedure of \Cref{lem:sep-pinning}. To bound the cost of this procedure, we first record the cost of the propagator step.

\begin{lemma}[Cost of propagator algorithm]\label{lem:prop-alg}
    Given a $(s,k)$-long-range Pauli Hamiltonian $H = \sum_{a\in\cA}c_a P_a$, let $S \subseteq [n]$, let $T \subseteq \supp(P_{a_*})$ for some $a_*\in\cA$, and set $\wh S = S \setminus T$. There is a randomized procedure which outputs $(b, E, C(E))$ where $b \geq 0$, and $C(E) = (a_1,\dots,a_t)$ is connected to $T$, and $E = \prod_{j=1}^t P_{a_j}$ is a Pauli string up to a phase, where $t \geq 1$ is drawn with probability $2^{-t}$. These quantities satisfy 
    \begin{align}
        \E{I+bE} = e^{-\beta H_S} e^{\beta H_{\wh S}}, \qquad b \leq (6|\beta|sk)^t,
    \end{align}
    and, conditioned on $t \leq L$, the procedure generates a sample in time $O(L|\cA|k + L^2k)$.
\end{lemma}
\begin{proof}
    The bound on $b$ is already shown in \Cref{lem:sep-propagator-sampling}. We restate the sampling procedure in algorithmic terms to note its cost.
    
    Sample $t \geq 1$ with probability $2^{-t}$, initialize the current Pauli $E$ to $I$, set $C$ to be empty, and introduce a likelihood ratio $B = 1$. Define the function $f_t(S,T)$ from \Cref{lem:sep-propagator-expansion} from the recurrence
    \begin{align}
        f_{t+1}(S,T) = -[H_S, f_t(S,T)] - f_t(S,T)(H_S - H_{S\setminus T}).
    \end{align}
    Suppose that the procedure has already generated labels $a_1,\dots,a_r$; set
    \begin{align}
        W_R = \sum_{a \in \cA_S: \supp(P_a) \cap T \neq \emptyset} |c_a| \leq sk, \qquad W_C = \sum_{a \in \cA_S : \supp(P_a) \cap \lr{\cup_{i=1}^r \supp(P_{a_i})} \neq \emptyset} |c_a| \leq sk \cdot r.
    \end{align}
    If $W_R=W_C=0$, return $b=0$ and $E=I$. Otherwise, choose a label $a_{r+1}$ and update the Pauli $E$ as follows. With probability $W_R/(W_R+2W_C)$, sample $a_{r+1}$ from $\{a \in \cA_S: \supp(P_a) \cap T \neq \emptyset\}$ with probability $|c_a|/W_R$, and set $E$ to $-\sgn(c_{a_{r+1}}) E P_{a_{r+1}}$. (This corresponds to going through the $-f_r(S,T) (H_S - H_{S\setminus T})$ term in \Cref{lem:sep-propagator-expansion}.) Otherwise, sample $a_{r+1}$ from $\{a \in \cA_S : \supp(P_a) \cap \lr{\cup_{i=1}^r \supp(P_{a_i})} \neq \emptyset\}$ with probability $|c_a|/W_C$, then set $E$ to $-\sgn(c_{a_{r+1}}) P_{a_{r+1}} E$ or $\sgn(c_{a_{r+1}}) EP_{a_{r+1}}$ with equal probability. (This corresponds to going through the $-[H_S,f_r(S,T)]$ term in \Cref{lem:sep-propagator-expansion}.) In either case, multiply $B$ by $W_R+2W_C$, and append $a_{r+1}$ to $C(E)$.

    By induction on $r$, the expectation of $B$ times the current $E$ after $r$ steps is exactly $f_r(S,T)$. After $t$ steps, choose $b = (2|\beta|)^tB/t!$ and absorb the phase of $\beta$ into $E$ to give
    \begin{align}
        \E{I+bE} = I + \EE_t\left[\frac{(2\beta)^t f_t(S,T)}{t!}\right] = e^{-\beta H_S} e^{\beta H_{\wh S}}
    \end{align}
    by \cref{eq:sep-propagator-series} of \Cref{lem:sep-propagator-expansion}. For $t \leq L$, all the sets needed above can be formed by scanning the Hamiltonian terms and checking whether their supports intersect $T$ or  $\bigcup_{i=1}^r\supp(P_{a_i})$; each check costs $O(k)$. Since there are at most $L$ iterations and $|\supp(E)|\leq Lk$, the total cost of this is $O(L|\cA|k+L^2k)$ after the lists $\cA(x)$ have been preprocessed (which costs $O(|\cA|k)$ and is thus dominated).
\end{proof}

Using the pinning procedure of \Cref{lem:sep-pinning}, we apply the propagator to go from the state $(S,\chi)$ to a new state $(\wh S, \wh \chi)$, where $\wh S = S\setminus T$ and $\wh \chi$ is updated by the output $(b, E, C(E))$ of \Cref{lem:prop-alg}. Explicitly, we use the propagator twice to generate two tuples $(b_1, E_1, C(E_1))$ and $(b_2, E_2, C(E_2))$, which uniquely specify the next $\lambda_{m+1}$ and $X_{m+1}$ following the table of \Cref{lem:sep-pinning}. If we run the above propagator algorithm at fixed $L \geq 1$, we truncate the tail probability $2^{-L}$ for each of the two propagator calls; hence, the total variational distance between the true and returned proposal distributions over $(\wh S, \wh \chi)$ is at most $2^{1-L}$ at each step.

Let us recall how the pinning procedure of \Cref{lem:sep-pinning} works in more detail. Let $(S_t, \chi_t)$ record the state of the pinning procedure at step $t$. We also track the invariant \cref{eq:sep-pinning-operator-invariant} in the pinning procedure, i.e.,
\begin{align}\label{eq:invar1}
    e^{-\beta H} = \E{\sigma_t}, \qquad \sigma_t = e^{-\beta H_{S_t}/2} \prod_{j=1}^t (I+\lambda_jX_j) e^{-\beta H_{S_t}/2}.
\end{align}
Due to this invariant, updating from $\omega_t$ to $\omega_{t+1}$ gives
\begin{align}\label{eq:ratiocond}
    \EE_{\omega_{t+1}}\left[\Tr(e^{-\beta H_{S_{t+1}}/2} \prod_{j=1}^{t+1} (I+\lambda_j X_j) e^{-\beta H_{S_{t+1}}/2}) \, \Bigg|\; \omega_t\right] = \Tr(e^{-\beta H_{S_t}/2} \prod_{j=1}^t (I+\lambda_j X_j) e^{-\beta H_{S_t}/2}).
\end{align}
At the start of the pinning procedure, $\omega_0 = ([n], \emptyset)$. At the end, $S_T = \emptyset$ and
\begin{align}
    \sigma_T = \prod_{j=1}^T (I+\lambda_jX_j).
\end{align}

\begin{lemma}[Cost of partition function ratio estimation]\label{lem:rh}
    For all $\beta \leq 1/(4096esk)$, every proposal $\omega_t \to \omega_{t+1}$, and every $\epsilon,\delta\in(0,1)$, there is a randomized algorithm that outputs $\wh r$ satisfying
    \begin{align}
        \pr{\abs{\wh r - \frac{\Tr(\sigma_{t+1})}{\Tr(\sigma_t)}} > \epsilon} \leq \delta
    \end{align}
    with cost
    \begin{align}
        \wt O\lr{|\cA|k+\frac{kn^4}{\epsilon^5}\log |\cA|\log\frac1\delta}.
    \end{align}
    Moreover,
    \begin{align}\label{eq:rhobound}
        0 \leq \frac{\Tr(\sigma_{t+1})}{\Tr(\sigma_t)} \leq 5.
    \end{align}
\end{lemma}
\begin{proof}
    We rewrite $\Tr(\sigma_{t+1})/\Tr(\sigma_t)$ in terms of Pauli expectations and partition functions, which we will estimate using \Cref{thm:expect} and \Cref{thm:globalpauli}. To do this, we follow the proof of \Cref{lem:sep-pinning}, where we consider updating from $\omega = (S, \chi)$ to $\omega' = (S', \chi')$; denote the last entry in $\chi$ by $(\lambda, X)$ and the new entry in $\chi'$ by $(\lambda', X')$. Since the earlier entries in $\chi$ have disjoint supports (by the second invariant in the proof of \Cref{lem:sep-pinning}), we only need to keep track of these last entries. Explicitly, for every $j < |\chi|$, the invariant $u_S(X_j) = \emptyset$ ensures that $\supp(X_j)$ is disjoint from $S$; since $S' \subseteq S$, these factors are supported outside both $S$ and $S'$, so when the trace factorizes over disjoint qubits their contributions cancel. For the last factor, cyclicity of trace gives
    \begin{align}
        \tr(e^{-\beta H_S/2}(I+\lambda X)e^{-\beta H_S/2}) = \tr(e^{-\beta H_S}) + \lambda \tr(X e^{-\beta H_S}) = \tr(e^{-\beta H_S}) + \lambda \tr(\tr_{[n]\setminus S}(X) e^{-\beta H_S}),
    \end{align}
    and similarly for $S', X'$. Hence, the ratio is
    \begin{align}\label{eq:sigmaratio}
        \frac{\Tr(\sigma_{t+1})}{\Tr(\sigma_t)} = \frac{\tr(e^{-\beta H_{S'}})}{\tr(e^{-\beta H_S})} \frac{1+\lambda' q_{S'}(X')}{1 + \lambda q_S(X)}, \qquad q_S(X) = \frac{\tr[\tr_{[n]\setminus S}(X) e^{-\beta H_S}]}{\tr(e^{-\beta H_S})}.
    \end{align}
    The right-hand side can be estimated using our prior results. Since $H_S$ and $H_{S'}$ are $(s,k)$-long-range Pauli Hamiltonians, we can use \Cref{thm:expect} to estimate the partition functions; we evaluate the expectations of $X, X'$ using \Cref{thm:globalpauli}. If we apply each estimator to accuracy $c\epsilon$ and failure probability $\delta/4$, then we claim that we can estimate the entire quantity to error $\epsilon$ with failure probability $\delta$; this follows from a Lipschitz estimate. For the ratio of partition functions, since $T = S \setminus S'$ and $T \subseteq \supp(P_{a_*})$ for some $a_* \in \cA$, we have that $\norm{H_S - H_{S'}} \leq sk$, which implies that $|\log \tr e^{-\beta H_{S'}} - \log \tr e^{-\beta H_S}| \leq \beta s k \leq 1/4096e$. For the second ratio, we use $1+\lambda' q_{S'}(X') \leq 2$ and $1+\lambda q_S(X) \geq 1/2$, since $|q_S(X)|, |q_{S'}(X')| \leq 1$ and the third invariant of \Cref{lem:sep-pinning} giving $|\lambda| \leq 1/2$ and $|\lambda'| \leq 1$. This also implies the claimed bound \cref{eq:rhobound}
    \begin{align}
        0 \leq \frac{\tr(e^{-\beta H_{S'}})}{\tr(e^{-\beta H_S})} \frac{1+\lambda' q_{S'}(X')}{1 + \lambda q_S(X)} \leq e^{1/4096e} \cdot \frac{2}{1/2} \leq 5.
    \end{align}
    We now bound $c$ conditioned on the event that each estimator succeeded. Since each log partition function is estimated to additive error $c\epsilon$, the relative error of the exponentiated ratio is at most $e^{2c\epsilon} - 1$. For $c\epsilon \leq 1/100$, this satisfies $e^{2c\epsilon} - 1 \leq \epsilon/30$. The derivatives with respect to $q_S$ and $q_{S'}$ in \cref{eq:sigmaratio} also have magnitude at most 10, and thus choosing $c = 1/100$ is enough to ensure $|\wh r - r| \leq \epsilon$.
    
    The preprocessing of $\cA(x)$ and $S(x)$ required for \Cref{thm:expect} and \Cref{thm:globalpauli} costs $O(|\cA|k)$. The two calls to \Cref{thm:expect} are dominated by the two calls to \Cref{thm:globalpauli} for Paulis of weight at most $n$, giving the claimed runtime.
\end{proof}

\begin{proof}[Proof of \Cref{thm:samp}]
We can now complete the proof of the algorithm by constructing a sampler that works by rejection sampling. For $\rho=5$ given by \Cref{lem:rh}, choose
\begin{align}\label{eq:propparams}
    L = \left\lceil\log_2 \frac{4000\rho n}{\epsilon}\right\rceil, \qquad \delta_{\rm prop} = 2^{1-L}, \qquad T=\left\lceil 4\rho \log \frac{1000n}{\epsilon}\right\rceil.
\end{align}
At current state $\omega$, make $T$ trials of the following proposal, stopping at the first accepted proposal. Draw $\omega'$ from \Cref{lem:prop-alg} truncated to $L$ and compute $\wh r(\omega')$ from \Cref{lem:rh} with accuracy $\epsilon_{\rm rej}$ and failure probability $\delta_{\rm rej}$ given by
\begin{align}\label{eq:rejparams}
    \epsilon_{\rm rej} = \frac{\epsilon}{1000n}, \quad \delta_{\rm rej} = \frac{\epsilon}{1000nT}.
\end{align}
Accept with probability $\operatorname{clip}_{[0,1]}(\wh r/\rho)$. If no proposal is accepted in $T$ trials, return an arbitrary product stabilizer state and declare failure.

We compare one accepted step of this algorithm with an ideal pinning step starting from the current state $\omega$. Let $P_\omega(d\omega')$ be the ideal, untruncated proposal distribution over $\omega'$ from \Cref{lem:sep-pinning}, and let $P^{L}_\omega(d\omega')$ be the proposal distribution obtained by truncating the two propagator calls to $L$. (We ignore the constraint of $T$ trials for now.) Here $d\omega'$ denotes summation over the possible next states $\omega'=(S',\chi')$. As previously observed, truncation to $L$ introduces TVD error at most $2^{1-L}$ (since $t$ was sampled with probability $2^{-t}$ in \Cref{lem:prop-alg}), so for our parameters \cref{eq:propparams} we have
\begin{align}\label{eq:pl}
    \norm{P_\omega-P^L_\omega}_{\rm TV}\leq \delta_{\rm prop}.
\end{align}
Define $r_\omega(\omega')=\Tr(\sigma_{\omega'})/\Tr(\sigma_\omega)$ as the quantity that $\wh r(\omega')$ is estimating; by \cref{eq:ratiocond}, we have that $\int r_\omega(\omega')P_\omega(d\omega')=1$. Condition on the event that $|\wh r(\omega') - r_\omega')| \leq \epsilon_{\rm rej}$ for every estimator $\wh r$; by \Cref{lem:rh}, this gives (taking the expectation over the internal randomness of the estimator)
\begin{align}\label{eq:abarerr}
    \abs{\E{{\rm clip}_{[0,1]}(\wh r/\rho) \mid \omega,\omega'} - \frac{r_\omega(\omega')}{\rho}} \leq \frac{\epsilon_{\rm rej}}{\rho} + \delta_{\rm rej}
\end{align}
and thus the acceptance probability of one trial is
\begin{align}
    A_\omega = \int \E{{\rm clip}_{[0,1]}(\wh r/\rho) \mid \omega,\omega'} P_\omega^L(d\omega'),
\end{align}
so the law of the accepted proposal without the cap of $T$ trials is $\E{{\rm clip}_{[0,1]}(\wh r/\rho) \mid \omega,\omega'} P_\omega^L(d\omega') / A_\omega$. By \Cref{lem:rh}, $0 \leq r_\omega(\omega')/\rho \leq 1$ and $\int r_\omega(\omega') P_\omega(d\omega')/\rho = 1/\rho$; combined with \cref{eq:pl,eq:abarerr}, this gives for every event $B$ of possible next states
\begin{align}
    \abs{\int_B \E{{\rm clip}_{[0,1]}(\wh r/\rho) \mid \omega,\omega'} P_\omega^L(d\omega') - \int_B \frac{r_\omega(\omega')}{\rho}P_\omega(d\omega')} \leq \delta_{\rm prop} + \frac{\epsilon_{\rm rej}}{\rho} + \delta_{\rm rej}.
\end{align}
Taking $B$ to be the whole state space gives
\begin{align}\label{eq:acc}
    \abs{A_\omega - \frac{1}{\rho}} \leq  \delta_{\rm prop} + \frac{\epsilon_{\rm rej}}{\rho} + \delta_{\rm rej}, \qquad A_\omega \geq \frac{1}{2\rho}.
\end{align}
For an arbitrary event $B$, we thus have
\begin{align}
    &\abs{\int_B \frac{\E{{\rm clip}_{[0,1]}(\wh r/\rho) \mid \omega,\omega'} P_\omega^L(d\omega')}{A_\omega} - \int_B r_\omega(\omega') P_\omega(d\omega')} \leq \frac{1}{A_\omega}\lr{\delta_{\rm prop} + \frac{\epsilon_{\rm rej}}{\rho} + \delta_{\rm rej}} + \frac{|A_\omega-1/\rho|}{A_\omega} \notag\\
    &\leq \frac{2}{A_\omega}\lr{\delta_{\rm prop} + \frac{\epsilon_{\rm rej}}{\rho} + \delta_{\rm rej}} \leq \frac{\epsilon}{100n},\label{eq:accepted-kernel-tv}
\end{align}
where the final inequality follows from plugging in $\delta_{\rm prop} \leq \epsilon/2000\rho n$, $\epsilon_{\rm rej}/\rho = \epsilon/1000\rho n$ and $\delta_{\rm rej} \leq \epsilon/4000\rho n$ for $T \geq 4\rho$. Crucially, \cref{eq:accepted-kernel-tv} holds for every current state $\omega$, even if $\omega$ was generated adaptively by preceding approximate steps. The total variational error over at most $n$ accepted pinning steps is therefore at most $\epsilon / 100$.

We can now account for only using $T$ trials. Since the trials are conditionally independent and identically distributed given $\omega$, conditional on at least one acceptance among the first $T$ trials, the first accepted proposal has the same law $\E{{\rm clip}_{[0,1]}(\wh r/\rho) \mid \omega,\omega'} P_\omega^L(d\omega') / A_\omega$ as above. Since each trial accepts with probability at least $1/(2\rho)$ by
\cref{eq:acc}, the probability no proposal is accepted in $T$ trials is at most $(1-1/(2\rho))^T\leq \epsilon/(1000n)$. Over the $\leq n$ steps of the pinning procedure, the probability of failing from the trial truncation is thus at most $\epsilon/1000$. Adding all errors gives
\begin{align}
    \frac{\epsilon}{100}
    +
    \frac{\epsilon}{1000}
    \leq
    \epsilon
\end{align}
in total variation distance from the ideal transcript law of $\omega_1, \dots, \omega_T$. We will shortly see that this upper-bounds the trace distance between the algorithm's state and the Gibbs state.

We now prepare the actual product state from the final $\chi_T = ((\lambda_1,X_1),\dots,(\lambda_T,X_T))$. By \Cref{lem:sep-pinning}, the Gibbs state is represented as a product of $\prod_j(I+\lambda_jX_j)$ for $|\lambda_j| \leq 1$ and Pauli operators $X_j$ with pairwise disjoint supports. We sample a product stabilizer state from each factor independently using
\cref{eq:icp} of \Cref{lem:sep-pauli-factor}: return the maximally mixed
state on $\supp(X_j)$, implemented for instance by a uniformly random product
state in the $Z$ basis, with probability $1-|\lambda_j|$. Otherwise, return a
uniformly random product eigenstate of $X_j$ with eigenvalue $\sgn(\lambda_j)$.
Conditioned on $\omega$, this procedure gives a pure product stabilizer state
whose expectation is $\sigma_T(\omega)/\Tr(\sigma_T(\omega))$. In the ideal case, this is precisely the Gibbs state by \cref{eq:invar1}: the transcript law $D(d\omega)$ must satisfy
\begin{align}
    \int \frac{\sigma_T(\omega)}{\Tr(\sigma_T(\omega))}D(d\omega) =\frac{e^{-\beta H}}{\Tr(e^{-\beta H})}.
\end{align}
Replacing $D$ by a transcript law within total variation distance $\epsilon$ changes the expected output state by trace distance at most $\epsilon$. 
This follows from a coupling argument: for $D, \wt D$ within $\epsilon$ TVD, couple entries $\omega, \wt \omega$ with $\pr{\omega\neq\wt\omega} \leq \epsilon$ and check the output states $\sigma_T(\omega)$:
\begin{align}
    \frac{1}{2}\norm{\int \frac{\sigma_T(\omega)}{\Tr \sigma_T(\omega)} D(d\omega) - \int \frac{\sigma_T(\omega)}{\Tr \sigma_T(\omega)} \wt D(d\omega)} \leq \E{\frac{1}{2} \norm{\frac{\sigma_T(\omega)}{\Tr \sigma_T(\omega)} - \frac{\sigma_T(\wt\omega)}{\Tr \sigma_T(\wt\omega)}}_1} \leq \pr{\omega\neq\wt\omega} \leq \epsilon.
\end{align}
Finally, we record the runtime. The algorithm makes at most $nT=\wt O(n)$ proposals. For each proposal, the two truncated propagator calls cost $\wt O(L|\cA|k)$ by \Cref{lem:prop-alg}, and the ratio estimate is called with accuracy $\epsilon_{\rm rej}, \delta_{\rm rej}$ given by \cref{eq:rejparams}. \Cref{lem:rh} therefore gives total runtime, up to logarithmic factors, bounded by
\begin{align}
    \wt O\lr{n|\cA|k + \frac{kn^{11}}{\epsilon^5} \log |\cA|}.
\end{align}
\end{proof}

\section{Acknowledgments}
The authors thank Ainesh Bakshi, Fernando Brandão, Chi-Fang Chen, Sitan Chen, Bobak Kiani, and Saeed Mehraban for useful discussions. The authors also thank Jorge Sánchez-Segovia and Álvaro M. Alhambra for sharing a draft of related work after the v1 posting of the paper. ChatGPT Pro $\leq 5.5$ contributed significantly to the technical content of \Cref{sec:upper,sec:nonlocal} and to improving the $k$-dependence in \Cref{thm:sep}; ChatGPT Pro 6 was used in v2 to improve the $k$ dependence of \Cref{thm:zf,lem:pin-zf,thm:zf-upper}. The remaining central proof ideas in the main text were human-made, although ChatGPT helped with proof development (including nontrivial but elementary aspects such as \Cref{fac:pincum}) and in finding errors. The human authors checked and refined all portions of the proofs developed with AI assistance. The authors wrote the manuscript and take full responsibility for the correctness of the paper. AZ was supported by a Hertz Fellowship and a grant from the Simons Foundation (MP-SIP-00001553, AWH); this work was done while AZ was visiting the Simons Institute for the Theory of Computing. HP was supported by the Department of Defense through the National Defense Science and Engineering Graduate (NDSEG) Fellowship Program.

\bibliography{main}
\bibliographystyle{alpha}
\newpage
\appendix

\section{Upper bounds on thresholds}
\label{sec:upper}
\subsection{Separability and stabilizerness}
\begin{theorem}[Entangled Hamiltonians]
\label{thm:sep-upper}
Define the $D$-dimensional lattice $\Lambda_L = \{1,\dots,L\}^D$ and fix $m \geq 2$. At each site $x \in \Lambda_L$, and indices $1 \leq i < j \leq m$, introduce qubits denoted by $q_{x,i,j,\ell}$ for $\ell \in \{1,2\}$, and introduce control qubits $c_{x,1},\dots,c_{x,m}$. For coefficients $J_{x,y} = J_{y,x} \geq 0$ indexed by $x, y \in \Lambda_L$ and $a \geq 0$, define
\begin{align}
    H = -a \sum_{x \in \Lambda_L} \sum_{i=1}^m  Q_{x,i} - \sum_{x,y\in\Lambda_L} \sum_{i=1}^m J_{x,y} Q_{x,i} Z_{c_{y,i}}
\end{align}
for operators
\begin{align}
    Q_{x,i} = \lr{\prod_{j < i} Z_{q_{x,j,i,1}}}\lr{\prod_{j > i} X_{q_{x,i,j,1}}}\lr{\prod_{j < i} Z_{q_{x,j,i,2}}}\lr{\prod_{j > i} X_{q_{x,i,j,2}}}.
\end{align}
This is a commuting Hamiltonian with locality $k=2m-1$ and is entangled if for some $x \in \Lambda_L$, $m \tanh(\beta(a+\sum_{y\in\Lambda_L} J_{x,y})) > 1$.
\begin{itemize}
    \item \emph{Long-range Hamiltonians.} For every $k \geq 2$ and $s > 0$, setting $m = \lfloor k/2\rfloor + 1$, $a=s/2$ and $J_{x,y} = 0$ results in an $(s,k)$-long-range Pauli Hamiltonian. For all
    \begin{align}
        \beta \geq \frac{2}{s}\arctanh \frac{3}{2m} = \Theta\lr{\frac{1}{sk}}
    \end{align}
    and sufficiently large $L$, the Gibbs state is entangled and has trace distance at least $1/2$ from any separable state.
    \item \emph{Low-intersection Hamiltonians.} Setting $m=2, a=0, J_{x,y}=1$ results in a $(\mathfrak d,3)$-low-intersection Pauli Hamiltonian with $\mathfrak d = 3L^D-2$. For all
    \begin{align}
        \beta > \frac{96 \log 2}{\mathfrak d + 2} = \Theta\lr{\frac{1}{\mathfrak d}}
    \end{align}
    and sufficiently large $L$, the Gibbs state is entangled and has trace distance at least $1/12$ from any separable state.
    \item \emph{Power-law Hamiltonians.} Set $m=2, a=0, J_{x,x}=0$ and for $x \neq y$, $J_{x,y} = 1/(1+\norm{x-y}_2)^\alpha$. For every fixed $\beta > 0$ and $\alpha \leq D$, the Gibbs state is entangled for all sufficiently large $L$. (Note that the local interaction strength remains uniformly bounded in $L$ when $\alpha > D$ and diverges when $\alpha \leq D$.)
\end{itemize}
\end{theorem}
\begin{proof}
    We begin by checking the basic properties of the Hamiltonian claimed in the theorem statement.
    \begin{itemize}
        \item \emph{Commuting.} Letting $P_{x,i,\ell} = \lr{\prod_{j < i} Z_{q_{x,j,i,\ell}}} \lr{\prod_{j > i} X_{q_{x,i,j,\ell}}}$, we can rewrite $Q_{x,i} = P_{x,i,1} P_{x,i,2}$. These operators satisfy $P_{x,i,\ell} P_{x,j,\ell} = - P_{x,j,\ell} P_{x,i,\ell}$ because they only overlap nontrivially on qubit $q_{x,i,j,\ell}$; they commute for different values of $\ell$. Hence, we have $Q_{x,i}Q_{x,j} = (-1)^2Q_{x,j}Q_{x,i} = Q_{x,j}Q_{x,i}$; since the qubits $c_{y,i}$ are acted only in the $Z$ basis, this suffices to show that the Hamiltonian commutes.
        \item \emph{Locality.} Each $P_{x,i,\ell}$ has weight $m-1$ and appears twice, and that $Z_{c_{y,i}}$ appears once. Hence, the locality is $k=2m-1$.
        \item \emph{Interaction strength.} Each $q_{x,i,j,\ell}$ appears in terms $Q_{x,i}$ and $Q_{x,j}$, and for each $y$ in $Q_{x,i}Z_{c_{y,i}}$ and $Q_{x,j}Z_{c_{y,j}}$. Hence, it receives interaction strength $2a + 2\sum_y J_{x,y}$. Each qubit $c_{y,i}$ has strength $\sum_x J_{x,y}$ by symmetry of $J$. We conclude that $s \leq 2(a+\max_x \sum_y J_{x,y})$.
    \end{itemize}
    We now show that the state must be entangled for all $\beta$ such that $m \tanh(\beta(a+\sum_{y\in\Lambda_L} J_{x,y})) > 1$. For the sake of contradiction, assume $\rho_\beta = \sigma_A \otimes \sigma_B$ for states $\sigma_A, \sigma_B$ defined on the regions
    \begin{align}\label{eq:ab}
        A=\{q_{x,i,j,1}\}_{x\in\Lambda_L,1\leq i< j\leq m} \cup \{c_{y,i}\}_{y\in\Lambda_L,i\in[m]}, \quad B=\{q_{x,i,j,2}\}_{x\in\Lambda_L,1\leq i< j\leq m}.
    \end{align}
    By convexity of the set of separable states, it suffices to consider this case. Cauchy-Schwarz gives
    \begin{align}\label{eq:abcs}
        \abs{\Tr((\sigma_A\otimes \sigma_B)\sum_{i=1}^m Q_{x,i})} = \abs{\sum_{i=1}^m \Tr(P_{x,i,1}\sigma_A) \Tr(P_{x,i,2}\sigma_B)} \leq \sqrt{\lr{\sum_{i=1}^m \Tr(P_{x,i,1}\sigma_A)^2} \lr{\sum_{i=1}^m \Tr(P_{x,i,2}\sigma_B)^2}}.
    \end{align}
    Since pairwise anticommuting operators $\{C_i\}$ satisfy for any state $\sigma$ that
    \begin{align}
        \sum_i \Tr(C_i\sigma)^2 = \Tr[\sum_i \Tr(C_i\sigma) C_i \sigma] \leq \norm{\sum_i \Tr(C_i\sigma) C_i} = \sqrt{\sum_i \Tr(C_i\sigma)^2},
    \end{align}
    where the inequality used the fact that $\Tr(\sigma)=1$ and $\sigma \geq 0$, and the final equality used $(\sum_i \Tr(C_i\sigma) C_i)^2 = (\sum_i \Tr(C_i\sigma)^2)I$ by anticommutation. Since $\sum_i \Tr(C_i\sigma)^2 \leq \sqrt{\sum_i \Tr(C_i\sigma)^2}$ we conclude that $\sum_i \Tr(C_i\sigma)^2 \leq 1$, i.e., \cref{eq:abcs} satisfies
    \begin{align}\label{eq:vcond}
        \abs{\Tr((\sigma_A\otimes \sigma_B)\sum_{i=1}^m Q_{x,i})} \leq 1
    \end{align}
    for every $x \in \Lambda_L$, since $\{P_{x,i,1}\}_i$ are pairwise anticommuting. Similarly, we find that for every $x, y \in \Lambda_L$,
    \begin{align}\label{eq:wcond}
        \abs{\Tr(\sum_{i=1}^m Q_{x,i} Z_{c_{y,i}} \sigma)} \leq 1.
    \end{align}
    We now contradict the condition \cref{eq:vcond}. Measure every qubit $c_{y,i}$ in the $Z$ basis and postselect on the all $+1$ outcome, i.e., on the Gibbs state of the Hamiltonian
    \begin{align}
        H_+ = -\sum_{x \in \Lambda_L}\lr{a + \sum_{y \in \Lambda_L} J_{x,y}} \sum_{i=1}^m Q_{x,i}.
    \end{align}
    This is noninteracting over $x \in \Lambda_L$ and thus we can identify the Gibbs state on an individual $x$ as
    \begin{align}
        \tau_\beta = 2^{-2 \binom{m}{2}} \prod_{i=1}^m \lr{I+\tanh(\beta(a+\sum_{y \in \Lambda_L} J_{x,y})) Q_{x,i}}.
    \end{align}
    To obtain the contradiction, we compute
    \begin{align}
        \Tr(\tau_\beta \sum_{i=1}^m Q_{x,i}) = m \tanh(\beta(a+\sum_{y \in \Lambda_L} J_{x,y}))
    \end{align}
    and conclude that the state must be entangled when $m \tanh(\beta(a+\sum_{y \in \Lambda_L} J_{x,y})) > 1$ since measuring and postselecting on qubits in $A$ cannot create entanglement between $A$ and $B$ if we started in a separable state.

    It remains to show the specific claims for the different families of Hamiltonians.
    \begin{itemize}
        \item \emph{Long-range Hamiltonians.} The Gibbs state $\rho_\beta$ of $H = -\frac{s}{2}\sum_x\sum_{i=1}^m Q_{x,i}$ satisfies for $\beta \geq \frac{2}{s}\arctanh \frac{3}{2m}$
        \begin{align}\label{eq:qbeta}
            \Tr(Q_{x,i}\rho_\beta) = \tanh(\frac{\beta s}{2}) \geq \frac{3}{2m}.
        \end{align}
        We compare $\rho_\beta$ to a state $\sigma = \sigma_A \otimes \sigma_B$ for regions $A$ and $B$ of \cref{eq:ab}. For each $x$, we draw uniformly random $i_x \in [m]$ and measure $P_{x,i_x,1}$ and $P_{x,i_x,2}$ and take their product $Y_x \in \{\pm1\}$. By \cref{eq:qbeta}, this random variable satisfies $\EE_{\rho_\beta} Y_x \geq 3/(2m)$. Consider the event $E$ that $\frac{1}{L^D} \sum_{x \in \Lambda_L} Y_x \geq \frac{5}{4m}$. Hoeffding's inequality gives that the probability $E^c$ occurs under $\rho_\beta$ is at most $\exp[-2\left(L^D/(4m)\right)^2/4L^D] = \exp[-L^D/32m^2]$. On the other hand, for a separable state $\sigma=\sigma_A\otimes\sigma_B$, \cref{eq:vcond} implies $\EE_\sigma Y_x \leq 1/m$. Indeed, this inequality holds even when conditioning on measurement outcomes of prior sites $x \in \Lambda_L$, since the system is noninteracting. We can thus apply Azuma's inequality to obtain that the probability $E$ occurs under $\sigma$ is also at most $\exp[-L^D/32m^2]$. The difference lower-bounds the trace distance between $\rho_\beta$ and $\sigma$:
        \begin{align}
            \frac{1}{2}\norm{\rho_\beta-\sigma}_1 \geq \Pr_{\rho_\beta}(E)-\Pr_\sigma(E) \geq 1-2\exp[-\frac{L^D}{32m^2}] \geq \frac{1}{2}
        \end{align}
        where the final inequality holds for sufficiently large $L^D \geq 64m^2$.
        \item \emph{Low-intersection Hamiltonians.} When $m=2$, we can relabel the qubits $q_{x,i,j,\ell}$ by $q_{x,i,\ell}$ since $i < j$. The Hamiltonian $H = -\sum_{x,y\in\Lambda_L}\sum_{i=1}^2 Q_{x,i}Z_{c_y,i}$ is then composed of terms $Q_{x,1} = X_{q_{x,1,1}} X_{q_{x,1,2}}$ and $Q_{x,2} = Z_{q_{x,1,1}} Z_{q_{x,1,2}}$. There are $2L^{2D}$ Hamiltonian terms; a term $Q_{x,i}Z_{c_{y,i}}$ overlaps all $2L^D$ terms with the same $x$ and all $L^D$ terms with the same $c_{y,i}$. These two collections only intersect on $Q_{x,i}Z_{c_{y,i}}$ itself, and thus the dual degree is $\mathfrak d = 3L^D - 2$ as claimed. Fix some $i \in \{1, 2\}$ (which doesn't matter since they correspond to two copies of the system that are noninteracting). We diagonalize the Hamiltonian to obtain Gibbs measure $\pi_\beta(\wh Q, \wh Z) \propto \exp[\beta \lr{\sum_x \wh Q_x}\lr{\sum_y \wh Z_y}]$ for eigenvalues $\wh Q_x \in \{\pm1\}$ and $\wh Z_y \in \{\pm 1\}$ of $Q_{x,i}$ and $Z_{c_{y,i}}$. We bound the distance from the Gibbs state to a separable state at $\beta = C_*/L^D$ for $C_* = 32 \log 2$ using the random variable $R = L^{-2D}\lr{\sum_x \wh Q_x}\lr{\sum_y \wh Z_y}$ for $\wh Q, \wh Z$ drawn from $\pi_\beta$. The all $+1$ and all $-1$ configurations for $\wh Q, \wh Z$ contribute $2e^{C_*L^D}$ to the normalization of $\pi_\beta$, and the total number of configurations is $4^{L^D}$; hence, we have that $\pr{R \leq 3/4} \leq 4^{L^D} e^{3C_*L^D/4} / 2e^{C_* L^D} = 2^{-1-6L^D}$ for $C_* = 32\log 2$. Since $R \geq -1$, we have
        \begin{align}
            \Tr(Q_{x,i} Z_{c_{y,i}} \rho_\beta) = \EE_{\pi_\beta} R \geq \frac{3}{4} \pr{R \geq \frac{3}{4}} - \pr{R \leq \frac{3}{4}} > \frac{3}{4} - 2^{-6L^D} \geq \frac{2}{3},
        \end{align}
        where the last inequality holds for sufficiently large $L$. Since $m=2$, we have $\Tr(\sum_{i=1}^m Q_{x,i} Z_{c_{y,i}} \rho_\beta) > 4/3$ and $\norm{\sum_{i=1}^m Q_{x,i} Z_{c_{y,i}}} \leq m = 2$. Applying \cref{eq:wcond} for a separable state $\sigma_A \otimes \sigma_B$, this gives trace distance bound
        \begin{align}
            \frac{1}{2}\norm{\rho_\beta - \sigma_A \otimes \sigma_B} \geq \frac{\Tr(\sum_{i=1}^m Q_{x,i} Z_{c_{y,i}} (\rho_\beta-\sigma_A\otimes\sigma_B))}{4} \geq \frac{4/3-1}{4} = \frac{1}{12}.
        \end{align}
        \item \emph{Power-law Hamiltonians.} As above, we set $m=2$ so $H = -\sum_{x,y\in\Lambda_L}\sum_{i=1}^2 Q_{x,i}Z_{c_y,i}$ for terms $Q_{x,1} = X_{q_{x,1,1}} X_{q_{x,1,2}}$ and $Q_{x,2} = Z_{q_{x,1,1}} Z_{q_{x,1,2}}$. When $\alpha \leq D$, we have for a typical site $x_0$ that $\sum_{y \neq x_0} J_{x,y} = \Omega(L^{D-\alpha})$ for $\alpha < D$ and $\Omega(\log L)$ for $\alpha = D$. Both of these diverge, and thus the entanglement condition $m \tanh(\beta(a+\sum_{y \in \Lambda_L} J_{x,y})) > 1$ is satisfied for any fixed $\beta$ given sufficiently large $L$. Conversely, for $\alpha > D$ the interaction strength is uniformly bounded: we can sum over shells containing all lattice sites distance $r$ away from $x_0$, which has size $\Theta(r^{D-1})$. This gives $\sum_{y\neq x_0} (1+\norm{x-y}_2)^{-\alpha} \leq \sum_{r \geq 1}O(r^{D-1}/(1+r)^\alpha) < \infty$ for $\alpha > D$.
    \end{itemize}
\end{proof}

\begin{theorem}[Long-range Pauli Hamiltonian magic]\label{thm:stab-upper}
For every $k\geq2$ and $w_{\rm free},w_{\rm pert}>0$, there is a $\lr{(w_{\rm free}+w_{\rm pert})/k,k}$-long-range Pauli
Hamiltonian with anticommutation parameters
$w_{\rm free},w_{\rm pert}$ whose Gibbs state has magic for all
\begin{align}
    \beta > \frac{1}{\sqrt{w_{\rm pert}^2+w_{\rm free}^2}}
    \arctanh\lr{
        \frac{\sqrt{w_{\rm pert}^2+w_{\rm free}^2}}
        {w_{\rm pert}+w_{\rm free}}
    }.
    \label{eq:betaplus}
\end{align}
Hence, $\beta_{\rm stab}(s,k;\epsilon) = O(\log(1/\epsilon)/sk)$.
\end{theorem}
\begin{proof}
    Introduce $k^2$ qubits indexed by $(i,j)\in[k]^2$, and set $H=H_0+V$ for
    \begin{align}
        H_0 = -\frac{w_{\rm free}}{k} \sum_{j=1}^k \prod_{i=1}^k Z_{i,j}, \quad V = -\frac{w_{\rm pert}}{k}\sum_{i=1}^k\prod_{j=1}^k X_{i,j}.
    \end{align}
    Every qubit belongs to exactly one row and one column, so $H$ is a $\lr{(w_{\rm free}+w_{\rm pert})/k,k}$-long-range Pauli Hamiltonian as claimed. Moreover, every term in $V$ anticommutes with every term in $H_0$, giving anticommutation parameters $w_{\rm free}$ and $w_{\rm pert}$ per \cref{eq:wfp}.
    
    Define the stabilizer projector
    \begin{align}
        \Pi = \lr{\prod_{i=2}^k\frac{I + \prod_{j=1}^k X_{1,j} X_{i,j}} {2}}\lr{\prod_{j=2}^k\frac{I+\prod_{i=1}^k Z_{i,1} Z_{i,j}}{2}}.
    \end{align}
    One can check that for the Hadamard gate $\wh H$, the Clifford unitary
    \begin{align}
        U = \lr{\prod_{i=2}^k \wh H_{i,1}} \lr{\prod_{j=2}^k \prod_{i=2}^k {\rm CNOT}_{(i,j)\to(1,j)}} \lr{\prod_{i=2}^k {\rm CNOT}_{(i,1)\to(1,1)}} \lr{\prod_{i=1}^k \prod_{j=2}^k {\rm CNOT}_{(i,1)\to(i,j)}}
    \end{align}
    satisfies
    \begin{align}
        U\Pi U^\dagger = I_{(1,1)} \otimes \lr{\bigotimes_{2 \leq i \leq k}\ketbra{0}_{(i,1)}} \otimes \lr{\bigotimes_{2 \leq j \leq k}\ketbra{0}_{(1,j)}} \otimes \lr{\bigotimes_{2 \leq i,j \leq k} I_{(i,j)}}
    \end{align}
    and
    \begin{align}
        \prod_{j=1}^k X_{i,j}\Pi = \prod_{j=1}^k X_{1,j}\Pi, \quad \prod_{i=1}^k Z_{i,j}\Pi = \prod_{i=1}^k Z_{i,1}\Pi.
    \end{align}
    This gives
    \begin{align}
        U\Pi H \Pi U^\dagger = -\lr{w_{\rm free} Z + w_{\rm pert} X}_{(1,1)} \otimes \lr{\bigotimes_{2 \leq i \leq k}\ketbra{0}_{(i,1)}} \otimes \lr{\bigotimes_{2 \leq j \leq k}\ketbra{0}_{(1,j)}} \otimes \lr{\bigotimes_{2 \leq i,j \leq k} I_{(i,j)}}.
    \end{align}
    Because $[H, \Pi] = 0$, the reduced density matrix on the first qubit of the state proportional to $U\Pi\rho_\beta\Pi U^\dagger$ is
    \begin{align}
        \psi = \frac{e^{\beta(w_{\rm free}Z+w_{\rm pert}X)}}{\Tr e^{\beta(w_{\rm free}Z+w_{\rm pert}X)}}.
    \end{align}
    Going from $\rho_\beta$ to $\psi$ only consisted of measuring commuting Pauli observables and postselecting on their $+1$ outcomes, applying a Clifford unitary, and discarding qubits. Since this cannot create magic from a convex combination of stabilizer states, it suffices to show magic in $\psi$ to conclude that $\rho_\beta$ also has magic. We do this by checking the 1-qubit stabilizer polytope. The Bloch vector of $\psi$ has $\ell_1$ norm
    \begin{align}
        r = \frac{w_{\rm pert} + w_{\rm free}}{\sqrt{w_{\rm pert}^2 + w_{\rm free}^2}} \tanh\lr{\beta\sqrt{w_{\rm pert}^2 + w_{\rm free}^2}},
    \end{align}
    which is only stabilizer for $r \leq 1$. The condition $r > 1$ is precisely the bound reported in \cref{eq:betaplus}. Finally, choosing $w_{\rm free} = (1-\epsilon) sk$ and $w_{\rm pert} = \epsilon sk$ yields magic at $\beta = O(\log(1/\epsilon)/sk)$ for asymptotically small $\epsilon$.
\end{proof}

\subsection{Nonlocal Pauli Hamiltonians}
\label{sec:no_go_high_low_weight_mix}
\begin{theorem}[Entangled nonlocal Hamiltonians]
    For any constant $g, h, \beta > 0$ and non-decreasing $w(n) = e^{o(n)}>0$, the Gibbs state of the Hamiltonian
    \begin{align}
        H = \sum_{a\in\cA} h_a = -h\sum_{i} Z_i - \frac{g}{w(n)}\prod_{i} X_i
    \end{align}
    is entangled for sufficiently large $n$. Furthermore, the Hamiltonian satisfies
    \begin{align}
        \sup_x \sum_{a:x\in\supp(h_a)} \Vert h_a \Vert w(|\supp(h_a)|) = h w(1)+g < \infty.
    \end{align}
\end{theorem}
\begin{proof}
We use shorthand $|a| = |\supp(h_a)|$. Direct evaluation gives
\begin{align}
    \sup_x \sum_{a:x\in\supp(h_a)} \Vert h_a \Vert w(|a|) = h w(1)+g.
\end{align}
To show that the Gibbs state is entangled we evaluate the witness
\begin{align}
    W_t = \ketbra{\phi_t}^{T_1}, \quad \ket{\phi_t} = \sqrt{t}\ket{10^{n-1}} - \frac{1}{\sqrt t}\ket{01^{n-1}}
\end{align}
for $t> 0$ and partial transpose $T_1$ on the first qubit (i.e., $(\ketbra{a}{b}\otimes M)^{T_1} = \ketbra{b}{a} \otimes M$). We check the witness by considering its expectation for a generic product state given by
\begin{align}
    |\psi\rangle = \bigotimes_{i=1}^n \left(\alpha_i |0\rangle + \beta_i |1\rangle\right).
\end{align}
The expectation of the witness is
\begin{align}
    \langle \psi|W_t|\psi\rangle = t |\beta_1 A|^2 + t^{-1} |\alpha_1 B|^2 - \beta_1^* B^* \alpha_1 A - \alpha_1^* A^* \beta_1 B \geq \left(\sqrt{t} |\beta_1 A| - \frac{1}{\sqrt{t}} |\alpha_1 B|\right)^2 ,
\end{align}
where $A = \prod_{i\geq 2} \alpha_i$ and $B = \prod_{i\geq 2} \beta_i$.  Hence, $W_t$ is nonnegative on every separable state. Next we compute the expectation of the witness on the Gibbs state.  Since the Hamiltonian $Z$ does not flip bits, whereas $X^{\otimes n}$ flips every bit, the Hamiltonian preserves subspaces spanned by bit-strings $|z\rangle$ and their complement $|\bar{z}\rangle$.  Restricted to this basis the Hamiltonian can be written as 
\begin{align}
    H|_z = -hm |z\rangle \langle z|  - \lambda |z\rangle \langle \bar z | - \lambda |\bar z\rangle \langle  z | + hm |\bar z\rangle \langle \bar z|
\end{align}
where $m=n-2k$ with $k = |z|$ the Hamming weight of $z$ and $\lambda = g/w(n)$.  The Gibbs state restricted to this sector can then be computed as 
\begin{align}
    e^{-\beta H}|_z = \cosh(\beta \epsilon_{m}) I - \frac{\sinh(\beta \epsilon_{m})}{\epsilon_{m}} \left(-hm |z\rangle \langle z|  - \lambda |z\rangle \langle \bar z | - \lambda |\bar z\rangle \langle  z | + hm |\bar z\rangle \langle \bar z|\right)
\end{align}
with $\epsilon_{m}= \sqrt{h^2 m^2 + \lambda^2}$.  Since the Hamiltonian can be broken into disjoint parts which operate in different subspaces the same can be done for the Gibbs state.  We can then directly read off the value of the expectation value to get 
\begin{align}
    \Tr(W_t e^{-\beta H}) &= t\left( \cosh(\beta \epsilon_{n-2})+\frac{h (n-2)}{\epsilon_{n-2}}\sinh(\beta \epsilon_{n-2})\right)+t^{-1}\left( \cosh(\beta \epsilon_{n-2})-\frac{h (n-2)}{\epsilon_{n-2}}\sinh(\beta \epsilon_{n-2})\right) \notag\\
    &\quad -2\frac{\lambda}{\epsilon_{n}}\sinh(\beta \epsilon_n)
\end{align}
Using $\min_{t>0}\left(ta + t^{-1} b\right)=2\sqrt{ab}$ we get that
\begin{align}
    \min_{t> 0} \left(\Tr(W_t e^{-\beta H})\right) &= 2\left(\sqrt{1+\left[\frac{\lambda}{\epsilon_{n-2}}\sinh(\beta \epsilon_{n-2})\right]^2} - \frac{\lambda}{\epsilon_n}\sinh(\beta \epsilon_n)\right).
\end{align}
The witness is negative when
\begin{align}
    \frac{\lambda}{\epsilon_n}\sinh(\beta \epsilon_n)> \sqrt{1+\left(\frac{\lambda}{\epsilon_{n-2}}\sinh(\beta \epsilon_{n-2})\right)^2}.
\end{align}
For asymptotically large $n$, if $\log w(n) = o(n)$ then
\begin{align}
    \log \lr{\frac{\lambda}{\epsilon_n}\sinh(\beta \epsilon_n)}^2 = 2 \beta h n - 2 \log w(n) - 2 \log n + O(1) = 2\beta h n - o(n)
\end{align}
and hence
\begin{align}
    \lr{\frac{\lambda}{\epsilon_n}\sinh(\beta \epsilon_n)}^2 - \lr{\frac{\lambda}{\epsilon_{n-2}}\sinh(\beta \epsilon_{n-2})}^2 = \lr{\frac{\lambda}{\epsilon_n}\sinh(\beta \epsilon_n)}^2\lr{1 - e^{-4\beta h} + o(1)}
\end{align}
diverges, implying the witness is negative for sufficiently large $n$.
\end{proof}

\subsection{Zero-freeness}
Here, we show upper bounds on $\beta_{\rm phase}$ to complete the proof of \Cref{thm:zf}; we then also show the pinned zero-freeness radius of $O(\log(k+1)/(sk))$, matching the lower bound of \Cref{lem:pin-zf}.

\begin{theorem}[Phase transition of a long-range Pauli Hamiltonian partition function]\label{thm:zf-upper}
    There is a sequence of $(s,k)$-long-range Pauli Hamiltonians that have a non-analytic free energy density at
    \begin{align}
        \beta = \Theta\lr{\frac{1}{\max\{s, \sqrt{sk}\}}}.
    \end{align}
\end{theorem}
\begin{proof}
The Curie-Weiss Hamiltonian
$H=-\frac{s}{n}\sum_{i<j}Z_iZ_j$ has a transition at
$\beta=1/s$~\cite{ellis1985large}. We construct a second family with a
transition at $\beta=\Theta(1/\sqrt{sk})$. Take a sequence of cubic bipartite Cayley graphs
$G=(A\sqcup B,E)$ of girth at least six, with
\begin{align}
    |\partial S|&\geq h|S| \qquad (|S|\leq |V|/2), &
    |A|&=|B|=M/2, \qquad M=|V|,
\end{align}
where $h>0$ is fixed. Such a sequence is provided by cubic
bipartite Ramanujan graphs~\cite{chiu1992cubic}.
Fix integers $r\geq2$, $\ell\geq1$, and put $m=r\ell$.
For every edge $uv$, oriented with $u\in A$ and $v\in B$, introduce
qubits $(uv,a,b)$, $a,b\in[\ell]$. At each vertex introduce qubits
$y_{v,a,j}$, $a\in[\ell]$, $j\in[r]$, and $z_v$.
The Hamiltonian is defined in terms of the following operators, taking $u \in A$ and $v \in B$:
\begin{align}
    P_{u,a,j}
    =Z_{y_{u,a,j}}\prod_{v\sim u}\prod_{b=1}^{\ell}X_{uv,a,b},
     \qquad
    P_{v,b,j}
    =Z_{y_{v,b,j}}\prod_{u\sim v}\prod_{a=1}^{\ell}Z_{uv,a,b}.
\end{align}
For $Q_v = P_{v,1,1}P_{v,1,2} =Z_{y_{v,1,1}}Z_{y_{v,1,2}}$, set Hamiltonian
\begin{align}\label{eq:h0}
    H&=H_0-\sum_v\left(
        \lambda\varepsilon_vQ_v+\frac14Z_{z_v}
        +\frac18\varepsilon_vQ_vZ_{z_v}\right), \qquad H_0=-\frac14\sum_{v,a,j}P_{v,a,j}
\end{align}
for parameters
\begin{align}
    \varepsilon_v = (-1)^{\mathbf 1_{\{v\in B\}}}, \qquad \lambda = -\frac{1}{2\beta_*} \log \frac{\cosh(3\beta_*/8)}{\cosh(\beta_*/8)}, \qquad \beta_* = \frac{4t_0}{\sqrt m}
\end{align}
defined in terms of a constant $t_0 > 0$ to be chosen below. For $r=\lfloor2s\rfloor$ and
$\ell=\lfloor(k-1)/3\rfloor$, this is an $(s,k)$-long-range Pauli Hamiltonian acting on $n=M(\frac{3\ell^2}{2}+m+1)$ qubits. Moreover, both $H$ and $H_0$ commute with both $Q_v$ and $R_v=\frac1m\left(\sum_{a,j}P_{v,a,j}\right)^2$ due to the relations
\begin{align}\label{eq:expander-pauli-relations}
    P_{u,a,j}P_{v,b,k}
    =(-1)^{\mathbf1_{\{uv\in E\}}}P_{v,b,k}P_{u,a,j},\qquad
    [P_{u,a,j},R_v]=[P_{u,a,j},Q_v]=0.
\end{align}
Define the projectors $E_\pm=\mathbf1_{\{\pm\sum_v\varepsilon_vR_v\geq Mt_0^2/4\}}$ satisfy $\Tr(E_\pm e^{-\beta_*H_0})/\Tr e^{-\beta_*H_0} = \Omega(1)$. Left translation by a group element interchanging $A$ and $B$,
together with Hadamards on the edge qubits and transposition of the
edge indices $a\leftrightarrow b$, preserves $H_0$ and changes the
signs of $\sum_v\varepsilon_vR_v$ and $\sum_v\varepsilon_vQ_v$.
Consequently, $\Tr(E_+e^{-\beta H_0})=\Tr(E_-e^{-\beta H_0})$ (for real $\beta$) and by \Cref{lem:support} for its constants $t_0, m_0$ with $m \geq m_0$, we have
\begin{align}\label{eq:thirds}
    \frac{\Tr(E_\pm e^{-\beta_*H_0})}
         {\Tr e^{-\beta_*H_0}}
    &\geq\frac{1-e^{-M}}2\geq\frac13
\end{align}
for sufficiently large $M$.

We now must relate $H_0$ back to the full $H$. Tracing out the $z$-qubits gives
\begin{align}
    \Tr_z e^{-\beta H}
    &=2^M\bigl[\cosh(3\beta/8)\cosh(\beta/8)\bigr]^{M/2}
      \exp\left[-\beta H_0+F(\beta)\sum_v\varepsilon_vQ_v\right]
\end{align}
for $F(\beta)=\beta\lambda+\frac12\log\frac{\cosh(3\beta/8)}{\cosh(\beta/8)}$ satisfying
\begin{align}
    F(\beta_*)=0 , \qquad F'(\beta_*)=\frac{\beta_*}{2}\int_{1/8}^{3/8}
        u\,\operatorname{sech}^2(\beta_*u)\,du>0.
\end{align}
This implies that $\Tr(E_\pm e^{-\beta_*H})/\Tr e^{-\beta_*H} \geq 1/3$ as well.

We now evaluate the free energy density. Define the energy densities
\begin{align}
    e_{\pm,M}
    &=\frac{\Tr(E_\pm H e^{-\beta_*H})}
             {n\Tr(E_\pm e^{-\beta_*H})} 
    =-\frac1n\left.
          \frac{d}{d\beta}\log\Tr(E_\pm e^{-\beta H})
        \right|_{\beta=\beta_*}.
\end{align}
Let $p_M(\beta)=n^{-1}\log\Tr e^{-\beta H}$ be the free energy density.
Since
\begin{align}
    \frac{\|H\|}{n}
    &\leq\frac{m/4+|\lambda|+3/8}{3\ell^2/2+m+1}\leq1,
\end{align}
we have $|e_{\pm,M}|\leq1$, $|p_M'(\beta)|\leq1$, and
$p_M(0)=\log2$. We may therefore pass to a subsequence along which
$e_{\pm,M}\to e_\pm$ and
$p_M\to p$ locally uniformly on $\mathbb R$.
The projections $E_\pm$ commute with $H$, so Jensen's inequality gives,
for every $u\in\mathbb R$,
\begin{align}
    p_M(\beta_*+u)
    &\geq\frac1n\log\Tr
        (E_\pm e^{-\beta_*H}e^{-uH})
       \geq p_M(\beta_*)
        +\frac1n\log\frac{\Tr(E_\pm e^{-\beta_*H})}
                              {\Tr e^{-\beta_*H}}
        -u e_{\pm,M}
      \\
    &\geq p_M(\beta_*)-\frac{\log3}{n}-u e_{\pm,M}.
\end{align}
Taking limits and using convexity, this implies
\begin{align}
    p'(\beta_*+)-p'(\beta_*-)
    &\geq e_-e_+.
\end{align}
We now evaluate this energy difference. Introduce notation $\langle X\rangle_\pm^0=\frac{\Tr(E_\pm X e^{-\beta_*H_0})}{\Tr(E_\pm e^{-\beta_*H_0})}$.
Using (for some ordering on $V$)
\begin{align}\label{eq:expander-trace-words}
    2^{-n}\Tr\left(\prod_{i=1}^{L}P_{v_i,a_i,j_i}\right)
    &={}
    (-1)^{|\{(i,i'):i<i',\ v_i>v_{i'},\ v_iv_{i'}\in E\}|} 
    \prod_{v,a,j}
    \mathbf1_{\{|\{i:(v_i,a_i,j_i)=(v,a,j)\}|\in2\mathbb Z\}},
\end{align}
which is invariant under
permutations of the $m$ labels at each vertex, we obtain the identities
\begin{align}
    \langle P_{v,a,j}P_{v,b,k}\rangle_\pm^0
    &=\langle Q_v\rangle_\pm^0
        \qquad((a,j)\ne(b,k)),\\
    \langle R_v\rangle_\pm^0
    &=1+(m-1)\langle Q_v\rangle_\pm^0,\\
    -\left\langle\sum_v\varepsilon_vQ_v\right\rangle_-^0 &= \left\langle\sum_v\varepsilon_vQ_v\right\rangle_+^0
    =\frac1{m-1}
       \left\langle\sum_v\varepsilon_vR_v\right\rangle_+^0
      \geq\frac{Mt_0^2}{4(m-1)},\\
    \langle H_0\rangle_-^0&=\langle H_0\rangle_+^0
\end{align}
which imply
\begin{align}\label{eq:expander-energy-gap}
    e_{-,M}-e_{+,M}
    &=\frac{2F'(\beta_*)}{n}
        \left\langle\sum_v\varepsilon_vQ_v\right\rangle_+^0 \geq\frac{F'(\beta_*)t_0^2}
        {2(m-1)(3\ell^2/2+m+1)}>0.
\end{align}
Thus
\begin{align}
    p'(\beta_*+)-p'(\beta_*-)
    &\geq e_-e_+
     \geq\frac{F'(\beta_*)t_0^2}
        {2(m-1)(3\ell^2/2+m+1)}>0
\end{align}
and hence the limiting free energy density $-p(\beta)/\beta$ is nonanalytic
at $\beta_*>0$.

Recall that $r=\lfloor2s\rfloor$ and
$\ell=\lfloor(k-1)/3\rfloor$.
For $s\geq1$, $k\geq4$, and $r\ell\geq m_0$,
\begin{align}
    \frac r2&\leq s, & 3\ell+1&\leq k, &
    \frac{sk}{6}&\leq r\ell\leq\frac{2sk}{3}.
\end{align}
Choosing between this construction and Curie-Weiss gives a transition
at $\Theta(1/\max\{s,\sqrt{sk}\})$.
In the remaining cases with $s\geq1$ and $k\geq2$,
$k/s\leq\max\{3,6m_0\}$, so the Curie-Weiss transition at $1/s$
already has the required order.
\end{proof}

\begin{lemma}\label{lem:support}
There is an absolute constant $0<c<1/4$ such that, for
$t=\beta\sqrt m/4>0$, the Hamiltonian $H_0$ of \eqref{eq:h0} satisfies
\begin{align}\label{eq:expander-probability}
\frac{\Tr\!\left[
 \mathbf1_{\{|\sum_v\varepsilon_vR_v|<Mt^2/4\}}e^{-\beta H_0}
 \right]}{\Tr e^{-\beta H_0}}
&\leq\left[
 \frac{e^{(1/4-c)t^2}}
      {\sqrt{2c}\,\cosh(t/\sqrt m)^{m/2}}
 \right]^M.
\end{align}
In particular, there are absolute constants $t_0,m_0>0$ such that this
probability is at most $e^{-M}$ whenever $t=t_0$ and $m\geq m_0$.
\end{lemma}
\begin{proof}
Write $\tr=2^{-n}\Tr$ and $T_v=m^{-1/2}\sum_{a,j}P_{v,a,j}$, so
$T_v^2=R_v$ and $-\beta H_0=t\sum_vT_v$.
Let $\mathcal P$ consist of the oriented nonbacktracking four-vertex
paths, which are induced because the girth is at least six. Set
\begin{align}\label{eq:ddef}
D(x_1,x_2,x_3,x_4)
&=1+\frac12\sum_i x_i^2-
 \sqrt{\sum_i x_i^2+2\sqrt{x_1^2x_3^2+x_1^2x_4^2+x_2^2x_4^2}},
\\
\mathcal D(x)&=\frac1{48}\sum_{p\in\mathcal P}D(x_{p_1},\ldots,x_{p_4}).
\end{align}
The Pauli relations give
\begin{align}\label{eq:expander-path-square}
\left[\left(\sum_{i=1}^4T_{p_i}\right)^2-\sum_{i=1}^4R_{p_i}\right]^2
&=4(R_{p_1}R_{p_3}+R_{p_1}R_{p_4}+R_{p_2}R_{p_4}).
\end{align}
Each vertex occurs in $48$ paths, so the resulting bound on
$\sum_iT_{p_i}$, applied on each joint eigenspace of the $R_v$, gives
\begin{align}\label{eq:expander-deficit}
-\beta H_0 = =\frac{t}{48\sqrt m}
        \sum_{p\in\mathcal P}\sum_{i=1}^4\sum_{a,j}P_{p_i,a,j}
&\leq \frac{Mt^2}{4}I+\frac12\sum_vR_v
       -t^2\mathcal D((\sqrt{R_v}/t)_v).
\end{align}
We will prove that, for an absolute $0<c<1/4$ and all $x\in[0,\infty)^V$,
\begin{align}\label{eq:expander-gap}
\left|\sum_v\varepsilon_vx_v^2\right|\leq M/4
\quad\Longrightarrow\quad
\mathcal D(x)&\geq c\left(M+\sum_vx_v^2\right).
\end{align}
By \eqref{eq:expander-trace-words} and $\cosh u\leq e^{u^2/2}$,
for $g\sim\mathcal N(0,1)$ we have
\begin{align}\label{eq:support-quadratic-moment}
\tr e^{(1/2-c)\sum_vR_v}
&=\left[\mathbb E_g
 \cosh\!\left(\sqrt{(1-2c)/m}\,g\right)^m\right]^M \leq\left[\mathbb E_g e^{(1/2-c)g^2}\right]^M=(2c)^{-M/2}.
\end{align}
Conjugation by $\prod_{v\in B,a,j}X_{y_{v,a,j}}$ fixes $T_u$ for
$u\in A$ and negates $T_v$ for $v\in B$. Convexity of $\tr e^X$ gives
\begin{align}\label{eq:denom}
\tr e^{-\beta H_0}
&\geq\tr e^{t\sum_{u\in A}T_u}
 =\cosh(t/\sqrt m)^{mM/2}.
\end{align}
Since the right-hand side of \eqref{eq:expander-deficit} commutes with
$H_0$, \eqref{eq:expander-gap} gives
\begin{align}
\tr\!\left[
 \mathbf1_{\{|\sum_v\varepsilon_vR_v|<Mt^2/4\}}e^{-\beta H_0}
 \right]
&\leq e^{M(1/4-c)t^2}\tr e^{(1/2-c)\sum_vR_v}
 \leq e^{M(1/4-c)t^2}(2c)^{-M/2}.
\end{align}
Dividing by \eqref{eq:denom} proves \eqref{eq:expander-probability}.
Choose $t_0$ sufficiently large that
\begin{align}
\lim_{m\to\infty}
 \frac{e^{(1/4-c)t_0^2}}
      {\sqrt{2c}\,\cosh(t_0/\sqrt m)^{m/2}}
&=\frac{e^{-ct_0^2}}{\sqrt{2c}}<e^{-1}.
\end{align}
For all sufficiently large $m\geq m_0$, the bracket in
\eqref{eq:expander-probability} is therefore at most $e^{-1}$.

It remains to prove \eqref{eq:expander-gap}. For $a,b\geq0$, put
$d(a,b)=\min\{(a-1)^2+b^2,a^2+(b-1)^2\}$. The following scalar bounds
will be used:
\begin{align}\label{eq:support-scalar-bounds}
d(a,b)&\leq32\left[(\sqrt{a^2+b^2}-1)^2+\min\{a^2b^2,1\}\right],\\
D(x)&\geq\frac12\left[(\sqrt{x_1^2+x_2^2}-1)^2
 +(\sqrt{x_3^2+x_4^2}-1)^2\right],\\
D(x)&\geq\frac1{32}\min\{x_2^2x_3^2,1\}.
\end{align}
For the first, the case $u=\sqrt{a^2+b^2}=0$ is immediate; otherwise
\begin{align}
d(a,b)&=(u-1)^2+\frac{2\min\{a^2,b^2\}}{u+\max\{a,b\}},\\
u\notin[1/2,2]
&\quad\Longrightarrow\quad d(a,b)\leq u^2+1\leq5(u-1)^2,\\
u\in[1/2,2],\ a^2b^2<1
&\quad\Longrightarrow\quad
 \frac{2\min\{a^2,b^2\}}{u+\max\{a,b\}}
 \leq4\min\{a^2,b^2\}\leq32a^2b^2,\\
u\in[1/2,2],\ a^2b^2\geq1
&\quad\Longrightarrow\quad d(a,b)\leq5.
\end{align} The second follows from
\begin{align}
x_1^2x_3^2+x_1^2x_4^2+x_2^2x_4^2
&=(x_1^2+x_2^2)(x_3^2+x_4^2)-x_2^2x_3^2.
\end{align}
For the third, rationalizing the two square roots gives
\begin{align}
D(x)&\geq \frac{x_2^2x_3^2}{\sqrt2\,\|x\|_2^3}
 \qquad(x\ne0),
&D(x)&\geq\frac12(\|x\|_2-\sqrt2)^2.
\end{align}
Use the first bound when $\|x\|_2^2\leq8$ and the second otherwise.
Every unoriented edge occurs eight times in each of the three edge
positions. Averaging the last two bounds in
\eqref{eq:support-scalar-bounds}, then applying the first, gives
\begin{align}\label{eq:support-edge-distance}
\sum_{uv\in E}d(x_u,x_v)
&\leq12288\mathcal D(x).
\end{align}
Round $\eta_v=\mathbf1_{\{x_v\geq1/2\}}$ and set
$S=\{v:\eta_v=\mathbf1_{\{v\in A\}}\}$. Each edge summand dominates
$(x_u-\eta_u)^2+(x_v-\eta_v)^2$ and
$\frac14\mathbf1_{\{\eta_u=\eta_v\}}$. Thus cubic regularity and the
assumed edge expansion give
\begin{align}
\sum_v(x_v-\eta_v)^2
&\leq\frac13\sum_{uv\in E}d(x_u,x_v),
\\
\min\{|S|,|S^c|\}
&\leq\frac{|\partial S|}{h}
 \leq\frac4h\sum_{uv\in E}d(x_u,x_v),\\
\min_{\sigma\in\{-1,1\}}\sum_v
 \left(x_v-\frac{1+\sigma\varepsilon_v}{2}\right)^2
&\leq\left(\frac23+\frac8h\right)\sum_{uv\in E}d(x_u,x_v)
 \leq K\mathcal D(x),
\end{align}
where $K=12288(2/3+8/h)$. Let $d$ denote the minimum and choose a
minimizing $\sigma$.
On $|\sum_v\varepsilon_vx_v^2|\leq M/4$, Cauchy--Schwarz implies
\begin{align}
\frac M4
&\leq\left|\sigma\sum_v\varepsilon_vx_v^2-\frac M2\right|
 \leq\sqrt{2Md}+d\leq\frac M8+5d,\\
M&\leq40d, \qquad \sum_vx_v^2\leq M+2d\leq42d, \\
\mathcal D(x)&\geq\frac dK
 \geq\frac{M+\sum_vx_v^2}{82K}.
\end{align}
This proves \eqref{eq:expander-gap} with $c=(82K)^{-1}$, which is
absolute because the expansion lower bound $h$ is fixed.
\end{proof}

\begin{theorem}[Non-analyticity of pinned free energy density]\label{thm:pin-upper}
    There is a sequence of $(s,k)$-long-range Pauli Hamiltonians such that $\lim_{n\to\infty} \frac{1}{n}\Tr_{[n]\setminus Y}(\bra{y} e^{-zH}\ket{y})$ is non-analytic at
    \begin{align}
        \beta_{\rm pin} = \Theta\lr{\frac{\log (k+1)}{sk}},
    \end{align}
    matching the lower bound of \Cref{lem:pin-zf} (up to a constant).
\end{theorem}
\begin{proof}
On a square torus $(V_L,E_L)$ put one unpinned qubit at each vertex and $k-1$ pinned qubits on each edge $e=uv$. With $X_e=\prod_{j=1}^qX_{e,j}$, define the $(s,k)$-long-range Pauli Hamiltonian
\begin{align}\label{eq:az-pinned-H}
    H_L=\sum_{e=uv}\left[\frac s2\sum_{j=1}^{k-1} Z_{e,j}+\frac s4X_eZ_u+\frac s4X_eZ_v\right].
\end{align}
We directly evaluate the pinned partition function  at $y=0^{(k-1)|E_L|}$ as
\begin{align}
    &\lim_{L\to\infty} \frac{1}{(2k-1)|V_L|} \log \Tr_{V_L} \bra{0^{(k-1)|E_L|}} e^{-\beta H_L} \ket{0^{(k-1)|E_L|}} \nonumber\\
    &= \frac{1}{2k-1}\lr{\log(f_0f_1) + \lim_{L\to\infty} \frac{1}{|V_L|} \log \sum_{z\in\{\pm1\}^{V_L}} \exp[\frac{1}{2} \log\lr{\frac{f_1}{f_0}} \sum_{\{u,v\}\in E_L} z_u z_v]}
\end{align}
for
\begin{align}
    f_0 = e^{-\beta s (k-1)/2}, \quad f_1 =\cosh \frac{\beta s \sqrt{(k-1)^2 + 1}}{2} - \frac{1}{\sqrt{1 + (k-1)^{-2}}} \sinh \frac{\beta s \sqrt{(k-1)^2 + 1}}{2}.
\end{align}
The second term is the partition function of the standard Ising model on a square lattice, and thus it inherits its phase diagram. In particular, observe that $\frac{1}{2} \log \frac{f_1}{f_0}$ crosses the Ising critical coupling $\frac{1}{2}\log\lr{1+\sqrt 2}$ \cite{kramers1941statistics, onsager1944crystal} exactly once, causing a non-analyticity at $\beta_{\rm pin} = \Theta(\log(k+1)/(sk))$.
\end{proof}

\section{Death of entanglement of nonlocal Hamiltonians}
\label{sec:nonlocal}
Throughout this appendix, we consider the nonlocal Hamiltonian family of \Cref{def:exp-nonlocal}, i.e.~$H = \sum_a \lambda_a P_a$ satisfying $\sup_{x\in \Lambda} \sum_{a:x\in\supp(P_a)} |\lambda_a| e^{\gamma |a|} \leq s$ where $|a| = |\text{supp}(P_a)|$. We will show that this class of models is separable at a constant temperature.  

We will use notation $\operatorname{supp}(X)$ for a Hermitian monomial and $\operatorname{supp}(E)$ for a Pauli coming from the sampling primitive, where the support refers to the true support of the operator (as opposed to the union of the supports of the constituent terms).

\subsection{Propagator sampling}
We begin with proving the following lemma for decomposing the propagator.  
\begin{lemma}[Propagator expansion]
\label{lem:sep-propagator-expansion-nonlocal}
    Consider a Hamiltonian $H = \sum_a \lambda_a P_a$ and let $y \in \Lambda$ be one of the sites in the lattice.  Assuming
    \begin{align}
        \sup_{x\in \Lambda} \sum_{a\ni x} |\lambda_{a}| e^{\gamma |a|} \leq s < \infty
    \end{align}
    we can expand the propagator in terms of Pauli strings $E_{\vec b}$
    \begin{align}
        e^{-\beta H} e^{\beta (H- H_{(y)})} = \sum_{\tau} \sum_{\vec{b}: |R(\vec{b})| =  \tau} \tilde\mu_{\vec{b}} E_{\vec{b}}
    \end{align}
    such that for $0 \leq \lambda \leq \gamma - 2 s \beta$
    \begin{align}
        \sum_{\tau\geq 1}\sum_{\vec{b}:|R(\vec{b})|=\tau} |\tilde{\mu}_{\vec{b}}| e^{\lambda|R(\vec{b})|}\leq L(\beta)
    \end{align}
    and hence for $\tau \geq 1$
    \begin{align}
        \sum_{\vec{b}: |R(\vec{b})| = \tau} |\tilde\mu_{\vec{b}}| e^{\lambda \tau} \leq L(\beta)
    \end{align}
    with $R(\vec{b}) = b_1 \cup \cdots \cup b_t$ for $\vec{b} = (b_1, \cdots, b_t)$ and $L(\beta) = e^{s \beta} - 1$.  
\end{lemma} 
\begin{proof}
We start with
\begin{align}
    e^{-\beta H} e^{\beta (H- H_{(y)})} = \sum_{t=0}^\infty \frac{\beta^t}{t!} f_t(H, H_{(y)}) \quad \text{for} \quad f_t(H, H_{(y)}) = \sum_{\vec{b} \in Q^{(t)}} \mu_{\vec{b}} E_{\vec{b}}.
\end{align}
We will focus on quantities of the form
\begin{align}
    B(\beta, \lambda) = \sum_{t \geq 0}\frac{\beta^t}{t!} b_t(\lambda) \quad \text{for} \quad b_t(\lambda) = \sum_{\vec{b}\in Q_{y}^{(t)}}|\mu_{\vec{b}}|e^{\lambda |R(\vec{b})|}.
\end{align}
Above, $b_t(\lambda)$ should be interpreted as collecting the weighted sum of the terms that appear at the $t$th order of the expansion.  Let us now compute $b_{t+1}(\lambda)$.  We have that
\begin{align}
    b_{t+1}(\lambda) = \sum_{\vec{b}\in Q_{y}^{(t+1)}}|\mu_{\vec{b}}|e^{\lambda |R(\vec{b})|} 
    &\leq \sum_{\vec{b} \in Q_{y}^{(t)}} \left[2\sum_{x\in R(\vec{b})}\sum_{a:x\in\supp(h_a)}|\mu_{\vec{b}}| |\lambda_a|e^{\lambda |R(\vec{b}) \cup A| } +\sum_{A\ni y}|\mu_{\vec{b}}| |\lambda_a|e^{\lambda |R(\vec{b}) \cup A| }\right] \nonumber \\
    &\leq \sum_{\vec{b} \in Q_{y}^{(t)}} |\mu_{\vec{b}}|e^{\lambda |R(\vec{b})|}\left[2|R(\vec{b})|\sup_z\sum_{A\ni z} |\lambda_a| e^{\lambda |a| } +\sup_z\sum_{A\ni z} |\lambda_a| e^{\lambda |a| }\right] \nonumber \\
    &\leq \left(2s\partial_\lambda + s\right)\sum_{\vec{b} \in Q_{y}^{(t)}} |\mu_{\vec{b}}| e^{\lambda |R(\vec{b})|} \nonumber \\
    &=\left(2s\partial_\lambda + s\right) b_t(\lambda).
\end{align}
In order to bound $\sum_{a:x\in\supp(h_a)} |\lambda_a| e^{\lambda |a| }$ we have relied on $\lambda < \gamma$.  Next we derive the full recurrence for $B(\beta, \lambda) = \sum_{\tau=0}^\infty e^{\lambda \tau} A_\tau (\beta)$. Using
\begin{align}
    \partial_\beta B(\beta, \lambda) &= \sum_{t=0}^\infty \frac{t\beta^{t-1}}{t!}b_t(\lambda) = \sum_{t=0}^\infty \frac{\beta^{t}}{t!}b_{t+1}(\lambda) \leq \sum_{t=0}^\infty \frac{\beta^{t}}{t!}\left(2s\partial_\lambda + s\right) b_t(\lambda) = \left(2s\partial_\lambda + s\right) B(\beta, \lambda)
\end{align}
we obtain the recurrence relation
\begin{align}
    \partial_\beta B(\beta, \lambda) - 2s \partial_\lambda B(\beta, \lambda) \leq s B(\beta, \lambda).
\end{align}
Next we show $B(\beta, \lambda) \leq e^{s\beta}$ which will imply the desired bound with $L(\beta) = e^{s \beta} -1$.  In order to derive this expression we will make use of the fact that we know the value at the initial point $B(0, \cdot)$.  Hence if we can control the behavior of the function along some path to $B(\beta, \lambda)$ we can infer its value at the endpoint.  To make the notation clear we will switch the inequality to
\begin{align}
\label{eq:inequality_eta}
    \partial_u B(u, \eta) - 2s \partial_\eta B(u, \eta) \leq sB(u, \eta).
\end{align}
The goal will be to show that for some range of $\lambda$ that $B(\beta, \lambda) \leq e^{s \beta}$. To do this we parametrize a path which varies from $(u, \eta) = (0, \lambda + 2s \beta)$ to $(u, \eta) = (\beta, \lambda)$.  We choose the specific parametrization where $B$ follows the path $B(u, \lambda + 2 s \beta - 2 s u)$ from $u \in [0, \beta]$.  Evaluating the derivative with respect to $u$ we find
\begin{align}
    \frac{d}{du} B(u, \lambda + 2 s \beta - 2 s u) = \partial_u B(u, \lambda + 2 s \beta - 2 s u) -2s \partial_\eta B(u, \lambda + 2 s \beta - 2 s u) \leq s B(u, \lambda + 2 s \beta - 2 s u)
\end{align}
where we have used \cref{eq:inequality_eta}.  This inequality only holds when 
\begin{align}
    \lambda + 2 s \beta - 2 s u \leq \gamma, \quad \lambda + 2 s \beta - 2 s u\geq 0.
\end{align}
Since $u \in [0, \beta]$ this provides the constraint $0\leq \lambda  \leq \gamma-2s\beta$.  Evaluating the following inequality we find that
\begin{align}
    \frac{d}{d u}\left(e^{-su} B(u, \lambda + 2 s \beta - 2 s u)\right) = e^{-s u}\left(\frac{d}{du}B(u, \lambda + 2 s \beta - 2 s u)-s B(u, \lambda + 2 s \beta - 2 s u)\right)\leq 0,
\end{align}
where we have used that $\frac{d}{du} B(u, \lambda + 2 s \beta - 2 s u) \leq s B(u, \lambda + 2 s \beta - 2 s u)$.  Over the valid range of $u$ values this implies that the function $e^{-su} B(u, \lambda + 2 s \beta - 2 s u)$ is  non-increasing.  Hence 
\begin{align}
    e^{-s u} B(u, \lambda + 2 s \beta - 2 s u) \leq B(0, \lambda + 2 s \beta) = 1
\end{align}
where for the last equality we have used that $B(0, \cdot)$ is the propagator with $u\propto\beta=0$ which only has $A_{\tau}(0) \neq 0$ for $\tau=0$ so $B(0, \cdot) = 1$.  At this point we have
\begin{align}
    B(u, \lambda + 2 s \beta - 2 s u) \leq e^{s u}.
\end{align}
Evaluating at $u=\beta$ we find that $B(\beta, \lambda) \leq e^{s \beta}$, and we recall that this is only valid for $0 \leq \lambda  \leq \gamma - 2 s \beta$. Since $B(\beta, \lambda) \leq e^{s \beta}$, we have 
\begin{align}
    \sum_{\tau \geq 1}e^{\lambda \tau} A_{\tau}(\beta) \leq e^{s \beta}-1.
\end{align}
For $\tau \geq 1$ we get 
\begin{align}
    A_\tau(\beta) \leq (e^{s \beta}-1) e^{-\lambda \tau}
\end{align}
as desired.
\end{proof}
\noindent Now we use this lemma to define a simple guarantee on sampling coefficients from the propagator.
\begin{lemma}[Propagator sampling]
\label{lem:sep-propagator-sampling-nonlocal}
    Under the conditions of \Cref{lem:sep-propagator-expansion-nonlocal} there exists a distribution over tuples $(b,E)$  such that
    \begin{align}
        e^{-\beta H} e^{\beta (H - H_{(y)})} = \sum_{i} p_i (I + b_i E_i) = \mathbb{E}_i [I + b_i E_i]
    \end{align}
    where $|b_i| \leq L(\beta)  e^{-\lambda |R(\vec{b})|}\leq L(\beta) e^{-\lambda |\operatorname{supp}(E_i)|}$ under the condition that $0\leq \lambda \leq \gamma - 2 s \beta$.  In this expression $L(\beta) = e^{s\beta}-1$ and $R(\vec{b}) = b_1 \cup \cdots \cup b_t$ is the union of the supports of the terms making up $E$.    
\end{lemma}
\begin{proof}
    Using the shorthand $Y_{\vec{b}} = \frac{\beta^{|\vec{b}|}}{|\vec{b}|!}\mu_{\vec{b}} E_{\vec{b}}$ where $|\vec{b}|$ is the number of terms in $\vec{b}$ we get
    \begin{align}
        I + \sum_{\vec{b}\neq \emptyset} Y_{\vec{b}}  =  \sum_{\vec{b}\neq \emptyset} \left(\frac{\Vert Y_{\vec{b}} \Vert e^{\lambda |R(\vec{b})|}}{\sum_{\vec{b}\neq \emptyset}\Vert Y_{\vec{b}} \Vert e^{\lambda |R(\vec{b})|}}\right)\left(I + e^{-\lambda |R(\vec{b})|}\sum_{\vec{b}\neq \emptyset}\Vert Y_{\vec{b}} \Vert e^{\lambda |R(\vec{b})|}\frac{Y_{\vec{b}}}{\Vert Y_{\vec{b}} \Vert}\right).
    \end{align}
    If we sample with probability $p_i = \frac{\Vert Y_{\vec{b}} \Vert e^{\lambda |R(\vec{b})|}}{\sum_{\vec{b}\neq \emptyset}\Vert Y_{\vec{b}} \Vert e^{\lambda |R(\vec{b})|}}$ the outcome with Pauli $\frac{Y_{\vec{b}}}{\Vert Y_{\vec{b}} \Vert}$ and coefficient $b_i = e^{-\lambda |R(\vec{b})|}\sum_{\vec{b}\neq \emptyset}\Vert Y_{\vec{b}} \Vert e^{\lambda |R(\vec{b})|}$ the expression matches the form of an expectation $\mathbb{E}_i [I + b_i E_i]$.  Since we have that $\sum_{\vec{b}\neq \emptyset}\Vert Y_{\vec{b}} \Vert e^{\lambda |R(\vec{b})|} \leq L(\beta)$ then $|b_i| \leq L(\beta) e^{-\lambda |R(\vec{b})|} \leq L(\beta) e^{-\lambda |\operatorname{supp}(E_i)|}$.  
\end{proof}

\subsection{Pinning procedure}
We give \Cref{alg:sep-iterative-pinning-nonlocal} we use for pinning in the nonlocal setting with the main difference from \Cref{alg:sep-iterative-pinning} being that we pin site by site instead of term by term.  
\begin{algorithm}[t]
\caption{\textsc{Iterative Pinning}}
\label{alg:sep-iterative-pinning-nonlocal}
\begin{algorithmic}[1]
\REQUIRE Hamiltonian $H = \sum_a \lambda_a P_a$ and sampling primitive satisfying $\mathbb{E}\left[I + b E\right] =  e^{-\eta H^{(S)}} e^{\eta \left(H^{(S)} - H^{(S)}_{(x^*)}\right)}$ with $|b| \leq L(\eta) q^{|\operatorname{supp}(E)|}$.
\ENSURE A configuration $\chi = \{(c_1, X_1), \cdots, (c_j, X_j)\}$ such that $\mathbb{E}[\sigma(\chi)]= e^{-\beta H}$.
\STATE $S = \Lambda$, $\chi = \emptyset$, $l=0$.
\WHILE{$\mathcal{A}^{(S)} \neq \emptyset$}
    \IF{$l \geq 1$ and $\exists \hat{x} \in S \cap \operatorname{supp}(X_l)$}
        \STATE $x^* \gets \hat{x}$ and set $\hat{l} \gets l$
    \ELSE
        \STATE $c_{l+1} \gets 0$, $X_{l+1} \gets I$, add $(c_{l+1}, X_{l+1})$ to $\chi$, $\hat{l} \gets l+1$, choose $a \in \mathcal{A}^{(S)}$ and set $x^*$ to any site in $a$.
    \ENDIF
    \STATE Sample $b_1$, $E_1$ with $\eta \rightarrow \beta/2$, $H \rightarrow H^{(S)}$, and selected site $x^*$ so that $\mathbb{E}[I + b_1 E_1] = e^{-\frac{\beta}{2}H^{(S)}} e^{\frac{\beta}{2}(H^{(S)} - H^{(S)}_{(x^*)})}$.
    \STATE Sample $b_2$, $E_2$ with $\eta \rightarrow \beta/2$, $H \rightarrow H^{(S)}$, and selected site $x^*$ so that $\mathbb{E}[I + b_2 E_2] = e^{-\frac{\beta}{2}H^{(S)}} e^{\frac{\beta}{2}(H^{(S)} - H^{(S)}_{(x^*)})}$.
    \STATE Sample $J\in \{1,2,3,4,5,6,7\}$ with probabilities $p_1 = q$, $p_{\neq 1} = (1-q)/6$.
    \IF{$J=1$}
        \STATE $\hat{c} \gets (1/p_1) c_{\hat{l}}$,  $\hat{X} \gets X_{\hat{l}}$.
    \ELSIF{$J=2$}
        \STATE $\hat{c} \gets (1/p_2)b_1$, $\hat{X} \gets (E_1+E_1^\dagger)/2$.
    \ELSIF{$J=3$}
        \STATE $\hat{c} \gets (1/p_3)b_2$, $\hat{X} \gets (E_2+E_2^\dagger)/2$.
    \ELSIF{$J=4$}
        \STATE $\hat{c} \gets (1/p_4)b_1 c_{\hat{l}}$, $\hat{X} \gets (E_1^\dagger X_{\hat{l}} +X_{\hat{l}} E_1)/2$.
    \ELSIF{$J=5$}
        \STATE $\hat{c} \gets (1/p_5)b_2 c_{\hat{l}}$, $\hat{X} \gets (E_2^\dagger X_{\hat{l}} +X_{\hat{l}} E_2)/2$.
    \ELSIF{$J=6$}
        \STATE $\hat{c} \gets (1/p_6)b_1 b_2$, $\hat{X} \gets (E_1^\dagger E_2 +E_2^\dagger E_1)/2$.
    \ELSE
        \STATE $\hat{c} \gets (1/p_7)b_1 b_2 c_{\hat{l}}$, $\hat{X} \gets (E_2^\dagger X_{\hat{l}} E_1 +E_1^\dagger X_{\hat{l}} E_2)/2$.
    \ENDIF
    \STATE Set $c_{\hat{l}} \gets \hat{c}$, $X_{\hat{l}} \gets \hat{X}$; if $\hat{X}=0$, set $(c_{\hat{l}}, X_{\hat{l}})$ to $(0,I)$.
    \STATE $l \gets \hat{l}$, $S \gets S \setminus \{x^*\}$
\ENDWHILE
\RETURN $\chi$.
\end{algorithmic}
\end{algorithm}
We now prove that this algorithm works as expected.  First, we prove the following lemma.
\begin{lemma}[Pinning validity]
\label{lem:sep-pinning-validity-nonlocal}
    Every iteration of the algorithm completes and produces a valid configuration if the input is a valid configuration.  Here a valid configuration is one where the $X_i$ making up the configuration are disjoint and each $X_i$ is indeed a Hermitian monomial.  Furthermore if for $i\in [l-1]$ we have $S \cap \operatorname{supp}(X_i) = \emptyset$ then after the loop iteration for $i\in [\hat{l}-1]$ we have $\hat{S} \cap \operatorname{supp}(X_i) = \emptyset$.  Note that if an update gives $\hat{X}=0$ we discard the corresponding factor from the configuration.
\end{lemma}
\begin{proof}
    At every iteration of the algorithm an $x^* \in S$ is chosen.  In either case $x^* \in S$ so the size of $S$ is strictly decreased and since $S$ is initialized with a finite size the algorithm will eventually terminate.  

    Assume that at the beginning of an iteration of the loop the monomials $X_i$ satisfy that $S \cap \operatorname{supp}(X_i) = \emptyset$ for all $i < l$.  During the loop iteration the first case is that there $\exists \hat{x} \in S \cap \operatorname{supp}(X_l)$ in which case the final monomial is updated but only using terms which have support fully in $S$, while $S \cap \operatorname{supp}(X_i)=\emptyset$ for $i < l$ means the earlier monomials have no terms intersecting $S$.  Alternatively $\nexists \hat{x} \in S \cap \operatorname{supp}(X_l)$ and a new monomial is made with $\hat{l} = l+1$ but the previous last monomial $\hat{S} \cap \operatorname{supp}(X_l)=\emptyset$ since the new monomial was only made because there were no $\hat{x}$ touching $S$.  
    
    Since the initial configuration in one loop iteration is valid and the loop only modifies the last monomial of the configuration then the new configuration is also valid because the invariant guarantees the final monomial is disjoint from the earlier monomials.  Each $X_i$ is by definition a Hermitian monomial since it is computed by starting with a Hermitian monomial and multiplying by the $E_i$ which are products of terms in the Hamiltonian.
\end{proof}
\noindent Next we show that throughout the procedure the expectation of the current monomials is consistent with the Gibbs state
\begin{lemma}[Gibbs state validity]
\label{lem:sep-gibbs-validity-nonlocal}
    The procedure produces a $\hat{S}$ and new configuration $\hat{\chi}$ such that 
    \begin{align}
        \mathbb{E}\left[e^{-\frac{\beta}{2} H^{(\hat S)}} \sigma(\hat\chi) e^{-\frac{\beta}{2} H^{(\hat S)}}\right] = e^{-\frac{\beta}{2} H^{(S)}} \sigma(\chi) e^{-\frac{\beta}{2} H^{(S)}}
    \end{align}
\end{lemma}
\begin{proof}
    Identical to \Cref{lem:sep-gibbs-validity} except with a pinned site instead of a whole term.  The invariant $S \cap \operatorname{supp}(X_j) = \emptyset$ for $j < \hat{l}$ can be invoked to conclude that every earlier $X_j$ commutes with $H^{(S)}$.  
\end{proof}
\noindent Next we show that the coefficient $|c|$ for each monomial stays less than 1.
\begin{lemma}[Coefficient control]
\label{lem:sep-constant-controlled-nonlocal}
    If the sampling primitive satisfies $|b_i| \leq L(\beta/2) q^{|\operatorname{supp}(E_i)|}$ and
    \begin{align}
        0 < q < 1, \quad \frac{6 L(\beta/2)}{1-q}\leq 1, \quad L(\beta/2) \leq 1
    \end{align}
    then at each step of the algorithm each monomial satisfies 
    \begin{align}
        |c_i| \leq q^{|S \cap \operatorname{supp}(X_i)|}.
    \end{align}
    Furthermore at the end of the algorithm $|c_i| \leq 1$ for all monomials.
\end{lemma}
\begin{proof}
We consider the following invariant of a given monomial
\begin{align}
    |c| \leq q^{|S \cap \operatorname{supp}(X)|}
\end{align}
where $X$ is the current monomial.  Let $c$, $X$ correspond to the monomial before one round of the loop and $\hat{c}, \hat{X}$ to after the loop.  We would like to show that if the invariant is satisfied before the loop then it will also be satisfied after.  The sampling primitive guarantee gives 
\begin{align}
    |b_1| \leq L(\beta/2) q^{|\operatorname{supp}(E_1)|} , \quad |b_2| \leq L(\beta/2) q^{|\operatorname{supp}(E_2)|}.
\end{align}
Now for each of the possible branching cases we show that the invariant holds after.  
\begin{enumerate}
    \item $p_1$:  If $c = 0$ then $\hat{c}=0$ and the claim holds.  Otherwise we have  
    \begin{align}
        |\hat{c}| = \frac{1}{q} |c| \leq  \frac{1}{q} q^{|S \cap \operatorname{supp}(X)|} \leq q^{|S \cap \operatorname{supp}(X)|-1} \leq q^{|\hat{S} \cap \operatorname{supp}(\hat{X})|}
    \end{align}
    where we have used that $|\hat{S} \cap \operatorname{supp}(\hat{X})|$ must have decreased relative to $|S \cap \operatorname{supp}(X)|$ since $\hat{X} = X$ but at least one site in the overlap was pinned based on how $x^*$ was chosen. 
    \item $p_2$:  Here we sample a new term and we make use of the guarantees from \Cref{lem:sep-propagator-sampling-nonlocal} that the coefficients are exponentially small in the length of the monomial.  We show the bound is met with the following operations 
    \begin{align}
        |\hat{c}| 
        & = \frac{6}{1-q} |b_1| \leq  \frac{6L}{1-q} q^{|\operatorname{supp}(E_1)|} \leq q^{|\operatorname{supp}(E_1)|} \leq q^{|\hat{S} \cap \operatorname{supp}(\hat{X})|}.
    \end{align}
    Here we have used that $\frac{6L}{1-q} \leq 1$.  
    \item $p_3$: Similar to $p_2$ 
    \item $p_4$: This follows with essentially the same manipulations as the $p_2$ case
    \begin{align}
        |\hat{c}|  
        = \frac{6}{1-q} |b_1| |c| 
        \leq  \frac{6L}{1-q} q^{|\operatorname{supp}(E_1)|} q^{|S \cap \operatorname{supp}(X)|} 
        \leq q^{|\operatorname{supp}(E_1)| + |S \cap \operatorname{supp}(X)|}  
        \leq q^{|\hat{S} \cap \operatorname{supp}(\hat{X})|}  
    \end{align}
    where in the last step we have used that $|\hat{S} \cap \operatorname{supp}(\hat{X})| \leq |S \cap \operatorname{supp}(X)|+|\operatorname{supp}(E_1)|$ since $|\operatorname{supp}(\hat{X})|$ may have at most grown by $|\operatorname{supp}(E_1)|$ and $\hat{S}$ can only have decreased relative to $S$.  
    \item $p_5$: Similar to $p_4$
    \item $p_6$: Similar to $p_2$
    \item $p_7$: Similar to $p_4$
\end{enumerate} 
Hence in all cases the invariant is preserved.  Since $q \leq 1$ at the end of the algorithm $|c|\leq 1$ for all monomials.
\end{proof}
\noindent Now we prove a sufficient condition for the algorithm to yield separability.
\begin{lemma}[Pinning procedure]
\label{lem:sep-pinning-nonlocal}
    Consider a sampling procedure, which takes a set of sites $S$ and a site $x^* \in \Lambda$ such that $x^* \in S$, and returns coefficients $b$ and Paulis $E$ such that
    \begin{align}
        \mathbb{E}\left[I + b E\right] =  e^{-\eta H^{(S)}} e^{\eta \left(H^{(S)} - H^{(S)}_{(x^*)}\right)} .
    \end{align}
    If the coefficient $b$ is bounded as 
    \begin{align}
        |b| \leq L(\eta) q^{|\operatorname{supp}(E)|}
    \end{align}
    then the Gibbs state $e^{-\beta H}$ is separable if
    \begin{align}
        0 < q < 1, \quad \frac{6 L(\beta/2)}{1-q} \leq 1.
    \end{align}
\end{lemma}
\begin{proof}
    By \Cref{lem:sep-pinning-validity-nonlocal} and \Cref{lem:sep-gibbs-validity-nonlocal} the algorithm eventually terminates with $H^{(S)}=0$ and a valid configuration $\chi$ satisfying 
    \begin{align}
        e^{-\beta H} = \mathbb{E}[\sigma(\chi)].
    \end{align}
    This means the final distribution over configurations is equivalent to the Gibbs state.  Since the configuration is valid each term consists of a Hermitian monomial.  By \Cref{lem:sep-hermitian-monomial-pauli} each of the Hermitian monomials is either 0 or a signed Pauli.  By \Cref{lem:sep-constant-controlled-nonlocal} the coefficient on each term has magnitude $\leq 1$ since the lemma provides the required constraints on $q$ and $L(\beta/2)$ (since $0<q<1$, $L(\beta/2) \leq 1$ is implied by $\frac{6 L(\beta/2)}{1-q} \leq 1$).  Hence by \Cref{lem:sep-pauli-factor} the state is separable.
\end{proof}
\noindent Lastly we tie all these results together to prove separability of $(s,\gamma)$-nonlocal Pauli Hamiltonians.
\begin{corollary}[Nonlocal separability thresholds]
\label{cor:sep-thresholds-nonlocal}
    The Gibbs state for an $(s,\gamma)$-nonlocal Pauli Hamiltonian at inverse temperature 
    \begin{align}
        \beta \leq \min\left(\frac{\gamma}{16s}, \frac{1}{16s}\right)
    \end{align}
    is separable.
\end{corollary}
\begin{proof}
    Using the condition $0\leq \lambda \leq \gamma - s \beta$ we can invoke \Cref{lem:sep-propagator-sampling-nonlocal}
    for a $(s,\gamma)$-nonlocal Hamiltonian to sample coefficients $b$ and Paulis $E$ such that
    \begin{align}
        \mathbb{E}\left[I + b E\right] =  e^{-\eta H^{(S)}} e^{\eta \left(H^{(S)} - H^{(S)}_{(x^*)}\right)} 
    \end{align}
    with the guarantee that 
    \begin{align}
        |b| \leq L(\beta/2) q^{|\operatorname{supp}(E)|}
    \end{align}
    with $q = e^{-\lambda}$ and $L(\beta/2) = e^{s\beta/2} -1$.
    Hence by \Cref{lem:sep-pinning-nonlocal} the Gibbs state is separable when 
    \begin{align}
        0 < q < 1, \quad \frac{6 L(\beta/2)}{1-q} \leq 1, \quad 0\leq \lambda \leq \gamma - s \beta.
    \end{align}
    We choose $\lambda = \gamma - s\beta$ to maximize $1-q = 1-e^{-(\gamma-s\beta)}$, giving
    \begin{align}
        \frac{6 L(\beta/2)}{1-q} = \frac{6 (e^{s \beta/2} - 1)}{1-q} \leq 1 \implies 6 (e^{s \beta/2}-1) \leq 1-q.
    \end{align}
    We now use $q = e^{- (\gamma - s \beta)}$, where we use the largest allowable value for $\lambda$.  
    \begin{align}
        6 (e^{s \beta/2}-1) \leq 1-e^{- (\gamma - s \beta)}.
    \end{align}
    If we take $s\beta \leq \frac{\gamma}{16}$ and $\gamma \leq 1$ then the inequality is satisfied as we can see by the following relation  
    \begin{align}
       6 (e^{s \beta/2}-1) \leq 6 (e^{\gamma/32}-1) \leq \frac{6\gamma}{16} \leq \frac{15\gamma}{32}\leq 1-e^{-\frac{30\gamma}{32}} = 1-e^{\gamma/16} e^{-\gamma} \leq 1-e^{- (\gamma - s \beta)}
    \end{align}
    where in the second inequality we used that $e^x -1 < 2 x$ for $x<1$ and later we used that $\frac{x}{2} \leq 1- e^{-x}$.  Hence for $\gamma \leq 1$ we have that $\beta \leq \frac{\gamma}{16s}$ satisfies the inequality.  For $\gamma \geq 1$ and $\beta s \leq \frac{1}{16}$
    \begin{align}
        6 (e^{s \beta/2}-1) \leq 6 s\beta \leq \frac{6}{16} < 1-e^{-15/16} \leq 1-e^{-(\gamma-s\beta)}.
    \end{align}
    Hence the simplified bound is 
    \begin{align}
        \beta \leq \min\left(\frac{\gamma}{16s}, \frac{1}{16s}\right).
    \end{align}

\end{proof}

\end{document}